\documentclass[11pt,a4paper]{article}

\usepackage[utf8]{inputenc}

\usepackage[utf8]{inputenc}
\usepackage[bookmarks]{hyperref}
\hypersetup{colorlinks=true,citecolor=blue,linkcolor=blue,filecolor=blue,urlcolor=blue,linktocpage,hyperfootnotes=false}
\usepackage{amsmath,amssymb,amsthm}

\usepackage{dsfont}
\usepackage{fullpage}

\usepackage{cleveref}

\usepackage{enumerate}
\usepackage{mathtools}

\usepackage{color}

\newtheorem{theorem}{Theorem}[section]
\newtheorem{lemma}[theorem]{Lemma}
\newtheorem{proposition}[theorem]{Proposition}
\newtheorem{corollary}[theorem]{Corollary}
\theoremstyle{definition}
\newtheorem{definition}[theorem]{Definition}
\newtheorem*{remark}{Remark}

\newcommand{\cH}{\mathcal{H}}
\newcommand{\cD}{\mathcal{D}}

\newcommand{\cP}{\mathcal{P}}

\newcommand{\cB}{\mathcal{B}}
\newcommand{\cO}{\mathcal{O}}
\newcommand{\cU}{\mathcal{U}}
\newcommand{\cC}{\mathcal{C}}
\newcommand{\cE}{\mathcal{E}}
\newcommand{\cV}{\mathcal{V}}
\newcommand{\cX}{\mathcal{X}}

\newcommand{\sk}{\mathfrak{s}}
\newcommand{\kE}{\mathfrak{E}}
\newcommand{\kD}{\mathfrak{D}}

\newcommand{\cHt}{\widetilde{\cH}}
\newcommand{\cDs}{\cD_\leq}

\newcommand{\uD}{\underline{D}}

\newcommand{\trho}{\tilde{\rho}}

\newcommand{\tk}{\tilde{k}}
\newcommand{\tq}{\widetilde{Q}}

\newcommand{\brho}{\bar{\rho}}
\newcommand{\hrho}{\hat{\rho}}
\newcommand{\hsigma}{\hat{\sigma}}
\newcommand{\homega}{\hat{\omega}}
\newcommand{\hI}{\hat{I}}
\newcommand{\hLambda}{\hat{\Lambda}}
\newcommand{\hpsi}{\hat{\psi}}
\newcommand{\one}{\mathds{1}}
\newcommand{\eps}{\varepsilon}
\newcommand{\X}{\rangle\langle}
\newcommand{\PP}{\mathbb{P}}
\DeclareMathOperator{\tr}{Tr}
\DeclareMathOperator{\id}{id}
\DeclareMathOperator{\im}{Im}
\DeclareMathOperator{\dom}{Dom}
\DeclareMathOperator{\Be}{\mathcal{B}^\varepsilon}

\newcommand{\Dmine}{D^\varepsilon_\text{\normalfont min}}
\newcommand{\Dmax}{D_\text{\normalfont max}}
\newcommand{\Dmaxe}{D^\varepsilon_\text{\normalfont max}}

\newcommand{\uds}{\underline{D}_s^\eps}
\newcommand{\ods}{\overline{D}_s^\eps}
\newcommand{\dsc}{D_s^\eps}
\newcommand{\Du}{\underline{D}}
\newcommand{\oD}{\overline{D}}
\newcommand{\uH}{\underline{H}}
\newcommand{\oH}{\overline{H}}
\newcommand{\uI}{\underline{I}}
\newcommand{\oI}{\overline{I}}
\newcommand{\oQ}{{\overline{Q}}}
\newcommand{\uhs}{\underline{H}_s^\eps}
\newcommand{\ohs}{\overline{H}_s^\eps}
\newcommand{\uis}{\underline{I}_s^\eps}
\newcommand{\ois}{\overline{I}_s^\eps}

\newcommand{\sumi}{\sum\nolimits}
\DeclareMathOperator{\supp}{supp}
\newcommand{\oW}{{\overline{W}}}
\newcommand{\invP}[1]{\Phi^{-1}\left(#1\right)}
\newcommand{\uS}{\underline{S}}
\newcommand{\oS}{\overline{S}}
\newcommand{\hcH}{\widehat{\cH}}
\newcommand{\nin}{n\in\mathbb{N}}

\newcommand{\cF}{\mathcal{F}}

\newcommand{\ba}[1]{\begin{array}{#1}}
\newcommand{\ea}{\end{array}}

\newcommand{\be}{\begin{equation}}
\newcommand{\ee}{\end{equation}}
\newcommand{\bea}{\begin{eqnarray}}
\newcommand{\eea}{\end{eqnarray}}
\newcommand{\beann}{\begin{eqnarray*}}
\newcommand{\eeann}{\end{eqnarray*}}

\def\reff#1{(\ref{#1})}
\DeclareMathOperator*{\argmin}{arg\,min}

\usepackage[backend=bibtex,sorting=nty,sortcites=false,firstinits=true,doi=false,isbn=false,url=false,maxbibnames=5]{biblatex}
\renewbibmacro{in:}{}
\makeindex
\begin{document}
\title{\textbf{Second-order asymptotics for source coding, dense coding and pure-state entanglement conversions}}
\author{Nilanjana Datta and Felix Leditzky\\[0.25cm]
\textit{\small Statistical Laboratory, Centre for Mathematical Sciences, University of Cambridge,}\\
\textit{\small Wilberforce Road, Cambridge CB3 0WA, United Kingdom} }
\maketitle

\begin{abstract}
We introduce two variants of the information spectrum relative entropy defined by Tomamichel and Hayashi~\cite{TH13} which have the particular advantage of satisfying the data-processing inequality, i.e.~monotonicity under quantum operations. 
This property allows us to obtain one-shot bounds for various information-processing tasksin terms of these quantities. 
Moreover, these relative entropies have a second order asymptotic expansion, which in turn yields tight second order asymptotics for optimal rates of these tasks in the i.i.d.~setting.
The tasks studied in this paper are fixed-length quantum source coding, noisy dense coding, entanglement concentration, pure-state entanglement dilution, and transmission of information through a classical-quantum channel. 
In the latter case, we retrieve the second order asymptotics obtained by Tomamichel and Tan \cite{TT13}. 
Our results also yield the known second order asymptotics of fixed-length classical source coding derived by Hayashi \cite{Hay08}. 
The second order asymptotics of entanglement concentration and dilution provide a refinement of the {\em{inefficiency}} of these protocols - a quantity which, in the case of entanglement dilution, was studied by Harrow and Lo~\cite{HL04}. 
We prove how the discrepancy between the optimal rates of these two processes in the second order implies the irreversibility of entanglement concentration established by Kumagai and Hayashi \cite{KH13}.
In addition, the spectral divergence rates of the Information Spectrum Approach (ISA) can be retrieved from our relative entropies in the asymptotic limit.
This enables us to directly obtain the more general results of the ISA from our one-shot bounds. 
\end{abstract}

\section{Introduction}
Optimal rates of information-processing tasks such as storage and transmission of information, and manipulation of 
entanglement, are of fundamental importance in Information Theory. These rates were originally evaluated in the
so-called asymptotic, i.i.d.\footnote{Here, i.i.d.~is the standard acronym for ``independent and identically distributed''.}~setting, in which it is assumed that the underlying resources 
(sources, channels or entangled states) employed in the tasks are available for asymptotically many uses,
and that there are no correlations between their successive uses. The rates in this scenario are given in terms
of entropic quantities obtainable from the relative entropy. It is, however, unrealistic to assume the 
availability of infinitely many copies of the required resources. In practice, we have finite resources and hence
a fundamental problem of both theoretical and practical interest is to determine how quickly the behaviour of a finite system 
approaches that of its asymptotic limit. 

The first step in this direction, from the standpoint of Information Theory, is to determine the second order asymptotics
of optimal rates.\footnote{The precise meaning of the phrase ``second order asymptotics'' is elucidated in the following
paragraph.} Interest in this was initiated in the classical realm by Strassen~\cite{Str62}, who evaluated the second order asymptotics 
for hypothesis testing and channel coding.  In the last decade there has been a renewal of interest in the evaluation of second order asymptotics for other classical information theoretic tasks (see e.g.~\cite{Hay08,Hay09,KH13a} and references therein) and, more recently, even in third-order asymptotics~\cite{KV13a}. Moreover, the recent papers by Tomamichel and Hayashi~\cite{TH13} and Li~\cite{Li14} have introduced the study of second order asymptotics in Quantum Information Theory as well. The achievability parts of the second order asymptotics for the tasks studied in \cite{TH13,Li14} were later also obtained by Beigi and Gohari \cite{BG13} via the collision relative entropy.

Let us explain what exactly is meant by the phrase ``second order asymptotics''. Consider the familiar task of
fixed-length quantum source coding. Schumacher~\cite{Sch95} proved that for a memoryless, quantum information source characterized 
by a state $\rho$, the optimal rate of reliable data compression is given by its von Neumann entropy, $S(\rho)$. The criterion
of reliability is that the error
incurred in the compression-decompression scheme vanishes in the asymptotic limit, $n \to \infty$, where $n$ denotes the 
number of copies (or uses) of the source. Let us now consider instead a finite number ($n$) of copies of the source and require
that the error incurred in compressing its state $\rho^{\otimes n}$ is at most $\eps$ (for some $0<\eps <1$). Suppose $m^{n,\eps}(\rho)$
denotes the compression length, i.e.~the minimum of the logarithm of the dimension of the compressed Hilbert space in this case. This quantity can be expanded as follows:
\be\label{2nd}
m^{n,\eps}(\rho) = a n + b \sqrt{n} + \cO(\log n)
\ee
Here, the coefficient $a$ of the leading term constitutes the first order asymptotics of the compression length. As expected, it turns out to be the  the optimal rate in the asymptotic limit, i.e.~the von Neumann entropy of the source. The second order asymptotics is given by the coefficient $b$, and is a function of both the error threshold $\eps$ and the state $\rho$. Determining the second order asymptotics comprises the evaluation of the coefficient $b$.

\Cref{thm:sc-second-order}(i) of Section~\ref{sec:information-tasks} gives an explicit expression for the coefficient $b$ in the case of fixed-length (visible) quantum source coding. In Figure~\ref{fig:rate} we plot $(an + b\sqrt{n})/n$ against $n$ for a memoryless quantum information source which emits the (mutually non-orthogonal) pure states $|0\rangle$ and $|+\rangle$ with equal probability, for three different values of the error threshold $\eps$. This exhibits how the rate of data compression converges to its asymptotically optimal value.

\begin{figure}[ht]
\centering
\includegraphics[width=0.9\textwidth]{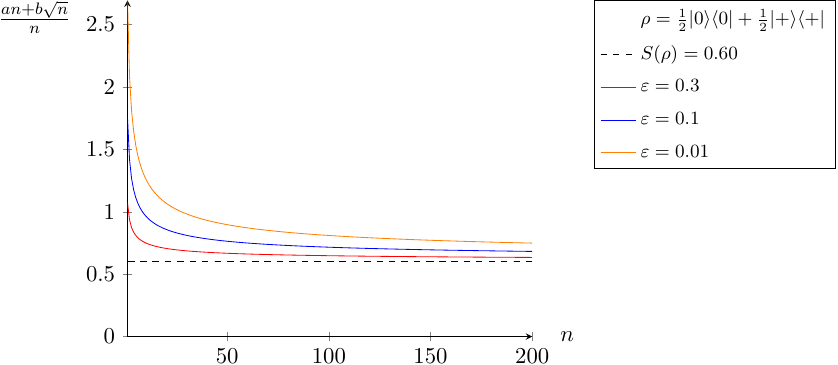}
\caption{Plot of $\frac{an+b\sqrt{n}}{n}$ for the source $\rho=\frac{1}{2}|0\rangle\langle 0|+\frac{1}{2}|+\rangle\langle +|$.}
\label{fig:rate}
\end{figure}

In this paper we study the second order asymptotics of various information-processing tasks: fixed-length quantum source coding, 
noisy dense coding, entanglement concentration, pure-state entanglement dilution, and transmission of information through a 
classical-quantum channel. For each task we consider the $n$-copy case, where $n$ denotes the number of copies
of the relevant resource (source, entangled state or channel) employed in the protocol, $\eps$ denotes the error threshold, and 
$m^{n,\eps}(\rho)$ denotes the characteristic quantity of the protocol (e.g.~the minimum compression length in the case of source coding). We arrive at an expression of the form \reff{2nd} for $m^{n,\eps}(\rho)$  -- which in this case is a quantity from which the optimal asymptotic rate $R$ 
of the protocol is obtained through the relation 
\begin{align*}
R = \lim_{\eps \to 0} \lim_{n \to \infty} \frac{1}{n} m^{n,\eps}(\rho).
\end{align*}
For each of the tasks studied, the coefficient 
$a$ turns out to be the entropic quantity characterizing $R$. Moreover, the coefficient $b$ is proportional to $\Phi^{-1}(\eps)$,
where $\Phi^{-1}$ denotes the inverse of the cumulative distribution function of the standard normal distribution,
and the constant of proportionality depends on the underlying resource (the latter is often called the {\em{dispersion}}
or {\em{information variance}}). This form of $b$ is a feature of all results on second order asymptotics and stems
from the Berry-Esseen theorem~\cite{Fel71} -- a refinement of the central limit theorem which takes into account
the rate of convergence of a distribution to a standard normal distribution. 

Two mathematical quantities play key roles in the derivation of our results. These are variants of the information spectrum
relative entropy defined by Tomamichel and Hayashi~\cite{TH13}, but have the particular advantage of satisfying the
data-processing inequality. For a state $\rho$ and a positive semi-definite operator $\sigma$, we denote them as
$\uD_s^\eps(\rho||\sigma)$ and $\oD_s^\eps(\rho||\sigma)$ (where $0<\eps<1$) and refer to them simply as {\em{information spectrum
relative entropies}}.\footnote{ To avoid confusion, we denote the information spectrum
relative entropy defined in~\cite{TH13} by $D_s^\eps(\rho||\sigma)$.} These notations and nomenclatures stem from the fact that for any arbitrary sequence of states
${\hrho}\coloneqq \{ \rho_n\}_{n\in\mathbb{N}}$ and positive semi-definite operators ${\hsigma}\coloneqq \{ \sigma_n\}_{n\in\mathbb{N}}$
such that $\rho_n$ and $\sigma_n$ are defined on the Hilbert space $\cH^{\otimes n}$ for each $n\in\mathbb{N}$, the 
following relations hold:
\be\label{relinf} \lim\limits_{\eps\rightarrow 0}\liminf\limits_{n\rightarrow\infty}\dfrac{1}{n}\uds(\rho_n\|\omega_n) = \Du(\hrho\|\homega)\quad\text{and}\quad
		\lim\limits_{\eps\rightarrow 0}\limsup\limits_{n\rightarrow\infty}\dfrac{1}{n}\ods(\rho_n\|\omega_n) = \oD(\hrho\|\hsigma),
\ee
where $ \Du(\hrho\|\homega)$ and $\oD(\hrho\|\hsigma)$ are the inf- and sup- spectral divergence rates of the so-called
quantum Information Spectrum Approach (ISA) (see \Cref{def:spectral-divergence-rates}). The ISA provides a unifying mathematical framework for obtaining asymptotic rate formulae for various different tasks in information theory. The power of the approach
lies in the fact that it does not rely on any particular structure or property of the resources used in the tasks.
It was introduced by Han and Verdu~\cite{HV93} in classical Information Theory, and generalized to the quantum setting
by Hayashi, Nagaoka and Ogawa~\cite{Hay06,HN03,NH07,ON00}.

The information spectrum relative entropies, $\uD_s^\eps(\rho||\sigma)$ and $\oD_s^\eps(\rho||\sigma)$, can also be related to other relative entropies which arise in one-shot information theory (see \cite{Ren05,KRS09,Dat09,Tom12,DH11,MW12} and references therein), e.g.~the hypothesis testing relative entropy $D_H^\eps(\rho||\sigma)$~\cite{WR12}, and the smooth max-relative entropy $D_{\max}^\eps(\rho||\sigma)$~\cite{Dat09}. In fact, one can prove that all these relative entropies are equivalent in the sense that upper and lower bounds to any one of them can be obtained in terms of any one of the others, modulo terms which depend only on the (smoothing) parameter $\eps$. These equivalences prove very useful. In particular, the bounds on the information spectrum relative entropy in terms of $D_H^\eps(\rho||\sigma)$ directly yield second order asymptotic expansions for $\uD_s^\eps(\rho||\sigma)$ and $\oD_s^\eps(\rho||\sigma)$ via the second order asymptotic expansion for $D_H^\eps(\rho||\sigma)$, which has been derived by Tomamichel and Hayashi in~\cite{TH13}.

In addition, as in the case of the usual relative entropy, one can derive other entropic quantities, namely, 
entropy, conditional entropy and mutual information, from the information spectrum
relative entropies. For each of the information-processing tasks considered in this paper, we obtain one-shot bounds in terms of 
quantities derived from the information spectrum relative entropies. The second order expansions
of these quantities then directly yield tight second order asymptotics for optimal rates of the corresponding tasks 
in the i.i.d.~setting.

Finally, the relations \reff{relinf} enable us to directly obtain the more general results of the ISA, for each of the tasks
considered, from our one-shot bounds. For example, our bounds for one-shot fixed-length source coding yields the optimal
data compression limit for a general (i.e.~not necessarily memoryless) quantum information source.

\section{Overview of results}
Here we summarize our main contributions and give pointers to the relevant theorems.
\begin{itemize}
\item{We define the information spectrum relative entropies in Definition \ref{def:information-spectrum-rel-ent} as a variant of a previously introduced relative entropy \cite{TH13}, and furthermore define quantities derived from them in Definition \ref{def:children-entropies}. These quantities all depend on a parameter $0<\eps <1$.}
\item{We prove the data-processing inequalities and other properties of the information spectrum relative entropies in \Cref{prop:D_s-properties}. We also show in \Cref{thm:D_s-bounds} that the information spectrum relative entropies are equivalent to the hypothesis testing relative entropy and the smooth max-relative entropy, in the sense that these quantities are bounded by each other.}
\item{Using their equivalences with the hypothesis testing relative entropy, and the second order asymptotic expansion for the latter~\cite{TH13}, we obtain second order asymptotic expansions for the information spectrum relative entropies in \Cref{thm:D_s-asymptotics}.}
\item{In \Cref{prop:limit-inf-spectrum} we prove that the information spectrum relative entropies reduce to the spectral divergence rates of the ISA in the asymptotic limit, when the parameter $\eps$ is taken to zero.}
\item{We obtain one-shot bounds for the following information-processing tasks in terms of quantities derived from the
information spectrum relative entropies. In each case the parameter $\eps$ plays the role of the error threshold allowed in the protocol.
\begin{enumerate}
\item Quantum fixed-length source coding [\Cref{thm:source-coding-one-shot}].
\item Noisy dense coding [\Cref{thm:dense-coding-one-shot}].
\item Entanglement concentration [\Cref{thm_ec}].
\item Pure-state entanglement dilution [\Cref{thm_entdil}].
\item Classical-quantum channel coding [\Cref{thm:cq-one-shot}].
\end{enumerate}}
\item We obtain second order asymptotic expansions for the optimal rates of the
above tasks in the i.i.d.~setting in Theorems~\ref{thm:sc-second-order}, \ref{thm:dc-second-order-expansion}, \ref{thm:ec-second-order}, \ref{thm:ed-second-order} and \Cref{cq-second-order-asymptotics} respectively.

In particular, \Cref{thm:sc-second-order}(i) gives the second order expansion for fixed-length visible quantum source coding. We also obtain asymptotic upper and lower bounds on the minimal compression length in the blind setting, stated in \Cref{thm:sc-second-order}(ii). The distinction between visible and the blind source coding is elaborated in Section \ref{sec:blind-vs-visible}.
\item{Even though the leading order terms
for the optimal rates for entanglement concentration and dilution are identical (and given by the entropy of entanglement), there is a difference in their second order terms. Explicit evaluation of these terms lead to a refinement of the {\em{inefficiency}} of these protocols. In the case of entanglement dilution, the latter quantity (studied by Harrow and Lo~\cite{HL04}) was introduced as a measure of the amount of entanglement wasted (or lost) in the dilution process. More precisely, in~\cite{HL04} it was proved that the number of ebits needed to create $n$ copies of a desired bipartite pure state $\psi_{AB}$ with entropy of entanglement $E$ was of the form $E n + \Omega(\sqrt{n})$. We prove that the number of ebits can, in fact, be expressed in the form $En + b\sqrt{n} + \cO(\log n)$, and we evaluate the coefficient $b$ explicitly. 

We also show how the irreversibility of entanglement concentration, established by Kumagai and Hayashi \cite{KH13}, can be proved using the discrepancy between the asymptotic expansions for distillable entanglement and entanglement
cost in the second order ($\sqrt{n}$).}
\item Finally, in \Cref{prop:ISA-recovery}, we recover the known expressions for optimal rates in the case of arbitrary resources as obtained by the Information Spectrum Approach.
\end{itemize}

The paper is organized as follows. In the next section, we introduce necessary notation and
definitions. The
rest of the paper proceeds in the order of the results mentioned above. We end with a
conclusion that summarizes our results and points to open questions for future research.

\section{Mathematical preliminaries}
For a Hilbert space $\cH$, let $\cB(\cH)$ denote the algebra of linear operators acting on $\cH$, and let  $\cP(\cH)$ denote
the set of positive semi-definite operators on $\cH$. Further, let $\cD(\cH)\coloneqq\lbrace\rho\in\cP(\cH)\mid \tr\rho=1\rbrace$ and $\cDs(\cH)\coloneqq \lbrace\rho\in\cP(\cH)\mid \tr\rho\leq 1\rbrace$ denote the set of 
states (density matrices) and subnormalized states on $\cH$ respectively. For a state $\rho\in\cD(\cH)$, the von Neumann entropy $S(\rho)$ is defined as $S(\rho)= - \tr\rho\log\rho$. Here and henceforth, all logarithms are taken to base $2$. Unless otherwise stated, we
assume all Hilbert spaces to be finite-dimensional. Let $\one\in\cP(\cH)$ denote the identity operator on $\cH$, and $\id\colon\cB(\cH)\rightarrow\cB(\cH)$ the identity map on operators on $\cH$. For a pure state $|\psi\rangle$, we denote the corresponding projector by $\psi\equiv|\psi\X\psi|$. For a completely positive, trace-preserving (CPTP) map $\Lambda\colon\cD(\cH_A)\rightarrow\cD(\cH_{A'})$, we also use the shorthand notation $\Lambda^{A\rightarrow A'}$.

For self-adjoint operators $A, B \in \cB(\cH)$, let $\lbrace A \geq B\rbrace$ denote the spectral projection of the difference operator $A-B$ corresponding to the interval $[0, +\infty)$; the spectral projections $\lbrace A > B\rbrace$, $\lbrace A \leq B\rbrace$ and $\lbrace A < B\rbrace$ are defined in a similar way. We also define $A_+= PAP$ where $P\coloneqq \lbrace A\geq 0\rbrace$ for any self-adjoint operator $A\in\cB(\cH)$. The following lemmas are used in our proofs:

\begin{lemma}\label{lem:tr-projector}
	{\normalfont \cite{ON00}} Let $A,B\geq 0$ and $0\leq P \leq\one$ be an arbitrary operator, then
	\begin{align*}
	\tr[\lbrace A\geq B\rbrace(A-B)]\geq\tr [P(A-B)],
	\end{align*}
	 and the same assertion holds for $\lbrace A>B\rbrace$.
\end{lemma}
\begin{lemma}\label{lem:dpi-tr-quantity}
	{\normalfont \cite{BD06}} Let $A,B\geq 0$ and $\Lambda$ be a CPTP map. Then 
	\begin{align*}
	\tr[(\Lambda(A)-\Lambda(B))\lbrace \Lambda(A)>\Lambda(B)\rbrace]\leq\tr[(A-B)\lbrace A> B\rbrace].
	\end{align*}
\end{lemma}
\begin{lemma}\label{lem3}
	{\normalfont \cite{BD06}} Let $\rho \in \cD_{\le}(\cH)$ and $\sigma \in \cP(\cH)$. Then for any $\gamma \in \mathbb{R}$,
\begin{align*}
\tr[\lbrace \rho \ge 2^{-\gamma}\sigma\rbrace \sigma]\leq 2^\gamma.
\end{align*}
\end{lemma}

For the convenience of the reader, we recall the definitions of the most important distance measures and state their relations \cite{TCR10}:
\begin{definition}\label{def:distance-measures}
	Let $\rho,\sigma\in\cD_\leq(\cH)$ be subnormalized states, then:
	\begin{enumerate}[{\normalfont (i)}]
		\item The \emph{generalized fidelity} $F(\rho,\sigma)$ of $\rho$ and $\sigma$ is defined by
		\begin{align*}
		F(\rho,\sigma) \coloneqq \|\sqrt{\rho}\sqrt{\sigma}\|_1 + \sqrt{(1-\tr\rho)(1-\tr\sigma)}.
		\end{align*}
		\item The \emph{purified distance} $P(\rho,\sigma)$ is defined by 
		\begin{align*}
		P(\rho,\sigma) \coloneqq \sqrt{1-F(\rho,\sigma)^2}
		\end{align*} 
		and constitutes a metric on $\cDs(\cH)$, i.e.~it satisfies the triangle inequality.
		\item The \emph{generalized trace distance} $d(\rho,\sigma)$ is defined by
		\begin{align*}
		d(\rho,\sigma)\coloneqq \frac{1}{2}\|\rho-\sigma\|_1+\frac{1}{2}|\tr\rho-\tr\sigma|
		\end{align*}
		and constitutes a metric on $\cDs(\cH)$.
	\end{enumerate}	
	If at least one of the subnormalized states $\rho$ and $\sigma$ is normalized, the generalized fidelity reduces to the usual fidelity, i.e. $F(\rho,\sigma)=\|\sqrt{\rho}\sqrt{\sigma}\|_1 = \tr(\sqrt{\rho}\,\sigma\sqrt{\rho})^{\frac{1}{2}}.$
	If both $\rho$ and $\sigma$ are normalized, the generalized trace distance also reduces to the usual trace distance, $d(\rho,\sigma) = \frac{1}{2}\|\rho-\sigma\|_1.$
\end{definition}
\begin{lemma}\label{lem:trace-purified-bound}
{\normalfont \cite{Tom12}} For $\rho,\sigma\in\cDs(\cH)$ we have the following bounds:
\begin{align*}
d(\rho,\sigma)\leq P(\rho,\sigma)\leq \sqrt{2d(\rho,\sigma)}
\end{align*}
\end{lemma}
The following entropic quantities play a key role in second order asymptotic expansions:
\begin{definition}\label{def:quantum-relative-entropy}
			For $\rho\in\cD(\cH)$ and $\sigma\in\cP(\cH)$, the \emph{quantum relative entropy} $D(\rho\|\sigma)$ is defined as 
			\begin{align*}
			D(\rho\|\sigma) \coloneqq  \begin{cases}\tr[\rho(\log\rho-\log\sigma)] & \supp\rho\subseteq\supp\sigma\\ \infty & \text{else.}\end{cases}
			\end{align*}
			The \emph{quantum information variance} $V(\rho\|\sigma)$ is defined as
			\begin{align*}
			V(\rho\|\sigma) \coloneqq  \tr\left[\rho(\log\rho-\log\sigma)^2\right]-D(\rho\|\sigma)^2,
			\end{align*}
			and we set 
			\begin{align}\label{eq:frak-s}
			\sk(\rho\|\sigma)\coloneqq \sqrt{V(\rho\|\sigma)}.
			\end{align}
\end{definition}

The inverse of the cumulative distribution function (cdf) of a standard normal random variable is defined by
\begin{align*}
\Phi^{-1}(\eps) \coloneqq \sup\lbrace z\in\mathbb{R}\mid\Phi(z)\leq \eps\rbrace,
\end{align*}
where $\Phi(z) = \frac{1}{\sqrt{2\pi}}\int_{-\infty}^z e^{-t^2/2}dt$. We frequently make use of the following lemma:
\begin{lemma}\label{lem:phi-trick}
	Let $\eps>0$, then \begin{align*}
	\sqrt{n}\,\Phi^{-1}\left(\eps\pm\frac{1}{\sqrt{n}}\right) = \sqrt{n}\,\Phi^{-1}(\eps)\pm \left(\Phi^{-1}\right)'(\xi)
	\end{align*} 
	for some $\xi$ with $|\xi-\eps|\leq\frac{1}{\sqrt{n}}$.
\end{lemma}
\begin{proof}
	We make the following general observation: Let $f\colon\mathbb{R}\rightarrow\mathbb{R}$ be a continuously differentiable function. Then by Taylor's theorem we can write 
	\begin{align}\label{eq:taylor}
		f\left(x\pm\frac{1}{\sqrt{n}}\right) = f(x) \pm \frac{1}{\sqrt{n}}f'(\xi)
	\end{align}
	for some $\xi\in\left[x-\frac{1}{\sqrt{n}},x\right]$ in the case `$-$' and $\xi\in\left[ x,x+\frac{1}{\sqrt{n}}\right]$ in the case `$+$'. Applying \eqref{eq:taylor} to $\sqrt{n}\,\Phi^{-1}\left(\eps\pm\frac{1}{\sqrt{n}}\right)$ yields the claim.
\end{proof}

\section{Information spectrum relative entropies}
\subsection{Definitions and mathematical properties}
The central quantities of this paper are the information spectrum relative entropies $\uds(\rho\|\sigma)$ and $\ods(\rho\|\sigma)$, which we now define.
\begin{definition}\label{def:information-spectrum-rel-ent}
	For $\rho\in\cD(\cH)$, $\sigma\in\cP(\cH)$ and $\eps\in (0,1)$, the \emph{information spectrum relative entropies} are defined as
	\begin{align*}
		\uds(\rho\|\sigma) &\coloneqq  \sup\lbrace \gamma\mid \tr\left[(\rho-2^\gamma\sigma)\lbrace \rho>2^\gamma\sigma\rbrace\right] \geq 1-\eps\rbrace\\
		\ods(\rho\|\sigma) &\coloneqq  \inf\lbrace \gamma\mid \tr[(\rho-2^\gamma\sigma)\lbrace \rho>2^\gamma\sigma\rbrace] \leq \eps\rbrace.
	\end{align*}
\end{definition}
These relative entropies are one-shot generalizations of the spectral divergences used in the ISA to quantum information theory. In Section \ref{sec:information-spectrum-method} these generalizations are discussed in detail. 

Note furthermore, that the information spectrum relative entropies as defined in \Cref{def:information-spectrum-rel-ent} are variants of the definition introduced by Tomamichel and Hayashi \cite{TH13}: 
\begin{align}\label{eq:tom-hay-definition}
\dsc(\rho\|\sigma)\coloneqq  \sup\lbrace \gamma\mid \tr[\rho\lbrace\rho\leq 2^\gamma\sigma\rbrace] \leq\eps\rbrace.
\end{align} 
Our definitions have the advantage of satisfying the data processing inequality, as shown in \Cref{prop:D_s-properties}(iii). First, we prove the following lemma which implies that the supremum and infimum in \Cref{def:information-spectrum-rel-ent} are attained.
\begin{lemma}~\label{lem:trace-quantity}
\begin{enumerate}[{\normalfont (i)}]
\item\label{item:one-parameter-continuous} For $\gamma\in\mathbb{R}$ let $A(\gamma)\in\cB(\cH)$ be a one-parameter family of Hermitian operators such that $\gamma\mapsto \|A(\gamma)\|$ is continuous.
Then the function $\gamma\mapsto \tr(A(\gamma))_+$ is continuous.
\item\label{item:trace-functional-continuous} For $\rho,\sigma\in\cP(\cH)$, the function $\gamma\mapsto \tr(\rho-2^\gamma\sigma)_+$ is continuous and monotonically decreasing, and strictly so if $\supp\rho\subseteq\supp\sigma$.

\item\label{item:sup-inf-achieved} The supremum and infimum in \Cref{def:information-spectrum-rel-ent} are attained, and unique if $\supp\rho\subseteq\supp\sigma$.
\end{enumerate}
\end{lemma}
\begin{proof}
(i) First, we observe that the trace of the positive part of a Hermitian operator is the sum of its non-negative eigenvalues, the eigenvalues being the roots of the characteristic polynomial. 
Since the determinant is a continuous function with respect to the operator norm on $\cB(\cH)$, the coefficients of the characteristic polynomial of $A(\gamma)$ continuously depend on $\gamma$, and by factorizing the polynomial into linear factors we see that this carries over to the roots. 
Composing the sum over the roots with the continuous maximum function $\max(.,0)$, we obtain that the function $\gamma\mapsto \tr(A(\gamma))_+$ is a composition of continuous functions and thus continuous itself.

(ii) Since $\gamma\mapsto \rho-2^\gamma\sigma$ is continuous with respect to the operator norm, the function $\gamma\mapsto\tr(\rho-2^\gamma\sigma)_+$ is continuous by (i).
To prove that $\gamma\mapsto \tr\left(\rho-2^\gamma\sigma\right)_+$ is decreasing, let $\gamma\leq \gamma'$ and observe that
\begin{align*}
\rho - 2^{\gamma} \sigma = \rho- 2^{\gamma'}\sigma + (2^{\gamma'} - 2^{\gamma})\sigma,
\end{align*}
where $(2^{\gamma'} - 2^{\gamma})\sigma\geq 0$.
Hence, Weyl's Monotonicity Theorem (see e.g.~Section III in \cite{Bha13}) implies that
\begin{align}\label{eq:weyl-monotonicity}
\lambda_j \left(\rho - 2^{\gamma} \sigma\right) \geq \lambda_j(\rho- 2^{\gamma'}\sigma) \quad\text{for all $j$,}
\end{align}
where for a Hermitian operator $A$ we write $\lambda_j(A)$ to denote the $j$-th largest eigenvalue of $A$.
It then follows from \eqref{eq:weyl-monotonicity} that $\tr\left(\rho-2^\gamma\sigma\right)_+ \geq \tr(\rho-2^{\gamma'}\sigma)_+$.

To prove strict monotonicity, let $\gamma<\gamma'$.
Without loss of generality, we restrict $\cH$ to the support of $\sigma$ (which is possible due to the assumption $\supp\rho\subseteq\supp\sigma$), such that $\sigma$ has strictly positive eigenvalues.
Consequently, $(2^{\gamma'} - 2^{\gamma})\sigma>0$, and the inequality in \eqref{eq:weyl-monotonicity} is a strict one for all $j$, from which we obtain that $\tr\left(\rho-2^\gamma\sigma\right)_+ > \tr(\rho-2^{\gamma'}\sigma)_+$.

(iii) Since $\lim_{\gamma\rightarrow -\infty}\tr(\rho-2^\gamma\sigma)_+ = \tr(\rho)=1$ and $\lim_{\gamma\rightarrow\infty}\tr(\rho-2^\gamma\sigma)_+=0$, we infer by (ii) and the Intermediate Value Theorem that the supremum in the definition of $\uds(\rho\|\sigma)$ as well as the infimum in the definition of $\ods(\rho\|\sigma)$ are attained. 
If $\supp\rho\subseteq\supp\sigma$, then they are moreover unique by the strict monotonicity of $\gamma\mapsto \tr\left(\rho-2^{\gamma}\sigma\right)_+$.
\end{proof}

We are now ready to record a series of properties of the information spectrum relative entropies:
\begin{proposition}\label{prop:D_s-properties}
	Let $\eps\in (0,1)$, $\rho\in\cD(\cH)$ and $\sigma\in\cP(\cH)$. Then the following properties hold:
	\begin{enumerate}[{\normalfont (i)}]
		\item $\uds(\rho\|\sigma)\leq\dsc(\rho\|\sigma)$
		\item $\uds(\rho\|\sigma) = \oD_s^{1-\eps}(\rho\|\sigma)$
		\item Data processing inequality: For any CPTP map $\Lambda$, we have
		\begin{align*}
			\uds(\rho\|\sigma)&\geq\uds(\Lambda(\rho)\|\Lambda(\sigma))\\
			\ods(\rho\|\sigma)&\geq\ods(\Lambda(\rho)\|\Lambda(\sigma)).
		\end{align*}
		\item Monotonicity in $\eps$: Let $\eps'\geq\eps$, then
		\begin{align*}
			\uds(\rho\|\sigma) &\leq \underline{D}_s^{\eps'}(\rho\|\sigma)\\
			\oD_s^{\eps'}(\rho\|\sigma) &\leq \ods(\rho\|\sigma).
		\end{align*}
		\item Let $\sigma'\geq 0$ with $\sigma\leq\sigma'$, then $\uds(\rho\|\sigma)\geq\uds(\rho\|\sigma').$
		\item Let $c>0$, then $\uds(\rho\|c\sigma) = \uds(\rho\|\sigma) - \log c$.
		\item Let $\delta > 0$ and $\rho'\in\cD(\cH)$ with $d(\rho,\rho')\leq \delta$, then $\uds(\rho'\|\sigma)\leq \uD_s^{\eps+\delta}(\rho\|\sigma).$
	\end{enumerate}
\end{proposition}

\begin{proof}
	(i) Let $\gamma=\uds(\rho\|\sigma)$. Then by the definition of $\uds(\rho\|\sigma)$ we have
	\begin{align*}
		1-\eps &=\tr[(\rho-2^\gamma\sigma)\lbrace\rho > 2^\gamma\sigma\rbrace]\\
		&\leq \tr[\rho\lbrace\rho>2^\gamma\sigma\rbrace].
	\end{align*}
	Hence,
	\begin{align*}
		\tr[\rho\lbrace\rho\leq 2^\gamma\sigma\rbrace] &= \tr\rho-\tr[\rho\lbrace\rho>2^\gamma\sigma\rbrace]\\
		&\leq \tr\rho-(1-\eps)\\
		&= \eps,
	\end{align*}
	since $\tr\rho= 1$ by assumption. Therefore, $\gamma$ is feasible for $\dsc(\rho\|\sigma)$ and consequently, from \eqref{eq:tom-hay-definition} it follows that
	\begin{align*}
	\dsc(\rho\|\sigma)\geq \gamma = \uds(\rho\|\sigma),
	\end{align*}
	which yields the claim.

	(ii) Let $\gamma=\oD_s^{1-\eps}(\rho\|\sigma)$. Then by the definition of $\oD_s^{1-\eps}(\rho\|\sigma)$, we have
	$$\tr(\rho-2^\gamma\sigma)_+ = 1-\eps.$$
	Hence, $\gamma$ is feasible for $\uD_s^{\eps}(\rho\|\sigma)$, and we obtain 
	$$\uD_s^{\eps}(\rho\|\sigma)\geq \gamma = \oD_s^{1-\eps}(\rho\|\sigma).$$
	
	Assume now that $$\gamma=\oD_s^{1-\eps}(\rho\|\sigma)= \uds(\rho\|\sigma)-\delta$$ for some $\delta>0$, i.e.~$\oD_s^{1-\eps}(\rho\|\sigma)< \uds(\rho\|\sigma)$. By the monotonicity of $\alpha\mapsto\tr(\rho-2^\alpha\sigma)_+$, we have 
	$$\tr(\rho-2^\gamma\sigma)_+> 1-\eps.$$ On the other hand, $\tr(\rho-2^\gamma\sigma)_+ = 1-\eps$ by definition of $\oD_s^{1-\eps}(\rho\|\sigma)$. This leads to a contradiction, yielding $\uds(\rho\|\sigma)=\oD_s^{1-\eps}(\rho\|\sigma)$.

	(iii) Let $\gamma=\uds(\Lambda(\rho)\|\Lambda(\sigma))$. Then \Cref{lem:dpi-tr-quantity} implies that 
	\begin{align*}
		\tr(\rho-2^\gamma\sigma)_+ \geq \tr(\Lambda(\rho)-2^\gamma\Lambda(\sigma))_+ = 1-\eps.
	\end{align*}
	Hence, $\gamma$ is feasible for $\uds(\rho\|\sigma)$, and we obtain
	$$\uds(\rho\|\sigma)\geq \gamma = \uds(\Lambda(\rho)\|\Lambda(\sigma)).$$
	
	Similarly, let $\gamma=\ods(\rho\|\sigma)$. Then
	\begin{align*}
		\eps = \tr(\rho-2^\gamma\sigma)_+ \geq \tr(\Lambda(\rho)-2^\gamma\Lambda(\sigma))_+
	\end{align*}
	by \Cref{lem:dpi-tr-quantity}. Hence, $\gamma$ is feasible for $\ods(\Lambda(\rho)\|\Lambda(\sigma))$, and we obtain
	\begin{align*}
		\ods(\Lambda(\rho)\|\Lambda(\sigma)) \leq \gamma = \ods(\rho\|\sigma).
		\end{align*}
		
	(iv) 	Let $\gamma=\underline{D}_s^{\eps}(\rho\|\sigma)$, then
	$$\tr[\lbrace\rho > 2^\gamma\sigma\rbrace(\rho-2^\gamma\sigma)]= 1-\eps \geq 1-\eps'.$$ Hence, $\gamma$ is feasible for $\underline{D}_s^{\eps'}(\rho\|\sigma)$, and consequently, $\underline{D}_s^{\eps'}(\rho\|\sigma) \geq \gamma=\uds(\rho\|\sigma)$.
	
	Similarly, let $\gamma=\ods(\rho\|\sigma)$. Then 
	$$\tr(\rho-2^\gamma\sigma)_+ =\eps\leq \eps',$$
	and hence, $\gamma$ is feasible for $\oD_s^{\eps'}(\rho\|\sigma)$, and we obtain $\oD_s^{\eps'}(\rho\|\sigma)\leq \gamma = \ods(\rho\|\sigma)$.
	
	(v) Let $\gamma=\uds(\rho\|\sigma')$ and $Q=\lbrace \rho>2^{\gamma}\sigma'\rbrace$, then we compute:
	\begin{align*}
	1-\eps & = \tr[Q(\rho-2^{\gamma}\sigma')]\\
	&= \tr Q\rho  - 2^\gamma \tr Q\sigma'\\
	&\leq \tr Q\rho - 2^\gamma \tr Q\sigma\\
	&\leq \tr (\rho-2^\gamma\sigma)_+
	\end{align*}
	where the first inequality follows from $\sigma\leq\sigma'$ and the second inequality follows from \Cref{lem:tr-projector}. Hence, $\gamma$ is feasible for $\uds(\rho\|\sigma)$ and we obtain $\uds(\rho\|\sigma)\geq\uds(\rho\|\sigma')$.
	
	(vi) Let $\gamma=\uds(\rho\|c\sigma)$ for $c>0$, then
	\begin{align*}
	1-\eps = \tr(\rho-2^\gamma c\sigma)_+ =\tr(\rho-2^{\gamma+\log c}\sigma)_+
	\end{align*}
	and hence,
	$$\uds(\rho\|\sigma) \geq\gamma + \log c = \uds(\rho\|c\sigma)+\log c.$$
	
	Conversely, let $\gamma=\uds(\rho\|\sigma)$. Then
	$$1-\eps = \tr(\rho-2^\gamma\sigma)_+ = \tr(\rho-2^{\gamma-\log c}c\sigma)_+,$$
	and hence, 
	$$\uds(\rho\|c\sigma) \geq \gamma - \log c = \uds(\rho\|\sigma)-\log c,$$
	which proves the claim.
	
	(vii) Let $\gamma=\uds(\rho'\|\sigma)$ and $Q=\lbrace \rho'\geq 2^\gamma\sigma\rbrace$. Then
	\begin{align*}
	1-\eps &= \tr[(\rho'-2^\gamma\sigma)Q]\\
	&= \tr[(\rho'-\rho)Q] + \tr[(\rho-2^\gamma\sigma)Q]\\
	&\leq \tr(\rho'-\rho)_+ + \tr(\rho-2^\gamma\sigma)_+\\
	&\leq \delta + \tr(\rho-2^\gamma\sigma)_+,
	\end{align*}
	where the first inequality follows from \Cref{lem:tr-projector} and the second inequality follows from the fact \cite{Tom12} that 
	$$d(\rho,\rho')=\tr(\rho-\rho')_+ \leq \delta$$ by assumption. This proves the claim.
\end{proof}
\begin{remark}
\Cref{prop:D_s-properties}(ii) shows that we only need to focus on one of the information spectrum relative entropies (we choose $\uds(\rho\|\sigma)$ without loss of generality). However, given their close relationship to the quantum spectral inf- and sup divergence (Section \ref{sec:information-spectrum-method}), we note that it is useful to keep the two alternative definitions in \Cref{def:information-spectrum-rel-ent}.
\end{remark}

\subsubsection*{Children entropies of the information spectrum relative entropies}

The quantum relative entropy acts as a parent quantity for other entropic quantities:
\begin{itemize}
	\item the von Neumann entropy $S(\rho)=-D(\rho\|\one)$
	\item the quantum conditional entropy $S(A|B)_\rho = -D(\rho_{AB}\|\one_A\otimes\rho_B)$
	\item the quantum mutual information $I(A:B)_\rho = \min_{\sigma_B}D(\rho_{AB}\|\rho_A\otimes\sigma_B)$
\end{itemize}
This motivates us to define the following information spectrum entropies:
\begin{definition}\label{def:children-entropies}
	Let $\eps\in (0,1)$ and $\omega\in\cD(\cH), \sigma_B\in\cD(\cH_B), \rho_{AB}\in\cD(\cH_{AB})$ be states. Then we define:
	\begin{enumerate}[{\normalfont (i)}]
		\item the \emph{information spectrum entropies}
			\begin{align*}
				\uhs(\omega)&\coloneqq -\ods(\omega\|\one)\\
				\ohs(\omega)&\coloneqq -\uds(\omega\|\one)
			\end{align*}
		\item the \emph{information spectrum conditional entropies}
			\begin{align*}
				\uhs(A|B)_\rho &\coloneqq -\ods(\rho_{AB}\|\one_A\otimes\rho_B)\\
				\ohs(A|B)_\rho &\coloneqq  -\uds(\rho_{AB}\|\one_A\otimes\rho_B)
			\end{align*}
		\item the \emph{information spectrum mutual informations}
			\begin{align*}
				\uis(A:B)_\rho &\coloneqq  \min_{\sigma_B}\uds(\rho_{AB}\|\rho_A\otimes\sigma_B)\\
				\ois(A:B)_\rho &= \min_{\sigma_B}\ods(\rho_{AB}\|\rho_A\otimes\sigma_B)
			\end{align*}
	\end{enumerate}
\end{definition}
Note that in \Cref{def:children-entropies}, (i) and (ii) the occurrence of the minus sign is the reason for changing the upper bar to a lower bar and vice versa. In Section \ref{sec:information-tasks} these quantities arise in one-shot bounds for operational tasks. 

The information spectrum conditional entropies satisfy the following interesting property under local operations and classical communication (LOCC) taking pure states to pure states:
\begin{lemma}\label{lem:lo}
Let $\phi_{AB}\coloneqq \Lambda(\psi_{AB})$ where $\Lambda$ denotes any LOCC operation which takes pure states to pure states,
and $\psi_{AB}\in\cD(\cH_{AB})$ is a bipartite pure state.
Then the following inequality holds:
\begin{align*}
\oH_s^\eps(A|B)_\psi \le \oH_s^\eps(A|B)_\phi
\end{align*}
\end{lemma}
\begin{proof}
By definition, we have $- \oH_s^\eps(A|B)_\psi = \uD_s^\eps(\psi_{AB}||\one_A \otimes \psi_B),$ and furthermore, 
\begin{align*}
- \oH_s^\eps(A|B)_\phi &= \uD_s^\eps(\phi_{AB}||\one_A \otimes \phi_B)\\
&= \uD_s^\eps(\Lambda(\psi_{AB})||\one_A \otimes \phi_B).
\end{align*}
However, by a result of Lo and Popescu~\cite{LP01}, the action of the
LOCC map $\Lambda$ on the pure state $\psi_{AB}$ can be expressed as follows:
\begin{align}\label{Lo-Pop}
\phi_{AB}= \Lambda(\psi_{AB}) = \sum_j (U_j \otimes K_j)\psi_{AB}(U_j^\dagger \otimes K_j^\dagger),
\end{align}
where the $U_j$ are unitary operators and $K_j$ are operators such that $\sum_j K_j^\dagger K_j = \one_B$. Consequently, using the cyclicity of the trace we obtain
\begin{align}\label{eq:red1}
\phi_B = \tr_A \phi_{AB} = \sum_j K_j\psi_{B}K_j^\dagger.
\end{align}
Further, for $\psi_B = \tr_A \psi_{AB}$, we have
\begin{align*}
\Lambda( \one_A \otimes \psi_B) =  \one_A \otimes \phi_B,
\end{align*}
which can be seen as follows:
\begin{align*}
\Lambda(\one_A \otimes \psi_B) &= \sum_j U_j U_j^\dagger \otimes K_j\psi_{B}K_j^\dagger\nonumber\\
&= \one_A \otimes \sum_j  K_j\psi_{B}K_j^\dagger\nonumber\\
&= \one_A \otimes \phi_B,
\end{align*}
where the second identity follows from the unitarity of the operators $U_j$, and the last identity follows from \eqref{eq:red1}. Hence,
\begin{align*}
- \oH_s^\eps(A|B)_\phi &=  \uD_s^\eps(\Lambda(\psi_{AB})||\one_A \otimes \phi_B)\nonumber\\
&=  \uD_s^\eps(\Lambda(\psi_{AB})||\Lambda(\one_A \otimes \psi_B))\nonumber\\
& \le  \uD_s^\eps(\psi_{AB}||\one_A \otimes \psi_B)\nonumber\\
&= -  \oH_s^\eps(A|B)_\psi,
\end{align*}
where the inequality follows from \Cref{prop:D_s-properties}, (iii).
\end{proof}

\subsection{Relation to other relative entropies}
\label{sec:relation-to-other-entropies}
In this section we prove that the information spectrum relative entropies are equivalent to the hypothesis testing relative entropies and the smooth max-relative entropy, which arise in one-shot information theory. This equivalence is in the sense that upper and lower bounds to any one of them can be obtained in terms of any one of the others, modulo terms which depend only on the (smoothing) parameter $\eps$. Let us first recall the definitions of the hypothesis testing relative entropies and the smooth max-relative entropy:
\begin{definition}~\label{def:rel-entropies}
	\begin{enumerate}[{\normalfont (i)}]
		\item For $0<\eps<1$ and $\rho,\sigma\in\cD(\cH)$, the \emph{hypothesis testing relative entropy} $D_H^\eps(\rho\|\sigma)$ is defined as
			$$D_H^\eps(\rho\|\sigma)\coloneqq  -\log\inf_{0\leq Q\leq\one}\lbrace \tr Q\sigma\mid\tr Q\rho\geq 1-\eps\rbrace.$$

		\item For $\eps>0$, $\rho\in\cDs(\cH)$ and $\sigma\geq 0$, the \emph{smooth max-relative entropy} $\Dmaxe(\rho\|\sigma)$ is defined as
			$$\Dmaxe(\rho\|\sigma) \coloneqq  \min_{\brho\in\Be(\rho)}\inf\lbrace\gamma\mid\brho\leq 2^\gamma\sigma\rbrace,$$
			where $\Be(\rho)\coloneqq \lbrace \brho\in\cDs(\cH)\mid F(\brho,\rho)^2\geq 1-\eps^2\rbrace$ is the $\eps$-ball with respect to the purified distance $P(\rho,\sigma)$.
	\end{enumerate}
\end{definition}
We prove the following relations between these relative entropies:
\begin{proposition}\label{thm:D_s-bounds}
	Let $0<\eps<1$, $\delta,\eta>0$, and $\rho,\sigma>0$ with $\tr\rho\leq 1$. Then we obtain the following bounds: 
	\begin{enumerate}[{\normalfont (i)}]
		\item $D_H^{\eps-\delta}(\rho\|\sigma)+\log\delta\leq\uds(\rho\|\sigma)\leq D_H^\eps(\rho\|\sigma)$
		\item $D_H^{\eps-\eta}(\rho\|\sigma) +\log\eta\leq \oD_s^{1-\eps}(\rho\|\sigma)\leq D_H^\eps(\rho\|\sigma)+\delta$
		\item $\Dmax^{\sqrt{8\eps}}(\rho\|\sigma) \leq \ods(\rho\|\sigma)\leq \Dmaxe(\rho\|\sigma)$
	\end{enumerate}
\end{proposition}
\begin{proof}
	(i) To prove the upper bound, let $Q\coloneqq\lbrace \rho > 2^\gamma\sigma\rbrace$ where $\gamma=\uds(\rho\|\sigma)-\xi$ for some $\xi>0$. Since
	$$\tr\rho Q \geq \tr[(\rho-2^\gamma\sigma)Q] \geq 1-\eps,$$
	we see that $Q$ is feasible for $D_H^\eps(\rho\|\sigma)$. Hence,
	\begin{align*}
		\tr\sigma Q &= \tr[\sigma\lbrace\rho > 2^\gamma\sigma\rbrace]\\
		&\leq 2^{-\gamma} \tr[\rho \lbrace\rho > 2^\gamma\sigma\rbrace]\\
		&\leq 2^{-\gamma}
	\end{align*}
	since $\tr\rho\leq 1$. By the definition of $D_H^\eps(\rho\|\sigma)$, we obtain
	$$D_H^\eps(\rho\|\sigma) \geq \gamma = \uds(\rho\|\sigma)-\xi,$$
	and hence the upper bound.

	Conversely, let $0\leq Q\leq\one$ be optimal for $D_H^\eps(\rho\|\sigma)$ and set $\mu=\log\delta + D_H^\eps(\rho\|\sigma)$. Then
	\begin{align*}
		\tr(\rho-2^\mu\sigma)_+ &\geq \tr[ Q(\rho-2^\mu\sigma)]\\
		&= \tr Q\rho - 2^\mu\tr Q\sigma\\
		&\geq 1-(\eps+\delta),
	\end{align*}
	where the first inequality follows from \Cref{lem:tr-projector}. Hence, $\mu$ is feasible for $\uD_s^{\eps+\delta}(\rho\|\sigma)$ and we obtain 
	$$\uD_s^{\eps+\delta}(\rho\|\sigma)\geq D_H^\eps(\rho\|\sigma)+\log\delta.$$

	(ii) We start with the upper bound. Let $\gamma=\oD_s^{1-\eps}(\rho\|\sigma)-\delta$ for some arbitrary $\delta>0$ and set $Q\coloneqq \lbrace \rho\geq 2^\gamma\sigma\rbrace$. We have $\tr [Q(\rho-2^\gamma\sigma)]> 1-\eps$ by definition of $\oD_s^{1-\eps}(\rho\|\sigma)$, and hence, 
	$$\tr Q\rho \geq \tr [Q(\rho-2^\gamma\sigma)] > 1-\eps.$$
	Thus, $Q$ is feasible for $D_H^\eps(\rho\|\sigma)$. Furthermore, 
	$$\tr Q\sigma = \tr[\lbrace \rho\geq 2^\gamma \sigma\rbrace\sigma]\leq 2^{-\gamma}\tr\rho \leq 2^{-\gamma},$$
	implying that $$D_H^\eps(\rho\|\sigma)\geq -\log\tr Q\sigma \geq \gamma = \oD_s^{1-\eps}(\rho\|\sigma)-\delta.$$
	
	Conversely, assume that $D_H^\eps(\rho\|\sigma)<\infty$ and let $0\leq Q\leq\one$ be optimal for $D_H^\eps(\rho\|\sigma)$, such that 
	\begin{align}\label{eq:optimal-Q}
		\tr Q\sigma = 2^{-D_H^\eps(\rho\|\sigma)}.
	\end{align}
	For arbitrary $\kappa$ we have
	\begin{align}
		1-\eps&\leq \tr Q\rho\notag\\
		&= \tr [Q(\rho-2^\kappa\sigma)] + 2^\kappa\tr Q\sigma\notag\\
		&\leq \tr[\lbrace\rho\geq 2^\kappa\sigma\rbrace(\rho-2^\kappa\sigma)] + 2^{\kappa-D_H^\eps(\rho\|\sigma)},\label{eq:kappa}
	\end{align}
	where we used \Cref{lem:tr-projector} and \eqref{eq:optimal-Q} for the first and second terms in the last inequality. Choosing $\kappa = \oD_s^{1-(\eps+\eta)}(\rho\|\sigma) + \delta$ for an arbitrary $\delta>0$ implies that
	\begin{align}
	\tr(\rho-2^\kappa\sigma)_+\leq 1-(\eps+\eta).\label{eq:kappa2}
	\end{align}
	From \eqref{eq:kappa} and \eqref{eq:kappa2} we have $1-\eps\leq 1-(\eps+\eta) + 2^{\kappa-D_H^\eps(\rho\|\sigma)}$ and hence
	$$\kappa = \oD_s^{1-(\eps+\eta)}(\rho\|\sigma) + \delta \geq D_H^{\eps}(\rho\|\sigma)+\log\eta,$$ which results in the lower bound, since $\delta$ is arbitrary.

		(iii) To prove the lower bound, we set $\gamma=\ods(\rho\|\sigma)+\delta$ for some arbitrary $\delta>0$. Then we obtain the bounds 
		\begin{align*}
			(\rho-2^\gamma\sigma)_+ \geq \rho-2^\gamma\sigma\quad\text{and}\quad\sqrt{8\tr(\rho-2^\gamma\sigma)_+}\leq\sqrt{8\eps},
		\end{align*}
		where the second bound follows from the definition of $\gamma$. Hence, we can apply Lemma 14 in \cite{Dat09}, which yields
		$$\Dmax^{\sqrt{8\eps}}(\rho\|\sigma)\leq\gamma = \ods(\rho\|\sigma)+\delta$$
		and therefore implies the result.
		
		Conversely, let $\trho$ achieve the maximum in the definition of $\Dmaxe(\rho\|\sigma)$, i.e.~$P(\rho,\trho)\leq\eps$ and $\trho\leq 2^\gamma\sigma$ with $\gamma=\Dmaxe(\rho\|\sigma)+\delta$ for some arbitrary $\delta>0$. Note that this implies
		\begin{align}\label{eq:Dmax-zero}
			\tr(\trho-2^\gamma\sigma)_+=0.
		\end{align}
		We proceed by employing the trace distance $d(\rho,\sigma)$, making use of the fact \cite{Tom12} that the trace distance of two subnormalized states $\rho$ and $\sigma$ can be expressed as 
		\begin{align}\label{eq:trace-distance-as-proj}
			d(\rho,\sigma)=\tr(\rho-\sigma)_+,
		\end{align}
		which is exactly of the form of the trace quantity in the definition of the information spectrum relative entropies. 
		We compute:
		\begin{align*}
			\tr(\rho-2^\gamma\sigma)_+ &\leq \tr[\lbrace\rho>2^\gamma\sigma\rbrace(\rho-2^\gamma\sigma)]\\
			&= \tr[\lbrace\rho > 2^\gamma\sigma\rbrace (\rho-\trho + \trho - 2^\gamma\sigma)]\\
			&\leq \tr[\lbrace\rho>\trho\rbrace(\rho-\trho)] + \tr[\lbrace\trho> 2^\gamma\sigma\rbrace(\trho-2^\gamma\sigma)]\\
			&= d(\rho,\trho)\\
			&\leq \eps,
		\end{align*}
		where we used \Cref{lem:tr-projector} in the second inequality and equations \eqref{eq:Dmax-zero}, \eqref{eq:trace-distance-as-proj} and \Cref{lem:trace-purified-bound} in the following lines. We infer that $\gamma$ is feasible for $\ods(\rho\|\sigma)$, and hence, we obtain
		$$\ods(\rho\|\sigma)\leq \gamma = \Dmaxe(\rho\|\sigma) + \delta,$$ which yields the claim.
\end{proof}

\subsection{Second order asymptotics of the information spectrum relative entropies}
\label{sec:2nd-order-asymptotics-D_s}
In the last section we established bounds for the information spectrum relative entropies $\uds(\rho\|\sigma)$ and $\ods(\rho\|\sigma)$ in terms of the hypothesis relative testing entropy $D_H^\eps(\rho\|\sigma)$. This implies that the second order asymptotic expansion of $D_H^\eps(\rho\|\sigma)$ readily carries over to $\uds(\rho\|\sigma)$ and $\ods(\rho\|\sigma)$. Before we prove this result, for the sake of completeness we briefly outline the procedure used by Tomamichel and Hayashi \cite{TH13} to obtain the second order asymptotic expansion of $D_H^\eps(\rho\|\sigma)$:

For a given state $\rho$ and a positive operator $\sigma$ with spectral decompositions $$\rho=\sumi_xr_x |v_x\X v_x|\quad\text{and}\quad\sigma = \sumi_y s_y |u_y\X u_y|,$$ consider first the Nussbaum-Szko\l a probability distributions defined by $$P_{\rho,\sigma}(x,y) \coloneqq r_x|\langle v_x|u_y\rangle|^2\quad\text{and}\quad Q_{\rho,\sigma}(x,y)\coloneqq s_y |\langle v_x|u_y\rangle |^2.$$ The usefulness of these distributions lies in the following fact: the first two moments of the log likelihood ratio of $P_{\rho,\sigma}$ and $Q_{\rho,\sigma}$, i.e.~the random variable $Z=\log P_{\rho,\sigma}(X,Y)-\log Q_{\rho,\sigma}(X,Y)$ distributed according to the distribution $P_{\rho,\sigma}(x,y)$, agree with the quantum relative entropy and the quantum information variance (given in \Cref{def:quantum-relative-entropy}), that is,
\begin{align}\label{eq:nussbaum-szkola}
D(P_{\rho,\sigma}\|Q_{\rho,\sigma}) = D(\rho\|\sigma)\quad\text{and}\quad V(P_{\rho,\sigma}\|Q_{\rho,\sigma})=V(\rho\|\sigma).
\end{align}
Furthermore, for i.i.d.~states $\rho^n\equiv\rho^{\otimes n}$ and $\sigma^n\equiv\sigma^{\otimes n}$ the Nussbaum-Szko\l a distributions take on the product form $P_{\rho^{n},\sigma^{n}} = P_{\rho,\sigma}^n$ and $Q_{\rho^{n},\sigma^{n}} = Q_{\rho,\sigma}^n$.

Consider the classical entropic quantity $$D_s^\eps(P||Q) \coloneqq \sup\lbrace R\mid \PP(\log P - \log Q\leq R)\leq \eps\rbrace.$$
It can be recognized to be the inverse of the cumulative distribution function of the log likelihood ratio $Z$. The Berry-Esseen theorem \cite{Fel71} states that the random variable $Y=\sqrt{n}\frac{Z-\mu}{s}$ converges to the standard normal distribution and provides a bound on the rate of this convergence. Here, $\mu$ and $s$ are the mean and standard deviation of $Z$ respectively. Carrying over the Berry-Esseen bound to the inverse cumulative distribution functions of $Y$ and $Z$, we obtain an asymptotic expansion for the quantity $D_s^\eps(P^n\|Q^n)$ in the form of
\begin{align}\label{eq:classical-entropy}
D_s^\eps(P^n\|Q^n) = n\mu + \sqrt{n}\,s\,\invP{\eps} + \cO(1),
\end{align}
where $\invP{\eps}$ is the inverse of the cdf of the standard normal distribution. Choosing $P=P_{\rho,\sigma}$ and $Q=Q_{\rho,\sigma}$ in \eqref{eq:classical-entropy} now connects this asymptotic expansion to $\rho$ and $\sigma$, since we have $\mu=D(\rho\|\sigma)$ and $s^2=V(\rho\|\sigma)$ for $Z=\log P_{\rho,\sigma}-\log Q_{\rho,\sigma}$ according to \eqref{eq:nussbaum-szkola}.

The last step consists in finding upper and lower bounds for the relative entropy in question in terms of $D_s^\eps(P_{\rho,\sigma}^n\|Q_{\rho,\sigma}^n)$, which then yields the asymptotic expansion of the former, stated precisely in the following proposition proved in \cite{TH13}.
\begin{proposition}\label{thm:D_H-asymptotics}
	 Given the i.i.d.~states $\rho^{\otimes n}$ and $\sigma^{\otimes n}$, the hypothesis testing relative entropy $D_H^\eps(\rho\|\sigma)$ has the following second order asymptotic expansion:
	$$D_H^\eps(\rho^{\otimes n}\|\sigma^{\otimes n}) = nD(\rho\|\sigma) + \sqrt{n}\,\mathfrak{s}(\rho\|\sigma)\Phi^{-1}(\eps)+\cO(\log n),$$
	where $\mathfrak{s}(\rho\|\sigma)$ is defined as in \eqref{eq:frak-s} and $\Phi^{-1}(\eps) \coloneqq  \sup\lbrace z\in\mathbb{R}\mid \Phi(z)\leq \eps\rbrace$ is the inverse of the cumulative normal distribution $\Phi$.
\end{proposition}
This yields the second order asymptotics of the information spectrum relative entropies:
\begin{proposition}\label{thm:D_s-asymptotics}
	Given i.i.d.~states $\rho^{\otimes n}$ and $\sigma^{\otimes n}$, the second order asymptotic expansions of $\uds(\rho\|\sigma)$ and $\ods(\rho\|\sigma)$ are given by:
	\begin{enumerate}[{\normalfont (i)}]
		\item $\uds(\rho^{\otimes n}\|\sigma^{\otimes n})= nD(\rho\|\sigma) + \sqrt{n}\,\sk(\rho\|\sigma)\Phi^{-1}(\eps) + \cO(\log n)$
		\item $\ods(\rho^{\otimes n}\|\sigma^{\otimes n}) = nD(\rho\|\sigma) - \sqrt{n}\,\sk(\rho\|\sigma)\Phi^{-1}(\eps) + \cO(\log n)$
	\end{enumerate}
	where $\mathfrak{s}(\rho\|\sigma)$ is defined as in \eqref{eq:frak-s} and $\Phi^{-1}(\eps) \coloneqq  \sup\lbrace z\in\mathbb{R}\mid \Phi(z)\leq \eps\rbrace$ is the inverse of the cdf of the standard normal distribution.
\end{proposition}
\begin{proof}
	We abbreviate $\rho^n\equiv\rho^{\otimes n}$ and $\sigma^n\equiv\sigma^{\otimes n}$. 
	
	(i) By \Cref{thm:D_s-bounds}, (i) we have the following bounds on $\uds(\rho^{n}\|\sigma^{n})$ for any $\delta>0$:
	\begin{align}\label{eq:D_H-iid-bounds}
		D_H^{\eps-\delta}(\rho^{n}\|\sigma^{n})+\log\delta\leq\uds(\rho^n\|\sigma^n)\leq D_H^\eps(\rho^n\|\sigma^n)
	\end{align}
	Setting $\delta=\frac{1}{\sqrt{n}}$ and applying \Cref{thm:D_H-asymptotics} to the lower bound yields
	\begin{align}
		D_H^{\eps-\delta}(\rho^n\|\sigma^n)+\log\delta &= D_H^{\eps-\frac{1}{\sqrt{n}}}(\rho^n\|\sigma^n)-\frac{1}{2}\log n\notag\\
		&=nD(\rho\|\sigma)+\sqrt{n}\,\sk(\rho\|\sigma)\Phi^{-1}\left(\eps-\frac{1}{\sqrt{n}}\right)+\cO(\log n).\label{eq:eps-delta}
	\end{align}
	\Cref{lem:phi-trick} applied to the second term yields 
	$$\sqrt{n}\,\Phi^{-1}\left(\eps-\frac{1}{\sqrt{n}}\right) = \sqrt{n}\,\Phi^{-1}(\eps)+\cO(\log n),$$ and substituting this in \eqref{eq:eps-delta} results in 
	$$D_H^{\eps-\delta}(\rho^n\|\sigma^n)+\log\delta = nD(\rho\|\sigma)+\sqrt{n}\,\sk(\rho\|\sigma)\Phi^{-1}(\eps)+\cO(\log n).$$ Using the above in the lower bound of \eqref{eq:D_H-iid-bounds} and \Cref{thm:D_H-asymptotics} for the upper bound in \eqref{eq:D_H-iid-bounds} finally proves the claim.
	
	(ii) \Cref{thm:D_s-bounds}, (ii) yields the following bounds on $\ods(\rho^n\|\sigma^n)$ for any $\delta>0$:
	\begin{align}\label{eq:D_H-iid-bounds-2}
		D_H^{1-\eps-\eta}(\rho^n\|\sigma^n)+\log\eta\leq\ods(\rho^n\|\sigma^n)\leq D_H^{1-\eps}(\rho^n\|\sigma^n) + \delta
	\end{align}
	Setting $\eta=\frac{1}{\sqrt{n}}$, we use \Cref{thm:D_H-asymptotics} in \eqref{eq:D_H-iid-bounds-2}, again applying \Cref{lem:phi-trick} to the lower bound, to obtain
	$$\ods(\rho^n\|\sigma^n) = nD(\rho\|\sigma) + \sqrt{n}\,\sk(\rho\|\sigma)\Phi^{-1}(1-\eps) + \cO(\log n).$$
	Observing that $\Phi^{-1}(1-\eps) = -\Phi^{-1}(\eps)$ now yields the result.
\end{proof}

The above results readily imply the asymptotic equipartition property (AEP) for the information spectrum relative entropies:
\begin{proposition}
	Given i.i.d.~states $\rho^{\otimes n}$ and $\sigma^{\otimes n}$, then for any $\eps\in(0,1)$ we have 
		$$\lim_{n\rightarrow\infty}\frac{1}{n}\uds(\rho^{\otimes n}\|\sigma^{\otimes n}) = D(\rho\|\sigma) = \lim_{n\rightarrow\infty}\frac{1}{n}\ods(\rho^{\otimes n}\|\sigma^{\otimes n}).$$
\end{proposition}

\subsection{Relation to the information spectrum method}
\label{sec:information-spectrum-method}
In the description of information-processing tasks, one usually assumes identical and independently distributed (i.i.d.) input states $\rho^{\otimes n}$, where $\rho$ is a state, if the given protocol is repeated $n$ times. In order to describe more general settings, two different approaches have been proposed. The first one is the smooth entropy framework of \emph{one-shot} quantum information theory, initiated by Renner \cite{Ren05,RK05,RW04} and then considerably extended by himself and others, e.g.~Datta \cite{Dat09}. This non-asymptotic approach uses the smooth entropies $\Dmine(\rho\|\sigma)$, $\Dmaxe(\rho\|\sigma)$ and $D_H^\eps(\rho\|\sigma)$ to obtain results about information-processing tasks with respect to some error $\eps$. On the other hand, Hayashi et al.~\cite{Hay06,HN03,NH07,ON00} adapted the classical information spectrum method, introduced by Han and Verdu \cite{HV93}, to the quantum setting. In \cite{BD06} the quantum information spectrum method was further developed, and \cite{DR09} eventually demonstrated that both approaches are equivalent in the sense that the spectral divergences used in the information spectrum method can be obtained as limits of the smooth entropies. 

In this section, we derive the analogous results to \cite{DR09} for the information spectrum relative entropies. Let us first recall the definitions of the spectral divergence rates:
\begin{definition}\label{def:spectral-divergence-rates}
	Let $\hrho=\lbrace\rho_n\rbrace_{n\in\mathbb{N}}$ be an arbitrary sequence of states with $\rho_n\in\cD(\cH^{\otimes n})$ and $\homega=\lbrace\omega_n\rbrace_{n\in\mathbb{N}}$ be an arbitrary sequence of positive semi-definite operators with $\omega_n\in\cP(\cH^{\otimes n})$. Then we define:
	\begin{enumerate}[{\normalfont (i)}]
		\item the \emph{quantum spectral inf-divergence rate}
			\begin{align}\label{infdv}\Du(\hrho\|\homega) \coloneqq  \sup\lbrace\gamma\mid \liminf_{n\rightarrow \infty}\tr[\lbrace\rho_n > 2^{n\gamma}\omega_n\rbrace(\rho_n-2^{n\gamma}\omega_n)] = 1\rbrace.\end{align}
		\item the \emph{quantum spectral sup-divergence rate}
			\begin{align}\label{supdv}\oD(\hrho\|\homega) \coloneqq  \inf\lbrace\gamma\mid \limsup_{n\rightarrow \infty}\tr[\lbrace\rho_n > 2^{n\gamma}\omega_n\rbrace(\rho_n-2^{n\gamma}\omega_n)] = 0\rbrace.\end{align}
	\end{enumerate}
\end{definition}
\begin{remark}
Note that the above quantities differ from the spectral divergence rates originally defined in \cite{HN03}. However, as proved in \cite{BD06}, they are equivalent.
\end{remark}
The quantum spectral divergence rates $\Du(\hrho\|\homega)$ and $\oD(\hrho\|\homega)$ can be recovered from the information spectrum relative entropies $\uds(\rho\|\sigma)$ and $\ods(\rho\|\sigma)$ respectively in the following way:
\begin{proposition}\label{prop:limit-inf-spectrum}
	Let $\hrho=\lbrace\rho_n\rbrace_{n\in\mathbb{N}}$ be an arbitrary sequence of states with $\rho_n\in\cD(\cH^{\otimes n})$ and let $\homega=\lbrace\omega_n\rbrace_{n\in\mathbb{N}}$ be an arbitrary sequence of positive operators with $\omega_n\in\cP(\cH^{\otimes n})$. Then the following relations hold:
	\begin{align*}
	\text{\normalfont (i)}\quad \lim_{\eps\rightarrow 0}\liminf_{n\rightarrow\infty}\frac{1}{n}\uds(\rho_n\|\omega_n) &= \Du(\hrho\|\homega) & 
	\text{\normalfont (ii)}\quad \lim_{\eps\rightarrow 0}\limsup_{n\rightarrow\infty}\frac{1}{n}\ods(\rho_n\|\omega_n) &= \oD(\hrho\|\homega)
	\end{align*}
\end{proposition}
\begin{proof}
(i) We first show 
\begin{align}\label{eq:inf-spec-upper-bound}
\lim_{\eps\rightarrow 0}\liminf_{n\rightarrow\infty}\frac{1}{n}\uds(\rho_n\|\omega_n) \leq \uD(\hrho\|\homega).
\end{align}
To this end, let $\gamma_n=\uds(\rho_n\|\omega_n)$ for $n\in\mathbb{N}$, and set $c(\eps)\coloneqq \liminf_{n\rightarrow\infty}\frac{\gamma_n}{n}$.
By \Cref{prop:D_s-properties}(iv) the function $c(\eps)$ is monotonically increasing in $\eps$, and hence the limit $c\coloneqq \lim_{\eps\to 0}c(\eps)$ exists in $\mathbb{R}\cup\lbrace-\infty\rbrace$.
Let us assume first that $|c|<\infty$.
It then follows from the definition of the limit that for all $\eta>0$ there exists an $\eps_0$ such that $|c-c(\eps)|\leq \eta$ holds for all $\eps\leq \eps_0$.
Moreover, by definition of the limit inferior, for all $\delta>0$ there exists an $N\in\mathbb{N}$ such that $\frac{\gamma_n}{n} >  c(\eps)-\delta$ for all $n>N$, or equivalently, 
\begin{align*}
\gamma_n > n(c(\eps)-\delta) \geq n(c-\eta-\delta)
\end{align*}
for $\eps\leq \eps_0$.
Hence, by definition of $\uds(\rho_n\|\omega_n)$, we have
\begin{align}\label{eq:lim-inf-proof-1}
\tr\left(\rho_n-2^{n(c-\eta-\delta)}\omega_n\right)_+\geq 1-\eps
\end{align}
for all $n>N$. 
Since \eqref{eq:lim-inf-proof-1} holds for arbitrarily small $\eps\leq\eps_0$, 
\begin{align*}
\liminf_{n\rightarrow\infty}\tr\left(\rho_n-2^{n(c-\eta-\delta)}\omega_n\right)_+ = 1,
\end{align*} 
which implies 
\begin{align*}
\underline{D}(\hrho\|\homega) \geq c-\eta-\delta
\end{align*}
by the definition of $\uD(\hrho\|\homega)$. 
As $\eta$ and $\delta$ were arbitrary, we obtain \eqref{eq:inf-spec-upper-bound}.

Conversely, let $\gamma=\uD(\hrho\|\homega)$. 
By the definition of the limit inferior, for all $\eps>0$ there exists an $N\in\mathbb{N}$ such that for all $n>N$ we have
\begin{align*}
\tr\left(\rho_n-2^{n\gamma}\omega_n \right)_+\geq 1-\eps.
\end{align*}
Hence, $n\gamma\leq \uds(\rho_n\|\omega_n)$ for all $n>N$ by the definition of the information spectrum relative entropy, and consequently, 
\begin{align}
\gamma\leq \lim_{\eps\rightarrow 0}\liminf_{n\rightarrow\infty}\frac{1}{n}\uds(\rho_n\|\omega_n),
\label{eq:gamma-limit-upper-bound}
\end{align} 
which yields the lower bound in the proposition.

Finally, in the case $c=-\infty$, the bound \eqref{eq:gamma-limit-upper-bound} shows that we also have $\gamma=\uD(\hrho\|\homega)=-\infty$, and hence the assertion of the proposition is trivially true.

(ii) To prove the upper bound, let $\gamma = \oD(\hrho\|\homega)$. 
By definition of $\oD(\hrho\|\homega)$ it holds that 
\begin{align*}
\limsup_{n\rightarrow\infty}\tr(\rho_n-2^{n\gamma}\omega_n)_+ = 0,
\end{align*} 
that is, for every $\eps >0$ there exists an $N\in\mathbb{N}$ such that for all $n>N$ we have
\begin{align*}
\tr(\rho_n-2^{n\gamma}\omega_n)_+ < \eps.
\end{align*}
By definition of the information spectrum relative entropy, this implies that $\ods(\rho_n\|\sigma_n)\leq n\gamma$ for all $n>N$, and hence,
\begin{align}
\lim_{\eps\rightarrow 0}\limsup_{n\rightarrow\infty}\frac{1}{n}\ods(\rho_n\|\omega_n)\leq \gamma = \oD(\hrho\|\homega),
\label{eq:gamma-limit-lower-bound}
\end{align} 
which proves the upper bound.

Conversely, let $\gamma_n= \ods(\rho_n\|\omega_n)$ and set $c(\eps) \coloneqq \limsup_{n\rightarrow \infty}\frac{\gamma_n}{n}$ and $c\coloneqq \lim_{\eps\to 0}c(\eps)$, which exists in $\mathbb{R}\cup\lbrace\infty\rbrace$ due to \Cref{prop:D_s-properties}(iv). 
If $|c|<\infty$, then by definition of the limit for all $\eta > 0$ there exists an $\eps_0$ such that $|c-c(\eps)|\leq \eta$ holds for all $\eps\leq \eps_0$.
Moreover, by the characterization of the limit superior, for all $\delta>0$ there exists an $N$ such that $\frac{\gamma_n}{n}< c(\eps) + \delta$ for all $n>N$, or equivalently, \begin{align*}
\gamma_n< n(c(\eps)+\delta) < n(c + \eta + \delta) 
\end{align*}
for all $\eps\leq\eps_0$.
This implies 
\begin{align*}
\tr(\rho_n-2^{n(c + \eta + \delta)}\omega_n)_+\leq \eps 
\end{align*}
for all $n>N$, and hence, $\limsup_{n\rightarrow \infty}\tr(\rho_n-2^{n(c + \eta + \delta)}\omega_n)_+ = 0$ as $\eps$ is arbitrarily small.
Therefore, $\oD(\hrho\|\homega)\leq c + \eta + \delta$, which yields the result, since $\delta$ and $\eta$ were arbitrary.
If $c=\infty$, then \eqref{eq:gamma-limit-lower-bound} shows that also $\gamma = \oD(\hrho\|\homega)=\infty$, in which case the assertion of the proposition is trivially true.
\end{proof}

\section{Information-processing tasks: One-shot bounds and second order asymptotics}
\label{sec:information-tasks}
After having discussed the information spectrum relative entropies and their second order asymptotics in the previous section, we now focus on their application in the following tasks in quantum information theory: 
\begin{itemize}
	\item fixed-length quantum source compression
	\item noisy dense coding
	\item entanglement concentration
	\item pure-state entanglement dilution
	\item capacity of classical-quantum (cq) channels
\end{itemize}
For the characteristic quantities in the above tasks, we obtain one-shot bounds in terms of the relative entropies $\uds(\rho\|\sigma)$ and $\ods(\rho\|\sigma)$. Furthermore, using \Cref{thm:D_s-asymptotics} we determine the second order asymptotics of the mentioned tasks, obtaining new results as well as reproducing the second order asymptotics for the capacity of cq-channels from \cite{TT13}. 

We employ the following useful lemma due to Hayashi and Nagaoka \cite{HN03}:
\begin{lemma}\label{lemHN}
Let ${\cal X}$ denote a finite alphabet and consider a classical-quantum channel $W\colon {\cal X} \mapsto {\cal D}({\cal H})$, i.e., $W_x \in {\cal D}(\cH)$ is the output of the channel when the input is $x \in {\cal X}$. Then for all $n \in \mathbb{N}$, $\gamma \in \mathbb{R}$, $M \in \mathbb{N}$, a probability distribution $\{p(x)\}_{x \in {\cal X}}$ on ${\cal X}$ and $c >0$ there exists a code ${\cal C}$ such that $|{\cal C}| = M$, and
\begin{align}\label{eqHN}
p_e({\cal C}) \le (1 +c) \sum_{x \in {\cal X}} p(x) \tr \left[W_x \lbrace W_x < 2^{\gamma} \oW\rbrace \right] + (2+c+c^{-1}) 2^{-\gamma} M,
\end{align}
where $\oW = \sum_{x \in {\cal X}} p(x) W_x$.
\end{lemma}
We also use the following notation:
\begin{definition}\label{def:blocklength-quantity}
	Let $\Theta^{(1), \eps}(\Omega)$ be a quantity characterizing an information-processing task using a resource $\Omega$ in the one-shot setting for an error threshold $\eps$. For $n\in\mathbb{N}$, we write 
	\begin{align*}
	\Theta^{n,\eps}(\Omega) \coloneqq \Theta^{(1), \eps}(\Omega^{\otimes n}).
	\end{align*} 
	Here, $\Omega$ is a placeholder for the resource of the protocol, e.g.~a source state $\rho\in\cD(\cH)$ (source coding), a channel $\Lambda$, or an entangled pure state $\psi_{AB}\in\cD(\cH_{AB})$ (entanglement conversion).
\end{definition}

\subsection{Fixed-length quantum source coding}
\subsubsection{Blind vs.~visible coding}\label{sec:blind-vs-visible}
Consider a quantum information source characterized by an ensemble of pure states  $\kE=\lbrace p_i,\psi_i\rbrace$ where $|\psi_i\rangle\in\cH$ and $\lbrace p_i\rbrace$ being a probability distribution. Note that the pure states $\psi_i$ are not mutually orthogonal in general. Henceforth we shall refer to such an ensemble as the {\em{source ensemble}} and the associated density matrix
 $\rho=\sum_ip_i\psi_i\in\cD(\cH)$ as the {\em{source state}}. The pure states
$\psi_i$ are the signals emitted by the source with corresponding probabilities $p_i$. 

In fixed-length quantum source coding the aim is to store the information emitted by the source in a {\em{compressed state}} $\rho_c \in \cD(\cH_c)$ with $\dim\cH_c < \dim\cH$, such that it can be later decompressed to yield a state which is sufficiently close to the source state $\rho$, with respect to some chosen distance measure. In the one-shot setting, in which one considers a single use 
of the source, it is natural to allow a non-zero error in the compression-decompression scheme. Hence, a one-shot source-coding protocol is characterized by 
a parameter $\eps$ which denotes the maximum allowed value of the distance between the source state and the state which is obtained after decompression.

There are two different scenarios \cite{Hay02,BCF+01,Win99a} for the compression part of the protocol outlined above: In the first scenario, which is called the {\em{blind setting}}, the encoder (Alice) knows the source state $\rho$ but has no knowledge about the individual signals $\psi_i$. Hence, in this case the compression map is necessarily a quantum operation $\cE\colon\cD(\cH)\rightarrow\cD(\cH_c)$ on the source state. This is called \emph{blind source coding}. 

In contrast, in the second scenario, Alice has complete knowledge about the
signals $\psi_i$ as well as their corresponding probabilities $p_i$. This is 
the {\em{visible setting}}. In fact, Alice essentially has classical information about the pure-state ensemble $\kE=\lbrace p_i,\psi_i\rbrace$ in the form of a list of the probabilities $p_i$ and corresponding pure states $\psi_i$. She then uses an \emph{arbitrary} map $\cV\colon\lbrace i\rbrace_i\rightarrow \cD(\cH_c)$ to encode the $\psi_i$ in a state $\cV(i)\in\cD(\cH_c)$ and forming the average ensemble state $\rho_c = \sum_i p_i\cV(i)$. Note that the encoding map $\cV$ is not necessarily a quantum operation, since Alice simply prepares a collection of quantum states $\cV(i)$ based on her knowledge of the signal states $\psi_i$. In particular, $\cV$ can be a non-linear map. This is called {visible source coding}.

In the decompression phase of the protocol, 
the compressed state $\rho_c = \cE(\rho)$ (blind setting) respectively $\rho_c=\sum_ip_i\cV(i)$ (visible setting) is subjected to a quantum operation $\mathfrak{D}\colon\cD(\cH_c)\rightarrow\cD(\cH)$.

Interestingly, the optimal rates of both visible and blind source-coding are identical, given by the von Neumann entropy $S(\rho)$ of the ensemble state $\rho=\sum_ip_i\psi_i$ \cite{Sch95,BFJS96,Hor98}. In the following, we give a complete description of the second order behaviour of visible source-coding and derive bounds on the second order asymptotic expansion in the blind case. In both settings, the optimal rate $S(\rho)$ is retrieved.
We derive one-shot bounds on the {\em{minimal compression length}}, which is the characteristic quantity of fixed-length quantum source coding and is defined below. These bounds are given in terms of the information spectrum entropy. 
We then derive second order expansions of these one-shot bounds for sufficiently large $n$.

\subsubsection{One-shot bounds}\label{sec:one-shot-bounds}
Let us first define the figures of merit that we use to determine the distance between the input state and the target state of the source-coding protocol:
\begin{definition}~\label{def:figures-of-merit}
\begin{enumerate}[{\normalfont (i)}]
\item For an arbitrary CPTP map $\Lambda\colon\cD(\cH)\rightarrow\cD(\cH)$ and a state $\rho\in\cD(\cH)$ with purification  $\psi^\rho\in\cD(\cH\otimes\cH')$, the \emph{entanglement fidelity} $F_e(\rho,\Lambda)$ is defined by
$$F_e(\rho,\Lambda)\coloneqq  \langle\psi^\rho|(\id_{\cH'}\otimes\Lambda)(\psi^\rho)|\psi^\rho\rangle.$$
\item Let $\kE=\lbrace p_i,\psi_i\rbrace$ with $\psi_i\in\cD(\cH)$ be a pure-state ensemble and consider an arbitrary map $\cV\colon\lbrace i\rbrace_i\rightarrow \cD(\cH')$ and a CPTP map $\kD\colon\cD(\cH')\rightarrow \cD(\cH)$. The \emph{ensemble average fidelity} $\bar{F}(\kE,\kD\circ\cV)$ is defined by
\begin{align*}
\bar{F}(\kE,\kD\circ\cV) \coloneqq \sumi_i p_i\tr((\kD\circ\cV)(i)\psi_i).
\end{align*}
\end{enumerate}
\end{definition} 

We employ the entanglement fidelity and the ensemble average fidelity to define $\eps$-admissible codes in the blind and visible settings respectively. The minimal compression length for a code is then obtained by optimizing over all such codes:
\begin{definition}\label{def:source-coding}
Consider a quantum information source with source ensemble $\kE=\lbrace p_i,\psi_i\rbrace$, and source state $\rho=\sum_i p_i\psi_i \in \cD(\cH)$.  Let $\cH_c$ be a Hilbert space with $M=\dim\cH_c<\dim\cH$.
\begin{enumerate}[(i)]
\item  For any $\eps\in(0,1)$, a visible encoder $\cV\colon\lbrace i\rbrace_i\rightarrow\cD(\cH_c)$ and a CPTP decoder $\kD\colon\cD(\cH_c)\rightarrow\cD(\cH)$, the triple $\cC_v=(\cV,\kD,M)$ is called an $\eps$-admissible (visible) code if $$\bar{F}(\kE,\kD\circ\cV)\geq 1-\eps.$$

\item For any $\eps\in(0,1)$, a CPTP encoder $\cE\colon\cD(\cH)\rightarrow\cD(\cH_c)$ and a CPTP decoder $\kD\colon\cD(\cH_c)\rightarrow\cD(\cH)$, the triple $\cC_b=(\cE,\kD,M)$ is called an $\eps$-admissible (blind) code if $$F_e(\rho,\kD\circ\cE)\geq 1-\eps.$$

\item The \emph{$\eps$-error one-shot compression length} $m_{v/b}^{(1), \eps}(\rho)$ is defined by
	\begin{align*}
	m_{v/b}^{(1), \eps}(\rho) \coloneqq  \inf\lbrace \log M\mid \exists \text{ an }\eps\text{-admissible code }\cC_{v/b}\text{ with }M=\dim\cH_c\rbrace,
	\end{align*}
	where the subscripts $v$ and $b$ are used for the visible and blind setting respectively.
\end{enumerate}
\end{definition}

In the following, we derive one-shot bounds for the minimal compression length.
\begin{theorem}\label{thm:source-coding-one-shot}
Consider a quantum information source with source ensemble $\kE=\lbrace p_i,\psi_i\rbrace$, and source state $\rho=\sum_i p_i\psi_i \in \cD(\cH)$.
Then for any $\eta>0$ and $\eps\in(0,1)$, the $\eps$-error one-shot compression length $m_{v/b}^{(1), \eps}$ satisfies the following bounds:
\begin{enumerate}[{\normalfont (i)}]
\item Visible case: $\oH_s^{\eps+\eta}(\rho)+\log\eta \leq m_v^{(1), \eps}(\rho)\leq \oH_s^{\eps}(\rho)$
\item Blind case: $\oH_s^{\eps+\eta}(\rho)+\log\eta \leq m_b^{(1), \eps}(\rho)\leq \oH_s^{\eps/2}(\rho)$
\end{enumerate}
\end{theorem}
Interestingly, the lower bounds are identical in the two settings. However,
 the upper bounds are given in terms of the information spectrum entropies with different error parameters. 
This has consequences on the second order behaviour of $m_{b}^{(1), \eps}(\rho)$, which we discuss in Section \ref{sec:source-coding-second-order}.

\subsubsection{Proof of one-shot bounds: Visible setting}
\begin{proof}[Proof of the upper bound in \Cref{thm:source-coding-one-shot}(i) (Achievability)]
	Assume that 
	\begin{align}\label{eq:visible-comp-rate-assumption}
		\gamma=\oH_s^\eps(\rho).
	\end{align}
We set $P=\lbrace \rho\ge 2^{-\gamma}\one\rbrace$ and define the compressed Hilbert space as the image of this projection, i.e.~ $\cH_c\coloneqq \im P$. The dimension of $\cH_c$ in this case is given by
\begin{align}
M= \tr P \le 2^\gamma,\label{eq:dimension-compressed-space}
\end{align}
where the inequality follows from Lemma~\ref{lem3}. Consider now the visible encoding map 
\begin{align*}
\cV\colon i\mapsto \frac{P\psi_i P}{\tr P\psi_i}.
\end{align*}
For the decoding map $\kD$ we choose the trivial embedding of the compressed state $\rho_c = \sum_i p_i \cV(i)$ into the original Hilbert space $\cH$. We can verify that this choice of $\cH_c$, $\cV$ and $\kD$ constitutes an $\eps$-admissible code $\cC_v$ as follows:
	\begin{align*}
		\bar{F}(\kE,\kD\circ\cV) &= \sumi_i p_i\tr((\kD\circ\cV)(i)\psi_i)\\
		&= \sumi_i p_i\tr\frac{P\psi_i P\psi_i}{\tr P\psi_i}\\
		&= \sumi_i p_i \frac{\langle\psi_i|P|\psi_i\rangle^2}{\langle\psi_i|P|\psi_i\rangle}\\
		&= \sumi_i p_i \langle \psi_i|P|\psi_i\rangle\\
		&= \tr P\rho\\
		&\geq \tr [P(\rho-2^{-\gamma}\one)]\\
		&\geq 1-\eps,
	\end{align*}
	where the last inequality follows from the definitions of $P$ and $\oH_s^\eps(\rho)$, and the choice of $\gamma$ given in \eqref{eq:visible-comp-rate-assumption}. Thus, $(\cV, \kD, M)$ is an $\eps$-admissible code and we obtain $m_v^{(1), \eps}(\rho)\leq \gamma= \oH_s^\eps(\rho)$ by \eqref{eq:dimension-compressed-space} and \Cref{def:source-coding}(iii) of the minimal compression length $m_v^{(1), \eps}(\rho)$. 
\end{proof}

In order to prove the lower (converse) bound of \Cref{thm:source-coding-one-shot}(i), we employ the following lemma \cite{NK01,Hay02}:
\begin{lemma}\label{lem:separable-states}
Let $\rho_{AB}\in\cD(\cH_{AB})$ be separable. Then
\begin{multline*}
\max\lbrace \tr P\rho_A \mid P\text{ is a projection on }\cH_A\text{ with }\tr P=k\rbrace \\
\geq \max\lbrace \tr P\rho_{AB} \mid P\text{ is a projection on }\cH_{AB}\text{ with }\tr P=k\rbrace.
\end{multline*}
\end{lemma}

The following crucial proposition is due to Hayashi \cite{Hay02}. We reproduce a proof in our notation for the convenience of the reader.

\begin{proposition}\label{prop:averaged-fidelity-bound}
Consider a quantum information source with source ensemble $\kE=\lbrace p_i,\psi_i\rbrace$, and source state $\rho=\sum_i p_i\psi_i \in \cD(\cH)$.
Let $\cV\colon \lbrace i\rbrace_i \rightarrow \cH_c$ be a visible encoding map into a compressed Hilbert space $\cH_c$, and $\kD\colon\cH_c\rightarrow\cH$ be the decompression map. Set $\rho=\sum_ip_i\psi_i$ and denote by $M=\dim\cH_c$ the dimension of the compressed Hilbert space $\cH_c$. Then
$$\bar{F}(\kE,\kD\circ\cV)\leq \max\lbrace \tr P\rho\mid P\text{ is a projection on }\cH\text{ with }\tr P=M\rbrace.$$
\end{proposition}

\begin{proof}
By considering eigenvalue decompositions $\cV(i) = \sum_j \mu_j^{(i)}\varphi_j^{(i)}$ of the encoded signals, it suffices to show that
\begin{align}
\sum_i p_i \tr(\kD(\chi)\psi_i) \leq \max\lbrace \tr P\rho\mid P\text{ is a projection on }\cH\text{ with }\tr P=M\rbrace\label{eq:hayashi-pure-state}
\end{align}
holds for any pure state $|\chi\rangle\in\cH_c$.

To show \eqref{eq:hayashi-pure-state}, let $\cU\colon\cH_c\rightarrow\cH\otimes\cH'$ be a Stinespring isometry of the CPTP map $\kD$, i.e.
$$\kD(\sigma) = \tr_{\cH'}(\cU\sigma\cU^\dagger)$$ for $\sigma\in\cD(\cH_c)$. Then $\cU\cU^\dagger$ is the projection onto the image of $\cU$ in $\cH\otimes\cH'$. Consider now the pure state
\begin{align}
\psi_i' \coloneqq \frac{(\psi_i\otimes\one_{\cH'})\cU\chi\cU^\dagger(\psi_i\otimes\one_{\cH'})}{\tr(\cU\chi\cU^\dagger(\psi_i\otimes\one_{\cH'}))}\in\cD(\cH\otimes\cH'),\label{eq:psi'-definition}
\end{align} which satisfies
\begin{align}
\tr (\cU\chi\cU^\dagger)\psi_i' &= \frac{\tr(\cU\chi\cU^\dagger(\psi_i\otimes\one_{\cH'})\cU\chi\cU^\dagger(\psi_i\otimes\one_{\cH'}))}{\tr(\cU\chi\cU^\dagger(\psi_i\otimes\one_{\cH'}))}\notag\\
&= \frac{\langle\chi|\cU^\dagger(\psi_i\otimes\one_{\cH'})\cU|\chi\rangle^2}{\langle\chi|\cU^\dagger(\psi_i\otimes\one_{\cH'})\cU|\chi\rangle}\notag\\
&= \langle\chi|\cU^\dagger(\psi_i\otimes\one_{\cH'})\cU|\chi\rangle\notag\\
&= \tr(\cU\chi\cU^\dagger(\psi_i\otimes\one_{\cH'}))\notag\\
&=\tr \kD(\chi)\psi_i.\label{eq:psi'-trace-relation}
\end{align}
We claim that 
\begin{align}
\tr_{\cH'}\psi_i'=\psi_i\label{eq:psi'-psi}.
\end{align} 
This can be seen as follows.
Firstly, for the denominator of \eqref{eq:psi'-definition} we have
\begin{align}
\tr(\cU\chi\cU^\dagger(\psi_i\otimes\one_{\cH'})) = \tr(\kD(\chi)\psi_i) = \langle \psi_i|\kD(\chi)|\psi_i\rangle.\label{eq:partial-denominator}
\end{align}
Secondly, the numerator of \eqref{eq:psi'-definition} can be rewritten as
\begin{align*}
(\psi_i\otimes \one_{\cH'})\cU\chi\cU^\dagger(\psi_i\otimes\one_{\cH'}) &= \left(|\psi_i\rangle\langle\psi_i|\otimes \sumi_j|j\rangle\langle j|\right))\cU\chi\cU^\dagger \left(|\psi_i\rangle\langle\psi_i|\otimes\sumi_k|k\rangle\langle k|\right)\notag\\
&=\sum_{j,k} \Big[(\langle\psi_i|\otimes \langle j|)\cU\chi\cU^\dagger(|\psi_i\rangle\otimes |k\rangle)\Big] |\psi_i\rangle\langle\psi_i|\otimes| j\rangle\langle k|
\end{align*}
where $\lbrace |j\rangle\rbrace$ constitutes an orthonormal basis for $\cH'$. Hence,
\begin{align}
\tr_{\cH'}\left((\psi_i\otimes \one_{\cH'})\cU\chi\cU^\dagger(\psi_i\otimes\one_{\cH'})\right) &=\tr_{\cH'}\bigg(\sumi_{j,k} \Big[(\langle\psi_i|\otimes \langle j|)\cU\chi\cU^\dagger(|\psi_i\rangle\otimes |k\rangle)\Big]\notag\\
&\qquad{}\times |\psi_i\rangle\langle\psi_i|\otimes| j\rangle\langle k|\bigg)\notag\\
&= \sum_{j} \Big[(\langle\psi_i|\otimes \langle j|)\cU\chi\cU^\dagger(|\psi_i\rangle\otimes |j\rangle)\Big] |\psi_i\rangle\langle\psi_i|\notag\\
&= \langle\psi_i|\kD(\chi)|\psi_i\rangle \psi_i,\label{eq:partial-nominator}
\end{align}
where the last equality follows from the definition of the partial trace. Using \eqref{eq:partial-denominator} and \eqref{eq:partial-nominator} in \eqref{eq:psi'-definition} yields \eqref{eq:psi'-psi}.

We now want to relate the state $\psi_i'$ to the source state $\rho\coloneqq \sum_ip_i\psi_i$ associated to $\kE$. Since $\psi'_i$ is a purification of the pure state $\psi_i$ by \eqref{eq:psi'-psi}, there exists a pure state $\phi_i\in\cD(\cH')$ such that $\psi_i'=\psi_i\otimes \phi_i$. Consequently, 
the state $$\rho'\coloneqq \sumi_i p_i\psi_i'=\sumi_i p_i \psi_i\otimes\phi_i$$ is separable and satisfies $\tr_{\cH'}\rho'=\rho$. Thus, we can apply \Cref{lem:separable-states} to infer that
\begin{multline}\label{eq:projectors}
\max\lbrace \tr P\rho' \mid P\text{ is a projection on }\cH\otimes\cH'\text{ with }\tr P=M\rbrace \\
\leq \max\lbrace \tr P\rho \mid P\text{ is a projection on }\cH\text{ with }
\tr P=M\rbrace,
\end{multline}
and we observe that the projector $\cU\cU^\dagger$ with $$\tr\cU\cU^\dagger=\tr\cU^\dagger\cU = \tr\one_{\cH_c}=M$$ is feasible for the left-hand side of \eqref{eq:projectors}. Noting that 
\begin{align}
\chi\leq \one_{\cH_c}\quad\text{and thus}\quad\cU\chi\cU^\dagger\leq\cU\cU^\dagger,\label{eq:isometry-projector-relation}
\end{align}
we finally arrive at
\begin{align*}
\sum_i p_i\tr(\kD(\chi)\psi_i) &= \sumi_i p_i\tr((\cU\chi\cU^\dagger)\psi_i')\\
&\leq \sumi_i p_i \tr(\cU\cU^\dagger\psi_i'	)\\
&= \tr(\cU\cU^\dagger\rho')\\
&\leq \max\lbrace \tr P\rho' \mid P\text{ is a projection on }\cH\otimes\cH'\text{ with }\tr P=M\rbrace \\
&\leq \max\lbrace \tr P\rho \mid P\text{ is a projection on }\cH\text{ with }\tr P=M\rbrace
\end{align*}
where the first line follows from \eqref{eq:psi'-trace-relation}, the second line uses \eqref{eq:isometry-projector-relation} and the last line is obtained by applying \eqref{eq:projectors}.
\end{proof}

\begin{proof}[Proof of the lower bound in \Cref{thm:source-coding-one-shot}(i) (Converse)]
Let $\cV\colon\lbrace i\rbrace_i\rightarrow \cH_c$ be a visible encoding of the source ensemble $\kE$ into the compressed Hilbert space $\cH_c$ with $\dim\cH_c = M$, and let $\kD\colon\cH_c\rightarrow \cH$ be an arbitrary CPTP map used in the decoding operation. Assuming that $$\bar{F}(\kE,\kD\circ\cV)\geq 1-\eps,$$ we have to show that 
\begin{align*}
\log M \geq \oH_s^{\eps+\eta}(\rho)+\log\eta.
\end{align*}
To this end, we compute:
\begin{align*}
1-\eps &\leq \bar{F}(\kE,\kD\circ\cV)\\
&\leq \max\lbrace \tr P\rho\mid P\text{ is a projection on }\cH\text{ with }\tr P=M\rbrace\\
&= \tr Q\rho\\
&= \tr Q(\rho-2^{-\log M + \log\eta}\one) + 2^{-\log M + \log\eta} \tr Q\\
&\leq \tr(\rho-2^{-\log M + \log\eta}\one)_+ + \eta
\end{align*}
where the second inequality follows from \Cref{prop:averaged-fidelity-bound}, with the maximizing projection denoted by $Q$ satisfying $\tr Q = M$, and the last inequality uses \Cref{lem:tr-projector}. We have thus shown that, for arbitrary $\eta>0$ 
$$1-(\eps+\eta) \leq \tr(\rho-2^{-\log M + \log\eta}\one)_+,$$ and hence $\oH_s^{\eps+\eta}(\rho)\leq \log M - \log\eta$, which proves the claim.
\end{proof}

\subsubsection{Proof of one-shot bounds: Blind setting}
\begin{proof}[Proof of the upper bound in \Cref{thm:source-coding-one-shot}(ii) (Achievability)]
	Assume that 
	\begin{align}\label{eq:comp-rate-assumption}
		\gamma=\oH_s^\eps(\rho).
	\end{align}
We set $P=\lbrace \rho\geq 2^{-\gamma}\one\rbrace$ and consider the compression map 
	\begin{align}
	\cE\colon\rho\mapsto P\rho P + \tr[\rho(\one-P)] |\varphi\rangle\langle \varphi|,\label{eq:compression-map}
	\end{align}
	where $|\varphi\rangle$ is an arbitrary pure state in the compressed Hilbert space $\cHt_c\equiv \im P$ with dimension
$M= \tr P \le 2^\gamma.$
Here, the inequality follows from Lemma~\ref{lem3}. For the decoding map $\kD$ we consider the trivial embedding of the compressed state $\cE(\rho)$ into the original Hilbert space $\cH_A$. It is known \cite{Sch96} that for a CPTP map $\Lambda$ with Kraus operators $\lbrace A_i\rbrace$ the entanglement fidelity is given by the expression $$F_e(\rho,\Lambda) = \sum_i |\tr A_i\rho|^2.$$ 
The Kraus form of the compression map $\cE$ as in \eqref{eq:compression-map} is given by
$$\cE(\rho) = P\rho P + \sumi_iA_i\rho A_i^\dagger,$$ where $A_i=|\phi\rangle\langle i|$ and $\lbrace |i\rangle\rbrace$ is an orthonormal basis for $(\im P)^\perp$. We therefore have
	\begin{align*}
		F_e(\rho,\Lambda) &= |\tr P\rho|^2 + \sumi_i|\tr (|0\rangle\langle i|\rho)|^2\\
		&\geq (\tr P\rho)^2\\
		&\geq (\tr [P(\rho-2^{-\gamma}\one)])^2\\
		&= (1-\eps)^2\\
		&\geq 1-2\eps,
	\end{align*}
	where the last equality follows from the definitions of $P$ and $\oH_s^\eps(\rho)$, and the choice of $\gamma$ given in \eqref{eq:comp-rate-assumption}. Thus, $(\cE, \kD, M)$ is an $2\eps$-admissible code, and by replacing $\eps$ with $\frac{\eps}{2}$ we obtain $m_b^{(1), \eps}(\rho)\leq \gamma= \oH_s^{\eps/2}(\rho)$ by \Cref{def:source-coding}(iii) of the compression length $m_b^{(1), \eps}$. 
\end{proof}

For the proof of the lower bound (converse) in \Cref{thm:source-coding-one-shot}(ii), we employ the Kraus representation of the CPTP encoding map $\cE$ as follows:
\begin{proof}[Proof of the lower bound in \Cref{thm:source-coding-one-shot}(ii) (Converse)]
	We need to prove that for every $\eps$-admissible (blind) coding scheme $\cC_b=(\cE,\kD,M)$ with $F_e(\rho,\kD\circ\cE)\geq 1-\eps$, the code size $M$ satisfies
	$$\log M \geq \oH_s^{\eps+\eta}(\rho)+\log\eta.$$
		
	To this end, let $\cH_c\subset \cH$ be the compressed Hilbert space with dimension $M=\dim\cH_c$, let $\cE\colon\cD(\cH)\rightarrow\cD(\cH_c)$ be a CPTP map, and let $P$ be the projector onto $\cH_c$, i.e.~$\tr P=M$. Furthermore, let $\kD\colon\cD(\cH_c)\rightarrow\cD(\cH)$ be an arbitrary CPTP map. If $\lbrace E_j\rbrace$ and $\lbrace D_k\rbrace$ are sets of Kraus operators for $\cE$ and $\kD$ respectively, then $\lbrace D_kE_j\rbrace$ is a set of Kraus operators for the map $\Lambda = \kD\circ\cE$. We also define $Q_k$ as the projector onto the image of $D_k$, i.e.~$Q_k = \Pi_{\im(D_k P)}$, and note that 
	\begin{align}\label{eq:Qk-Dk}
		D_kE_j=Q_kD_kE_j.
	\end{align}	
		For all $k$ we have
	\begin{align}
		\tr Q_k = \dim\im(D_k P) =\dim\im D_k|_{\im P}\leq \dim\dom D_k|_{\im P}=\dim\im P=\tr P.\label{eq:tr-Q-tr-P}
	\end{align}
	Let us first derive an upper bound for the entanglement fidelity:
	\begin{align*}
		F_e(\rho,\Lambda) &= \sum_{j,k}|\tr(D_kE_j\rho)|^2\\
		&= \sum_{j,k}|\tr(Q_kD_kE_j\rho)|^2\\
		&= \sum_{j,k}|\tr(\sqrt{\rho}Q_kD_kE_j\sqrt{\rho})|^2\\
		&\leq \sum_{j,k}\tr(Q_k\rho Q_k)\tr(\sqrt{\rho}E_j^\dagger D_k^\dagger D_k E_j\sqrt{\rho})\\
		&= \sum_{j,k}\tr(Q_k\rho Q_k)\tr(D_k E_j\rho E_j^\dagger D_k^\dagger)\\
	\end{align*}
	where we used \eqref{eq:Qk-Dk} in the second equality and the Cauchy-Schwarz inequality for the Hilbert-Schmidt inner product to obtain the inequality. Now let $\widetilde{Q}=\arg\max_{Q_k}\tr(Q_k\rho Q_k)$ and $\eta>0$ arbitrary, then we can further bound the entanglement fidelity by
	\begin{align}
		F_e(\rho,\Lambda) &\leq \tr(\tq\rho\tq)\sum_{j,k}\tr(D_k E_j\rho E_j^\dagger D_k^\dagger)\notag\\
		&= \tr\tq\rho\underbrace{\tr\left(\sum\nolimits_k D_k\left(\sum\nolimits_j E_j\rho E_j^\dagger\right)D_k^\dagger\right)}_{=\tr(\Lambda(\rho))=1}\notag\\
		&= \tr[\tq(\rho-2^{-\log M + \log\eta}\one)]+ 2^{-\log M + \log\eta}\tr \tq\notag\\
		&\leq\tr(\rho-2^{-\log M + \log\eta}\one)_+ + 2^{-\log M + \log\eta}\tr P\notag\\
		&= \tr(\rho-2^{-\log M + \log\eta}\one)_+ + \eta\label{eq:converse-bound}
	\end{align}
	where we used \Cref{lem:tr-projector} and \eqref{eq:tr-Q-tr-P} in the second inequality. Hence, \eqref{eq:converse-bound} together with the assumption $F_e(\rho,\kD\circ\cE)\geq 1-\eps$ imply that
	\begin{align*}
		1-(\eps+\eta) &\leq \tr(\rho-2^{-\log M + \log\eta}\one)_+.
	\end{align*}
	By definition of the information spectrum entropy, we infer that $\oH_s^{\eps+\eta}\leq \log M -\log\eta$.
\end{proof}

\subsubsection{Second order asymptotics of source coding}\label{sec:source-coding-second-order}
The one-shot bounds for the optimal source compression length in terms of the information spectrum entropy in \Cref{thm:source-coding-one-shot} together with the asymptotic expansion of the latter in \Cref{thm:D_s-asymptotics} readily yield the second order asymptotic expansion of the visible coding compression length $m_v^{n,\eps}(\rho)$, which is related to the one-shot compression length via \Cref{def:blocklength-quantity}. In the blind setting, \Cref{thm:D_s-asymptotics} provides second order asymptotic bounds for $m_b^{n,\eps}(\rho)$.
\begin{theorem}\label{thm:sc-second-order}
	Consider a memoryless quantum source characterized by the pure-state ensemble $\kE=\lbrace p_i,\psi_i\rbrace$ with $\psi_i\in\cD(\cH)$ and the associated average ensemble state $\rho=\sum_ip_i\psi_i$. For any $\eps\in(0,1)$, the following holds for $m_{v/b}^{n,\eps}(\rho)$:
	\begin{enumerate}[{\normalfont (i)}]
	\item In the visible setting, the second order asymptotic expansion of  $m_{v}^{n,\eps}(\rho)$ is given by
	\begin{align}
		m_v^{n,\eps}(\rho) = nS(\rho)-\sqrt{n\left(\tr[\rho(\log\rho)^2]-S(\rho)^2\right)}\,\Phi^{-1}(\eps)+\cO(\log n).\label{eq:sc-second-order}
		\end{align}
		\item In the blind setting, we obtain the following asymptotic bounds:
		\begin{align}\label{eq:blind-sc-second-order}
			\begin{aligned}
			m_b^{n,\eps}(\rho) &\geq nS(\rho)-\sqrt{n\left(\tr[\rho(\log\rho)^2]-S(\rho)^2\right)}\,\Phi^{-1}(\eps)+\cO(\log n)\\
			m_b^{n,\eps}(\rho) &\leq nS(\rho)-\sqrt{n\left(\tr[\rho(\log\rho)^2]-S(\rho)^2\right)}\,\invP{\frac{\eps}{2}}+\cO(\log n).\\
			\end{aligned}
			\end{align}
	\end{enumerate}	
\end{theorem}
\begin{proof}
	We abbreviate $\rho^n\equiv\rho^{\otimes n}$. By \Cref{thm:source-coding-one-shot}(i) we obtain the following bounds for the compression length in the visible setting with block length $n$:
	$$-\uD_s^{\eps+\eta}(\rho^{n}\|\one_n) + \log\eta \leq m_v^{n,\eps}(\rho)\leq -\uD_s^\eps(\rho^n\|\one_n).$$
	Setting $\eta=\frac{1}{\sqrt{n}}$, \Cref{thm:D_s-asymptotics} implies
	\begin{align*}
		m_v^{n,\eps}(\rho)&\geq -nD(\rho\|\one) - \sqrt{n}\,\sk(\rho\|\one)\Phi^{-1}\left(\eps+\frac{1}{\sqrt{n}}\right) + \cO(\log n)\\
		m_v^{n,\eps}(\rho) &\leq -nD(\rho\|\one) - \sqrt{n}\,\sk(\rho\|\one)\Phi^{-1}(\eps) + \cO(\log n).
	\end{align*}
	We apply \Cref{lem:phi-trick} to the lower bound and use $S(\rho)=-D(\rho\|\one)$ to obtain 
	$$m_v^{n,\eps}(\rho) = nS(\rho) - \sqrt{n}\,\sk(\rho\|\one)\Phi^{-1}(\eps) + \cO(\log n).$$
	Expanding the term $\sk(\rho\|\one)$ gives
	\begin{align*}
		\sk(\rho\|\one) &= \sqrt{\tr[\rho_A(\log\rho - \log\one)^2]-D(\rho\|\one)^2}\\
		&= \sqrt{\tr[\rho(\log\rho)^2]-S(\rho)^2},
	\end{align*}
	which yields the result. The bounds for the blind setting follow analogously.
\end{proof}

\begin{remark}~
\begin{enumerate}[(i)]
\item We note that the `gap' in the second order asymptotic bounds \eqref{eq:blind-sc-second-order} for $m_b^{(1),\eps}(\rho)$ is due to the different error parameters $\eps$ and $\frac{\eps}{2}$ in \Cref{thm:source-coding-one-shot}(ii). Hence, our method only yields non-matching asymptotic bounds in the blind setting. 
This problem has since been resolved by Abdelhadi and Renes \cite{AR20}, who derived a closed form of the second order asymptotic expansion of blind quantum source coding.

\item Using the result in \Cref{thm:sc-second-order}, we are able to recover the second order asymptotics of classical fixed-length source coding derived by Hayashi (see Theorems 3 and 9 of \cite{Hay08}).

\item Note that $\invP{\eps}>0$ for $\eps>\frac{1}{2}$, and therefore, the second order term in \eqref{eq:sc-second-order} is strictly negative. In this case, the compression rate to second order drops below the von Neumann entropy $S(\rho)$, as illustrated in Figure \ref{fig:below}, since the former is given by $\frac{an+b\sqrt{n}}{n}$ where $a=S(\rho)$ and $b=-\sqrt{\left(\tr[\rho(\log\rho)^2]-S(\rho)^2\right)}\,\Phi^{-1}(\eps)$.
\begin{figure}[ht]
\centering
\includegraphics{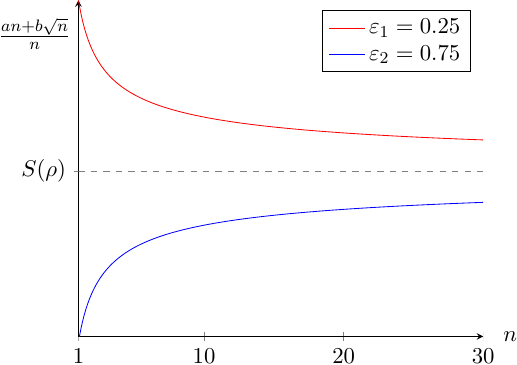}
\caption{Plot of $\frac{an+b\sqrt{n}}{n}$ for the source state $\rho=\frac{1}{2}|0\rangle\langle 0|+\frac{1}{2}|+\rangle\langle +|$.}
\label{fig:below}
\end{figure}
\end{enumerate}
\end{remark}

\subsection{Noisy dense coding}
Dense coding is the protocol by which prior shared entanglement between a sender (Alice) and a receiver (Bob) is exploited for sending classical messages through a noiseless quantum channel. If the entanglement shared between Alice and Bob is in the form of a mixed state $\rho_{AB} \in {\cal D}(\cH_A \otimes \cH_B)$ (instead of a maximally entangled pure state), then the entanglement is said to be noisy.

Our aim is to derive second order asymptotics for the optimal rate of reliable transmission of classical information through a noiseless quantum channel, assisted by prior shared noisy entanglement. If the noisy entanglement is in the form of multiple, identical copies of a bipartite state $\rho_{AB}$, then the optimal rate in the asymptotic limit is referred to as the {\em{dense coding capacity}} of the state $\rho_{AB}$ and denoted by $C_{dc}(\rho_{AB})$. 

In \cite{HHH+01}, it was shown that $$C_{dc}(\rho_{AB}) = \log d + S(\rho_B) - \inf_n\inf_{\Lambda^{(n)}}\frac{1}{n} S((\Lambda^{(n)}\otimes\id_{B^n})\rho_{AB}^{\otimes n}),$$ where $\Lambda^{(n)}$ is a CPTP map acting on states on $\cH_{A}^{\otimes n}$ and $d=\dim\cH_A$. If only product encodings of the form $\Lambda^{\otimes n}$ are allowed, the above expression reduces to 
\begin{align}
C_{dc}(\rho_{AB}) = \log d + S(\rho_B) - \inf_{\Lambda:\text{ CPTP}} S((\Lambda\otimes\id_B)\rho_{AB}),\label{eq:dense-coding-first-order}
\end{align}
which was derived independently by Winter in \cite{Win02}. 

In the present work, we restrict ourselves to the latter case. As a first step, we derive upper and lower bounds on the maximum number of bits of a classical message which can be transmitted through a single use of the channel, with a probability of error at most $\eps$, for a given $0<\eps <1$. We then discuss the i.i.d.~case under the assumption of a product encoding. In the following, we briefly summarize the coding procedure:

Let $\rho_{AB}$ be the noisy entangled state that Alice shares with Bob, the system $A$ being with Alice and $B$ being with Bob. Suppose Alice has a set of classical messages labelled by the elements of the set ${\cal M} \coloneqq  \{1,2 ,\ldots, M\}$, which she wishes to send to Bob through a noiseless quantum channel. The most general protocol for this consists of an encoding map by Alice on her system A, transmission of the encoded state through the noiseless channel to Bob, followed by a decoding operation by Bob on the joint state of the system that he receives and the system $B$.

The codewords are given by 
$$\varphi(m) = \left( {\cal E}^m_A \otimes \id_B \right)\rho_{AB} = \sigma^m_{AB}, \quad m \in {\cal M}.$$
Here, $\varphi$ denotes the encoding map for a code of size $M$, as defined in terms of the CPTP maps $ {\cal E}^m_A$ for $m \in {\cal M}$. Since the system $A$ is sent through a noiseless channel, the final state in Bob's posession (when Alice sends the message $m$) is $\sigma^m_{AB}$. Let Bob's decoding map be given by the POVM $Y= \{Y^m_{AB}\}_{m \in {\cal M}}$. The triple $\left( \varphi, Y, M\right)$ defines a code ${\cal C}$ of size $M$, with average probability of error given by
\begin{align*} 
p_e({\cal C}) \coloneqq  \frac{1}{M}\sum_{m=1}^M \left(1 - \tr\left(\sigma^m_{AB}Y^m_{AB}\right)\right).
\end{align*}
\begin{definition}
For $\eps >0$, the $\eps$-error one-shot dense coding capacity of the state $\rho_{AB}$ is defined as
\begin{align*}
C_{dc}^{(1), \eps}(\rho_{AB}) \coloneqq  \sup \{ \log M \mid\exists\text{ a code }{\cal C} = (\varphi, Y, M)\text{ such that }p_e({\cal C}) \le \eps\}.
\end{align*}
\end{definition}

\subsubsection{One-shot bounds for noisy dense coding}
We obtain the following bounds on the $\eps$-error one-shot dense coding capacity of a bipartite state $\rho_{AB}$:
\begin{theorem}\label{thm:dense-coding-one-shot}
Fix $\eps >0$. Let $c,\eta, \eta'>0$ with $\eta>\frac{c}{1+c}\eps$ and $\delta>0$. The one-shot $\eps$-error dense coding capacity of a state $\rho_{AB}\in\cD(\cH_{A}\otimes\cH_B)$ satisfies the following bounds: 
\begin{align*} 
	\begin{aligned}
		 C_{dc}^{(1), \eps}(\rho_{AB})&\geq\log d_A - \min_{\Lambda^{A\rightarrow A}} \oH_s^{\eps-\eta}(A|B)_\sigma  +  \log\frac{c}{1+c}+\log\left(\eta - \frac{c}{1+c}\eps\right)\\
C_{dc}^{(1), \eps}(\rho_{AB})&\le \log d_A - \min_{\Lambda^{A\rightarrow A}} \oH_s^{\eps+ \eta'}(A|B)_\sigma + \delta - \log \eta',
	\end{aligned}
\end{align*}
where $d_A = \dim\cH_A$ and $\sigma_{AB}\coloneqq  (\Lambda \otimes \id_B) \rho_{AB}$ for any CPTP map $\Lambda^{A\rightarrow A}$.
\end{theorem}
\begin{proof}[Proof of the upper bound in \Cref{thm:dense-coding-one-shot} (Converse)]
To establish the converse bound, it suffices to prove that for any code ${\cal C}$ of $M$ codewords with
\begin{align}\label{eq:dc-assumption}
\log M > \log d_A - \min_{\Lambda^{A\rightarrow A}} \oH_s^{\eps+\eta'}(A|B)_\sigma + \delta - \log \eta',
\end{align}
where $d_A = \dim\cH_A$, and $\sigma_{AB}\coloneqq  (\Lambda \otimes \id_B) \rho_{AB}$ for any CPTP map $\Lambda^{A\rightarrow A}$, we must have $p_e(\cC) >\eps$.

If a code ${\cal C}$ of size $M$ has codewords 
$ \sigma^m_{AB}= \left( {\cal E}^m_A \otimes \id_B \right) \rho_{AB}$, and measurement operators $Y^m_{AB}$ for $m \in {\cal M}$, then for any $\gamma >0$ we have
\begin{align}
p_e({\cal C}) & =  1 - \frac{1}{M}\sum_{m=1}^M \tr \left[ Y^m_{AB} \left( \sigma^m_{AB} - 2^{- \gamma}( \one_A \otimes \rho_B) \right)\right] - \frac{2^{- \gamma}}{M}\sum_{m=1}^M \tr \left[ Y^m_{AB} (  \one_A \otimes \rho_B)\right]\notag\\
&\ge 1 - \frac{1}{M}\sum_{m=1}^M \tr \left[ \{\sigma^m_{AB} \ge 2^{- \gamma} \one_A \otimes \rho_B\} \left( \sigma^m_{AB} - 2^{- \gamma} \one_A \otimes \rho_B\right) \right] -  \frac{2^{- \gamma}}{M} \tr \one_A\notag\\
&\ge 1 - \max_{m \in {\cal M}} \tr \left[ \{\sigma^m_{AB} \ge 2^{- \gamma} \one_A \otimes \rho_B\} \left( \sigma^m_{AB} - 2^{- \gamma} \one_A \otimes \rho_B\right) \right] -  \frac{2^{\log d_A - \gamma}}{M},
\label{eq_bd1}
\end{align}
where we used \Cref{lem:tr-projector} and the fact that $\sum_{m \in {\cal M}} Y^m_{AB} = \one_{AB}$ in the first inequality and $\tr \one_A = d_A$ in the second inequality.

For some arbitrary $\eta', \delta >0$, choose 
\begin{align}
	\gamma &= \min_{\Lambda^{A\rightarrow A}} \oH_s^{\eps + \eta'}(A|B)_\sigma - \delta\le  \oH_s^{\eps + \eta'}(A|B)_{\sigma^m} - \delta < \oH_s^{\eps + \eta'}(A|B)_{\sigma^m}\label{eq:dc-gamma}
\end{align} 
where $\sigma_{AB}\coloneqq  (\Lambda \otimes \id_B) \rho_{AB}$ and $\sigma^m_{AB}$ is any codeword. The first inequality holds because $ \sigma^m_{AB} =   \left( {\cal E}^m_A \otimes \id_B \right) \rho_{AB}$ with ${\cal E}^m_A \otimes \id_B$ being a CPTP map, and hence 
$$ \oH_s^{\eps + \eta'}(A|B)_{\sigma^m} \ge  \min_{\Lambda^{A\rightarrow A}} \oH_s^{\eps + \eta'}(A|B)_\sigma.$$
The choice of $\gamma$ in \eqref{eq:dc-gamma} yields 
\begin{align}\label{eq:dc-tr-upper-bound}
\tr \left[ \{\sigma^m_{AB} \ge  2^{- \gamma} \one_A \otimes \rho_B\} \left( \sigma^m_{AB} - 2^{- \gamma} \one_A \otimes \rho_B\right) \right] < 1 - (\eps + \eta'),
\end{align}
and by \eqref{eq:dc-assumption}, \eqref{eq_bd1} and \eqref{eq:dc-tr-upper-bound}, we obtain
\begin{align*}
p_e(\cC) &> \eps+\eta' - 2^{\log d - \gamma - \log M}\\
&> \eps + \eta' - 2^{\log\eta'}\\
&=\eps.\qedhere
\end{align*}
\end{proof}

\begin{proof}[Proof of the lower bound in \Cref{thm:dense-coding-one-shot} (Achievability)]
To establish the lower bound, we use \Cref{lemHN}. For our purpose, we define ${\cal X}\coloneqq  \{(p,q)\mid 0\leq p,q\leq d_{A}-1\}$ (such that $|\mathcal{X}|=d_{A}^2$) and consider the classical-quantum (c-q) channel 
\begin{align}\label{cq} 
W : x \in {\cal X} \mapsto W_x \equiv W^x_{AB}\coloneqq  ({\cal U}_x \otimes \id_B) \sigma_{AB}.
\end{align}
Here, $\sigma_{AB}\coloneqq  (\Lambda \otimes \id_B) \rho_{AB}$ 
where $\Lambda^{A\rightarrow A}$ denotes the minimizing CPTP map in 
$$ \min_{\Lambda^{A\rightarrow A}} \oH_s^{\eps + \eta}(A|B)_\sigma,$$ 
and $\rho_{AB}$ is the initial entangled state shared between Alice and Bob. For $x=(p,q) \in {\cal X}$, the unitary encoding ${\cal U}_x$ is defined by $\cU_x(\rho)\coloneqq U_{p,q}\rho U_{p,q}^\dagger$. The Weyl operators \cite{Hol11} $U_{p,q}$ with $p,q \in \{0,1, \ldots, (d_{A}-1)\}$ are defined by $U_{p,q} = X^q Z^p$ where for any orthonormal basis $\{|j\rangle : j \in \{0,1,\ldots, (d-1)\}$ in a $d$-dimensional Hilbert space $\cH$ the action of $X$ and $Z$ is given by $X |j \rangle = |(j + 1)\bmod{d})\rangle$ and $Z |j \rangle = e^{\frac{2\pi i j}{d}}|j \rangle$. Hence, $$U_{p,q} =  e^{\frac{2\pi i p j}{d}} |(j + q)\bmod{d})\rangle.$$ Furthermore, it is known that 
\begin{align}\label{weyl} 
\sum_{x \in {\cal X}} {\cal U}_x (\omega) \equiv \sum_{p,q} U_{p,q} \omega U_{p,q}^\dagger = d \, \one
\end{align}
for any 
$\omega \in {\cal D}(\cH)$, with $d = \dim(\cH)$. 

Now let $\{p(x)\}_{x \in {\cal X}}$ denote the uniform distribution on ${\cal X}$, i.e.~$p(x) = 1/d_{A}^2$ for all $x \in {\cal X}$. For the ensemble $\{p(x), W^x_{AB}\}$, where the states $W^x_{AB}$ are defined by \eqref{cq}, we have by \eqref{weyl} that
\begin{align*}
\oW_{AB} &\coloneqq  \sum_{x \in {\cal X}} p(x) W^x_{AB} = \frac{1}{d_{A}^2} \sum_{x \in {\cal X}} 
( {\cal U}_x \otimes \id_B)\sigma_{AB}\\
&= \frac{\one_{A}}{d_{A}} \otimes \rho_B,
\end{align*}
with $\rho_B = \sigma_B = \tr_{A} \sigma_{AB}$. Setting $\alpha \coloneqq  \sum_{x \in {\cal X}} p(x) \tr \left[ W^x_{AB} \{ W^x_{AB} \ge 2^\gamma \oW_{AB}\}\right]$,
we obtain
\begin{align*}
\alpha & \ge \frac{1}{d_{A}^2} \sum_{x \in {\cal X}} \tr \left[ \{W^x_{AB} \ge 2^\gamma \oW_{AB}\} (W^x_{AB} -  2^\gamma \oW_{AB})\right]\nonumber\\
& = \frac{1}{d_{A}^2} \sum_{x \in {\cal X}} \tr\left[  \{W^x_{AB} \ge 2^{- ( \log d_{A} - \gamma)}(\one_{A} \otimes \rho_B)\} (W^x_{AB} - 2^{- ( \log d_{A} - \gamma)}(\one_{A} \otimes \rho_B) ) \right]\nonumber\\
&= \frac{1}{d_{A}^2} \sum_{x \in {\cal X}} \tr\left[  \{\sigma_{AB} \ge 2^{- (\log d_{A} - \gamma)}  (\one_{A} \otimes \rho_B)\} (\sigma_{AB} - 2^{- (\log d_{A} - \gamma)}  (\one_{A} \otimes \rho_B)) \right].\nonumber\\
\end{align*}
In the above we have made use of the fact that 
\begin{align*}
\{W^x_{AB} \ge 2^\gamma \oW_{AB}\}& =\left\lbrace({\cal U}_x \otimes \id_B)\sigma_{AB}\ \ge 2^\gamma ({\cal U}_x \otimes \id_B) \left(\frac{\one_{A}}{d_{A}} \otimes \rho_B\right)\right\rbrace\\
& = \{\sigma_{AB} \ge 2^{- (\log d_{A} - \gamma)}  (\one_{A} \otimes \rho_B)\},
\end{align*}
and the fact that the trace remains invariant under a unitary transformation.

Choose 
\begin{align}
\gamma = \log d_{A} -  \oH_s^{\eps- \eta}(A|B)_\sigma =  \log d_{A} - \min_{\Lambda^{A\rightarrow A}} \oH_s^{\eps - \eta}(A|B)_\sigma,
\label{id2}
\end{align}
for some arbitrary $\eta >0$. The second identity in \eqref{id2} follows from the fact that $\sigma_{AB}$ is the minimizing state in \eqref{cq}. For this choice of $\gamma$, $\alpha \ge 1 - (\eps - \eta)$. Using this in \eqref{eqHN} of \Cref{lemHN} we infer that, for the c-q channel $W$ defined by \eqref{cq}, there exists a code ${\cal C}$ of size $M$ such that
\begin{align*}
p_e({\cal C}) \le  (1 +c) (\eps - \eta)  + (2+c+c^{-1}) 2^{-\gamma} M.
\end{align*}
Thus, $p_e({\cal C}) \le \eps$ for the choice
$$(1 +c) (\eps - \eta)  + (2+c+c^{-1}) 2^{-\gamma} M =  \eps,$$
and hence,
$$ \frac{(1+c)^2}{c} 2^{-\gamma}M  = (1+c) \eta - c\eps.$$
We obtain 
\begin{align*}
\log M &= \gamma + \log\frac{c}{1+c}+\log\left(\eta - \frac{c}{1+c}\eps\right)\\
 &=  \log d_{A} -  \min_{\Lambda^{A\rightarrow A}} \oH_s^{\eps-\eta}(A|B)_\sigma +  \log\frac{c}{1+c}+\log\left(\eta - \frac{c}{1+c}\eps\right).
\end{align*}
Note that the argument of the logarithm in the right-most term yields the condition $\eta>\frac{c}{1+c}\eps$. The above implies that for the dense coding protocol employing the entangled 
state $\rho_{AB}$, we have the following bound:
$$ C_{dc}^{(1), \eps}(\rho_{AB})  \ge 
\log d_{A} - \min_{\Lambda^{A\rightarrow A}} \oH_s^{\eps-\eta}(A|B)_\sigma +  \log\frac{c}{1+c}+\log\left(\eta - \frac{c}{1+c}\eps\right)$$
This inference follows from the fact that in the dense coding protocol, after her encoding, Alice can post-process her system $A$ by ${\cal U}_x \circ \Lambda$ (where $\Lambda$ is some CPTP map) before sending it through the noiseless channel to Bob.
\end{proof}

\subsubsection{Second order asymptotics of noisy dense coding}
In order to obtain the second order asymptotics of the $\eps$-error one-shot dense coding capacity $C_{dc}^{(1),\eps}$, we restrict our analysis to the case where Alice is only allowed to use product encodings. More precisely, given an initial state $\rho_{AB}$, Alice chooses an encoding map $\Lambda^{A\rightarrow A}$. Her encoding map on $\rho_{AB}^{\otimes n}$ results in the state $((\Lambda\otimes\id_B)(\rho_{AB}))^{\otimes n}$, which is transmitted to Bob through the noiseless channel. This is in contrast to the general setting, where Alice selects an encoding map $\Lambda^{(n)}\colon A^n\rightarrow A^n$ for each $n$ and transmits the state $(\Lambda^{(n)}\otimes\id_{B^n})(\rho_{AB}^{\otimes n})$ \cite{HHH+01}.

Using the results from Section \ref{sec:2nd-order-asymptotics-D_s}, we obtain the following second order expansion for $C_{dc}^{n,\eps}(\rho_{AB})$, which is related to the one-shot dense coding capacity via \Cref{def:blocklength-quantity}.
To state the theorem, recall from \eqref{eq:dense-coding-first-order} that the first order asymptotics of noisy dense coding is determined by the following expression (note that the infimum can be replaced by a minimum since the set of all quantum channels is compact):
\begin{align*}
\log d + S(\rho_B) - \min_{\Lambda^{A\to A}} S((\Lambda\otimes\id_B)\rho_{AB}).
\end{align*}
To analyze the second order asymptotics, we define the set $\cF(\rho_{AB})$ of all CPTP maps $\Lambda^{A\to A}$ achieving the minimum in the last term of the above expression:
\begin{align}
\cF(\rho_{AB}) \coloneqq \argmin_{\Lambda^{A\to A}} S((\Lambda\otimes\id_B)\rho_{AB}).\label{eq:F-set}
\end{align}
We are now ready to prove the main result of this section.
\begin{theorem}\label{thm:dc-second-order-expansion}
	Consider an arbitrary bipartite state $\rho_{AB}\in\cD(\cH_{AB})$. 
	Then for $\eps\in(0,1/2]$ we have the following closed second order asymptotic expansion:
	\begin{align*}
	C_{dc}^{n,\eps}(\rho_{AB}) &= n\left(\log d_A+S(\rho_B)-\min_{\Lambda^{A\rightarrow A}}S((\Lambda\otimes\id)(\rho_{AB}))\right)\\
	&\quad {}+ \sqrt{n}\,\max_{\Lambda\in\cF(\rho_{AB})}\sk((\Lambda\otimes\id)(\rho_{AB})\|\one_{A}\otimes\rho_B)\Phi^{-1}(\eps)+\cO(\log n),
	\end{align*}
	where the maximization in the second term is over the set $\cF(\rho_{AB})$ of first-order-achieving CPTP maps defined in \eqref{eq:F-set}.
	For $\eps\in (1/2,1)$, we have the second order asymptotic bounds
\begin{align*}
	C_{dc}^{n,\eps}(\rho_{AB}) &\geq n\left(\log d_A+S(\rho_B)-\min_{\Lambda^{A\rightarrow A}}S((\Lambda\otimes\id)(\rho_{AB}))\right)\\
	&\quad {}+ \sqrt{n}\,\max_{\Lambda\in\cF(\rho_{AB})}\sk((\Lambda\otimes\id)(\rho_{AB})\|\one_{A}\otimes\rho_B)\Phi^{-1}(\eps)+\cO(\log n)\\
	C_{dc}^{n,\eps}(\rho_{AB}) &\leq n\left(\log d_A+S(\rho_B)-\min_{\Lambda^{A\rightarrow A}}S((\Lambda\otimes\id)(\rho_{AB}))\right)\\
	&\quad {}+ \sqrt{n}\,\max_{\Lambda^{A\rightarrow A}}\sk((\Lambda\otimes\id)(\rho_{AB})\|\one_{A}\otimes\rho_B)\Phi^{-1}(\eps)+\cO(\log n).
\end{align*}
\end{theorem}
\begin{proof}
	For all $\eps\in (0,1)$ we have the following one-shot bounds on the dense coding capacity $C_{dc}^{n,\eps}(\rho_{AB})$ by \Cref{thm:dense-coding-one-shot}: for $c>0$ and $\eta>\frac{c}{1+c}\eps$,
	\begin{align}
		C_{dc}^{n,\eps}(\rho_{AB}) &\geq n\log d_A - \min_{\Lambda^{A\rightarrow A}} \oH_s^{\eps-\eta}(A^n|B^n)_{\sigma^n} + \log\frac{c}{1+c}+\log\left(\eta - \frac{c}{1+c}\eps\right),\label{eq:dense-coding-one-shot-n-achieve}
		\intertext{and for $\delta,\eta'>0$,}
		C_{dc}^{n,\eps}(\rho_{AB}) &\leq n\log d_A - \min_{\Lambda^{A\rightarrow A}} \oH_s^{\eps+\eta'}(A^n|B^n)_{\sigma^n}+\delta - \log\eta'.\label{eq:dense-coding-one-shot-n-converse}
	\end{align}
	In both \eqref{eq:dense-coding-one-shot-n-achieve} and \eqref{eq:dense-coding-one-shot-n-converse} we set $\sigma^n\equiv\sigma_{AB}^{\otimes n}$ where $\sigma_{AB}\coloneqq (\Lambda\otimes\id_B)(\rho_{AB})$ . 

	We would like to set $\eta=\eta'=\frac{1}{\sqrt{n}}$ in the above bounds, just as in the proof of \Cref{thm:D_s-asymptotics}. To this end, let us choose $c=\frac{1}{n}$ for the constant from \Cref{lemHN} in \eqref{eq:dense-coding-one-shot-n-achieve}, which then guarantees the existence of a code $\cC_n$ of size $M_n$ such that
	$$\log M_{n}=\gamma +\log\left(\frac{1+n}{n}\eta-\frac{1}{n}\eps\right) - \log n + 2\log(1+n),$$
	subject to the condition 
	\begin{align*}
		\eta>\frac{1}{1+n}\eps. 
	\end{align*}
	This in turn implies a lower bound on the $\eps$-error capacity $C_{dc}(n,\eps)$ given by
	\begin{align}
		\begin{aligned}
		C_{dc}^{n,\eps}(\rho_{AB})&\geq n\log d_A - \min_{\Lambda^{A\rightarrow A}}\oH_s^{\eps-\eta}(A^n|B^n)_{\sigma^n} \\
											&\quad{} +\log\left(\frac{1+n}{n}\eta-\frac{1}{n}\eps\right) - \log n + 2\log(1+n).
		\end{aligned}\label{eq:blocklength-achieve-bound}
	\end{align}
	We observe that for $0<\eps<1$ we have
	$$\frac{1}{\sqrt{n}} > \frac{1}{1+n} > \frac{1}{1+n}\eps,$$ and thus, the choice $\eta=\frac{1}{\sqrt{n}}$ is valid in \eqref{eq:blocklength-achieve-bound}. 
	
	We can therefore choose $\eta=\eta'=\frac{1}{\sqrt{n}}$ and apply \Cref{thm:D_s-asymptotics} to $\oH_s^{\eps\pm\eta}(A^n|B^n)_{\sigma^n}$ in \eqref{eq:dense-coding-one-shot-n-converse} and \eqref{eq:blocklength-achieve-bound} in the following way:
	\begin{align}
		-\oH_s^{\eps\pm\eta}(A^n|B^n)_{\sigma^n} &= \underline{D}_s^{\eps\pm\eta}(\sigma_{AB}^n\|\one_{A^n}\otimes\rho_B^n)\notag\\
		&= n D(\sigma_{AB}\|\one_{A}\otimes\rho_B) + \sqrt{n}\,\sk(\sigma_{AB}\|\one_{A}\otimes\rho_B)\Phi^{-1}(\eps\pm\eta)+\cO(\log n)\notag\\
		&=  -nS(\sigma_{AB}) + nS(\rho_B) + \sqrt{n}\,\sk(\sigma_{AB}\|\one_{A}\otimes\rho_B)\Phi^{-1}(\eps)+\cO(\log n),\label{eq:H_s-expansion}
	\end{align}
	where the last equality follows from \Cref{lem:phi-trick} and the fact that $\tr_{A}\sigma_{AB}=\rho_B$.
	Finally, we take the maximum over all CPTP maps $\Lambda^{A\to A}$ in \eqref{eq:H_s-expansion}:
	\begin{multline}
	\max_{\Lambda^{A\rightarrow A}} \left\lbrace-\oH_s^{\eps\pm\eta}(A^n|B^n)_{\sigma^n}\right\rbrace = nS(\rho_B) + \max_{\Lambda^{A\rightarrow A}} \left\lbrace -nS(\sigma_{AB}) +  \sqrt{n}\,\sk(\sigma_{AB}\|\one_{A}\otimes\rho_B)\Phi^{-1}(\eps) \right\rbrace\\
	 +\cO(\log n),\label{eq:max-Hs-second-order}
	\end{multline}
	where $\sigma_{AB} = (\Lambda\otimes \id_B)(\rho_{AB})$ as before.
	
	We now use the expansion \eqref{eq:max-Hs-second-order} in the upper and lower bounds on $C_{dc}^{n,\eps}(\rho_{AB})$ in \eqref{eq:dense-coding-one-shot-n-converse} and \eqref{eq:blocklength-achieve-bound}, respectively.
	For \eqref{eq:blocklength-achieve-bound}, we bound the maximum on the right-hand side of \eqref{eq:max-Hs-second-order} from below by picking a CPTP map $\Lambda^*\in \argmin_{\Lambda\in\cF(\rho_{AB})}\sk(\sigma_{AB}\|\one_{A}\otimes\rho_B)$.
	In other words, among all CPTP maps optimizing the first order term for $C_{dc}^{n,\eps}(\rho_{AB})$, we pick one that additionally maximizes the second order term determined by the quantity $\sk(\sigma_{AB}\|\one_{A}\otimes\rho_B)$.
	This yields the following lower bound on $C_{dc}^{n,\eps}(\rho_{AB})$ valid for all $\eps\in(0,1)$:
	\begin{multline}
	C_{dc}^{n,\eps}(\rho_{AB}) \geq n \left(\log d_A + S(\rho_B) - \min_{\Lambda^{A\rightarrow A}}S(\sigma_{AB}) \right) + \sqrt{n} \max_{\Lambda\in\cF(\rho_{AB})} \sk(\sigma_{AB}\|\one_{A}\otimes\rho_B)\Phi^{-1}(\eps) \\ +\cO(\log n).\label{eq:second-order-lower-bound}
	\end{multline}
	
	To use \eqref{eq:max-Hs-second-order} in the upper bound \eqref{eq:dense-coding-one-shot-n-converse} on $C_{dc}^{n,\eps}(\rho_{AB})$, we first use the inequality $\max\lbrace x + y\rbrace \leq \max x + \max y$ in \eqref{eq:max-Hs-second-order} to obtain
	\begin{multline}
	C_{dc}^{n,\eps}(\rho_{AB}) \leq n \left(\log d_A + S(\rho_B) - \min_{\Lambda^{A\rightarrow A}}S(\sigma_{AB}) \right) + \sqrt{n} \max_{\Lambda^{A\rightarrow A}} \sk(\sigma_{AB}\|\one_{A}\otimes\rho_B)\Phi^{-1}(\eps) \\ +\cO(\log n).\label{eq:second-order-upper-bound}
	\end{multline}
	This bound is valid for all $\eps\in(0,1)$, but it does not match the second order expansion in \eqref{eq:second-order-lower-bound}, as the maximization in the second order term is over arbitrary CPTP maps.
	However, recall that $\Phi^{-1}(\eps)\leq 0$ for $\eps\leq 1/2$.
	Hence, in this case we can restrict the maximization in \eqref{eq:second-order-upper-bound} to the smaller set $\cF(\rho_{AB})$, giving $$\max_{\Lambda^{A\to A}} \sk(\sigma_{AB}\|\one_{A}\otimes\rho_B)\Phi^{-1}(\eps) \leq \max_{\Lambda\in\cF(\rho_{AB})} \sk(\sigma_{AB}\|\one_{A}\otimes\rho_B)\Phi^{-1}(\eps).$$ 
	Substituting this in \eqref{eq:second-order-upper-bound} leads to the following upper bound on $C_{dc}^{n,\eps}(\rho_{AB})$ valid for $\eps\in(0,1/2]$,
	\begin{multline*}
	C_{dc}^{n,\eps}(\rho_{AB}) \leq n \left(\log d_A + S(\rho_B) - \min_{\Lambda^{A\rightarrow A}}S(\sigma_{AB}) \right) + \sqrt{n} \max_{\Lambda\in\cF(\rho_{AB})} \sk(\sigma_{AB}\|\one_{A}\otimes\rho_B)\Phi^{-1}(\eps) \\ +\cO(\log n),
	\end{multline*}
	which finishes the proof.
\end{proof}
\Cref{thm:dc-second-order-expansion} immediately implies the first order asymptotics of noisy dense coding as proved in \cite{Win02}:
\begin{corollary}
	In the asymptotic limit, the noisy dense coding capacity of a bipartite state $\rho_{AB}\in\cD(\cH_{A}\otimes\cH_B)$ reads
	$$C_{dc}(\rho_{AB}) = \log d_A + S(\rho_B) - \inf_{\Lambda:\text{ CPTP}}S((\Lambda\otimes\id_B)\rho_{AB}),$$
	where $d_A=\dim\cH_A$.
\end{corollary}

\subsection{Entanglement concentration}\label{sec:ent-concentration}
Entanglement concentration is the protocol in which two parties, Alice and Bob, share a partially entangled state $|\psi_{AB}\rangle \in \cH_A \otimes \cH_B$, and they wish to convert it into a maximally entangled state $|\Phi^+_M\rangle$ (of Schmidt rank $M$) by local operations and classical communication (LOCC) alone. If $\Lambda$ denotes the LOCC map used by Alice and Bob, then the fidelity of the protocol is given by
\begin{align*}
F(\Lambda(\psi_{AB}), \Phi_M^+ ) \coloneqq  \langle \Phi_M^+ | \Lambda\left(\psi_{AB}\right) |\Phi_M^+\rangle.
\end{align*}
Note that the above definition of fidelity (which is also used in Sections \ref{sec:ent-dilution} and \ref{sec:irreversibility}) is the square of the one defined previously. For given orthonormal bases $\{|\tk_A\rangle\}_{k=1}^{d_A}$ and  $\{|\tk_B\rangle\}_{k=1}^{d_B}$ in Hilbert spaces $\cH_A$ and $\cH_B$ of dimensions $d_A$ and $d_B$ respectively, we define the canonical maximally entangled state of Schmidt rank $M = \min\{d_A, d_B\}$ as 
\begin{align}\label{eq:mes}
|\Phi^+_M\rangle = \frac{1}{\sqrt{M}} \sum_{k=1}^M |\tk_A\rangle |\tk_B\rangle.
\end{align}
In fact, in the following we consider $\cH_A \simeq \cH_B$, for simplicity, so that $d_A = d_B= d$ (say).

\begin{definition}
For any $\eps > 0$ the one-shot $\eps$-error distillable entanglement of the
pure state $\psi_{AB}$ is defined as
\begin{align*}
E_D^{(1), \eps}(\psi_{AB}) \coloneqq  \sup \{ \log M\mid\exists\text{ an LOCC map }\Lambda\text{ such that }F(\Lambda(\psi_{AB}), \Phi_M^+ ) \ge 1- \eps\}. 
\end{align*}
\end{definition}
Bennett et al.~\cite{BBPS96} established that in the asymptotic limit requiring perfect concentration (i.e.~$\eps\rightarrow 0$) yields the optimal rate of entanglement concentration to be the entropy of entanglement, i.e.~the von Neumann entropy $S(\rho_A)$ of the reduced state $\rho_A=\tr_B\psi_{AB}$:
$$\lim_{\eps\rightarrow 0}\lim_{n\rightarrow \infty}\frac{1}{n}E_D^{n,\eps}(\psi_{AB})=S(\rho_A)$$ 
In this section we determine the asymptotic behaviour to second order of the $\eps$-error distillable entanglement.
In the converse proof of \Cref{thm_ec} (as well as \Cref{thm_entdil} in Section \ref{sec:ent-dilution}) we employ the following result, which follows directly from Lemma 2 in \cite{VJN00}:
\begin{lemma}\label{lem:LOCC-to-pure}
Let $\psi_{AB}$ and $\phi_{AB}$ be bipartite pure states, and let $\Lambda$ be an arbitrary LOCC operation. Then there exists an LOCC operation $\bar{\Lambda}$ such that $\bar{\Lambda}(\phi_{AB})$ is pure, and
\begin{align*}
F(\psi_{AB},\Lambda(\phi_{AB})) \leq F(\psi_{AB},\bar{\Lambda}(\phi_{AB})),
\end{align*}
where we use the notation $F(\rho,\sigma)\coloneqq \|\sqrt{\rho}\sqrt{\sigma}\|_1^2$.
\end{lemma}

\subsubsection{One-shot bounds for entanglement concentration}
We obtain the following bounds on the distillable entanglement:
\begin{theorem}\label{thm_ec}
Fix $\eps >0$ and let $\delta, \eta>0$. Then the one-shot $\eps$-error distillable entanglement of a
pure state $\psi_{AB}$ satisfies the bounds
\begin{align}\label{eq_thm}
\uH_s^{\eps-\eta}(\rho_A) + \log \eta + \log (1-\eps) - \Delta
 \le E_D^{(1), \eps}(\psi_{AB})  \le - \oH_s^{\eps + \eta}(A|B)_\psi + \delta - \log \eta,
\end{align}
where $\Delta$ is a number included to ensure that the left-hand side is the logarithm of an integer.
\end{theorem}

\begin{proof}[Proof of the lower bound in \Cref{thm_ec} (Achievability)]
Let the bipartite state $|\psi_{AB}\rangle$ have the Schmidt decomposition $$|\psi_{AB}\rangle = \sum_{k} \sqrt{\lambda_k}|e_k^A\rangle |e_k^B\rangle,$$
and let $\rho_A = \tr_B \psi_{AB}$ and define projection operators  $Q = \{\rho_A < 2^{-\gamma} \one_A\}$ and $\oQ = \one_A - Q$, for some $\gamma >0$. 

The first step of the protocol is for one of the parties (say, Alice) to do a von Neumann measurement, described by the projection operators $Q$ and $\oQ$, on her part of the shared bipartite state $\psi_{AB}$. If the outcome of the measurement corresponds to $\oQ$, then the protocol is aborted as unsuccessful. This occurs with probability
\begin{align*}
P_{\text{fail}} = \tr\oQ \rho_A.
\end{align*}
If the outcome of the measurement corresponds to $Q$, then the post-measurement state is given by
\begin{align*}
|\psi_{AB}^\prime\rangle \coloneqq  \frac{1}{\sqrt{\tr Q\rho_A }}
\sum_{\lambda_k < 2^{-\gamma}} \sqrt{\lambda_k}|e_k^A\rangle |e_k^B\rangle,
\end{align*}
and each of the eigenvalues of the reduced density matrix of this state is bounded from above by
$$\frac{2^{-\gamma}}{\tr Q\rho_A}.$$
Nielsen's majorization theorem \cite{Nie99} states that a bipartite pure state $\Psi$ with subsystem state $\sigma$ may be converted by LOCC into the pure state $\Phi$ with subsystem state $\omega$, if and only if the ordered eigenvalues of $\sigma$ are majorized by those of $\omega$. Specifically, 
\begin{align}\label{eq:ine2}
\sum_{m=1}^k \nu_m \le \sum_{m=1}^k \mu_m
\end{align}
for all $1\leq k\leq d$, with equality holding for $k=d$. Here $\{\nu_m\}$ and $\{\mu_m\}$ denote the
sets of eigenvalues of $\Psi$ and $\Phi$ respectively, labelled in a manner such that $\nu_1 \ge \nu_2 \ge \ldots$, and $\mu_1 \ge \mu_2 \ge \ldots$, and 
$d = {\rm{dim}}\,\cH_A =  {\rm{dim}}\,\cH_B.$ It follows from Nielsen's theorem that the state $|\psi_{AB}^\prime\rangle$ may be transformed by LOCC into the maximally entangled state $|\Phi^+_M\rangle$ of Schmidt rank
\begin{align}\label{eq:srank}
M = \lfloor 2^\gamma \tr Q\rho_A\rfloor,
\end{align}
as the eigenvalues all obey the inequality in \eqref{eq:ine2}. This concludes the protocol.

For some $\gamma' >0$ we can therefore bound the probability of failure of the protocol as follows: 
\begin{align*}
P_{\text{fail}} =  \tr\oQ \rho_A & = \tr \left[\oQ \left( \rho_A - 2^{-\gamma^\prime}\one_A\right)\right] +  2^{-\gamma^\prime}\tr\oQ \one_A\\
& \le \tr \left( \rho_A - 2^{- \gamma^\prime}\one_A\right)_+ + 2^{-\left(\gamma^\prime - \gamma\right)},
\end{align*}
where the inequality follows from the fact that $ \tr \oQ \le 2^\gamma$ (see \Cref{lem3}). Choosing $\gamma = \gamma^\prime + \log \eta$ for some arbitrary $\eta>0$, and
$$\gamma^\prime = \uH_s^{\eps - \eta}(\rho_A),$$
we obtain
\begin{align}\label{eq:fail}
P_{\text{fail}} = \tr \oQ \rho_A \le \eps - \eta + \eta = \eps.
\end{align}
Moreover, it follows from \eqref{eq:srank} that for this choice of $\gamma$ we have
\begin{align*}
\log M &=  \gamma + \log (\tr Q\rho_A) - \Delta\\
&\ge  \uH_s^{\eps - \eta}(\rho_A) + \log \eta + \log (1- \eps) - \Delta,
\end{align*}
where $\Delta$ is a constant included to ensure that $\log M$ is the logarithm of an integer, and the inequality follows from the fact that
$$\tr Q\rho_A = 1 -  \tr\oQ\rho_A \ge 1-\eps,$$
since we have from \eqref{eq:fail} that $\tr\oQ\rho_A \le \eps$. 
\end{proof}

\begin{proof}[Proof of the upper bound in \Cref{thm_ec} (Converse)]
To prove the upper bound, we need to establish that if $\Lambda$ is any LOCC operation such that $F(\Lambda(\psi_{AB}),\Phi^+_M)\geq 1-\eps$ for a given $\eps\in (0,1)$, then
\begin{align}\label{eq:ec-to-prove}
\log M \leq - \oH_s^{\eps + \eta}(A|B)_\psi + \delta - \log \eta.
\end{align}
We prove this by contradiction. Hence, let us assume that
\begin{align}\label{eq:ec-assumption}
\log M >  - \oH_s^{\eps + \eta}(A|B)_\psi + \delta - \log \eta.
\end{align}
By \Cref{lem:LOCC-to-pure}, there is an LOCC operation $\bar{\Lambda}$ such that $\phi_{AB} \coloneqq \bar{\Lambda}(\psi_{AB})$ is pure and 
\begin{align*}
F\left( \Lambda(\psi_{AB}), \Phi^+_M\right) \leq F\left(\bar{\Lambda}(\psi_{AB}), \Phi^+_M\right).
\end{align*} 
Hence, for an arbitrary constant $\gamma>0$, we have
\begin{align}
1-\eps &\leq F(\Lambda(\psi_{AB}),\Phi^+_M)\nonumber\\
&\leq F\left( \bar{\Lambda}(\psi_{AB}), \Phi^+_M\right)\nonumber\\
&=\langle  \Phi^+_M |  \bar{\Lambda}(\psi_{AB}) | \Phi^+_M\rangle\nonumber\\
&= \tr \left[  \Phi^+_M \left( \phi_{AB}-  2^{-\gamma} (\one_A \otimes \phi_B)\right)\right] +  2^{-\gamma} \tr \left[  \Phi^+_M  (\one_A \otimes \phi_B)\right]\nonumber\\
&\le \tr \left[  \{\phi_{AB} \ge 2^{-\gamma}(\one_A \otimes \phi_{B})\}\left( \phi_{AB} - 2^{-\gamma}(\one_A \otimes \phi_{B})\right)\right] + \tr \left[  \Phi^+_M  
 2^{-\gamma}(\one_A \otimes \phi_{B})\right]\nonumber\\
&\le \tr \left[  \{\phi_{AB} \ge 2^{-\gamma}(\one_A \otimes \phi_{B})\}\left( \phi_{AB} - 2^{-\gamma}(\one_A \otimes \phi_{B})\right)\right] + 
 \frac{2^{-\gamma}}{M},\label{eqa}
\end{align}
where the second inequality follows from \Cref{lem:tr-projector} and the last identity
holds because $$\tr_A \Phi^+_M = \frac{\one_A}{M}$$ and $\phi_B$ is a normalized state.

Choose 
\begin{align}\label{eq:ec-gamma}
\gamma = \oH_s^{\eps + \eta}(A|B)_\phi - \delta
\end{align}
for some arbitrary $\delta >0$. Then we have 
\begin{align}\label{eq:ec-gamma-bound}
	-\gamma = -\oH_s^{\eps+\eta}(A|B)_\phi + \delta \leq -\oH_s^{\eps+\eta}(A|B)_\psi +\delta
\end{align}
by \Cref{lem:lo}. Substituting \eqref{eq:ec-assumption}, \eqref{eq:ec-gamma}, and \eqref{eq:ec-gamma-bound} in \eqref{eqa} yields
\begin{align*}
1-\eps & < 1-(\eps+\eta) + 2^{-\gamma-\log M}\\
	&< 1-(\eps+\eta) + 2^{\log\eta}\\
	&= 1-\eps,
\end{align*}
which is clearly a contradiction. Hence, it follows that \eqref{eq:ec-to-prove} must hold for any LOCC operation $\Lambda$ for which $F\left(\Lambda(\psi_{AB}),\Phi^+_M\right)\geq 1-\eps$ for some given $\eps\in (0,1)$.
\end{proof}

\subsubsection{Second order asymptotics of entanglement concentration}
The second order asymptotics of entanglement concentration are recorded in the following theorem, where $E_D^{n,\eps}(\psi_{AB})$ is related to the one-shot distillable entanglement via  \Cref{def:blocklength-quantity}.
\begin{theorem}\label{thm:ec-second-order}
	The second order asymptotic expansion of the distillable entanglement of a pure state $\psi_{AB}$ with error $\eps\in(0,1)$ is given by
	$$E_D^{n,\eps}(\psi_{AB}) = nS(\rho_A) + \sqrt{n\left(\tr[\rho_A(\log\rho_A)^2]- S(\rho_A)^2\right)}\,\Phi^{-1}(\eps) + \cO(\log n),$$
	where $\rho_A=\tr_B\psi_{AB}$.
\end{theorem}
\begin{proof}
	We abbreviate $\rho^n\equiv\rho^{\otimes n}$. From \Cref{thm_ec} we obtain the following bounds:
	\begin{align}\label{eq:ec-blocklength-bounds}
		\uH_s^{\eps-\eta}(\rho_A^n)+\log\eta + \log(1-\eps)-\Delta \leq E_D^{n,\eps}(\psi_{AB}) \leq -\oH_s^{\eps+\eta}(A^n|B^n)_{\psi^n} + \delta - \log\eta
	\end{align}
	Noting that $\log(1-\eps) - \Delta = \cO(1)$, setting $\eta=\frac{1}{\sqrt{n}}$ and applying \Cref{thm:D_s-asymptotics} to $\uH_s^{\eps-\eta}(\rho_A^n)$, we get
	\begin{align}
		\uH_s^{\eps-\eta}(\rho_A^n) &= - \oD_s^{\eps-\eta}(\rho_A^n\|\one_{A^n})\notag\\
		&= -nD(\rho_A\|\one_A) + \sqrt{n}\,\sk(\rho_A\|\one_A)\invP{\eps-\eta} + \cO(\log n)\notag\\
		&= nS(\rho_A) + \sqrt{n\left( \tr[\rho_A(\log\rho_A - \log\one_A)^2]-D(\rho_A\|\one_A)^2\right)}\,\invP{\eps} + \cO(\log n)\notag\\
		&= nS(\rho_A) + \sqrt{n\left( \tr[\rho_A(\log\rho_A)^2]-S(\rho_A)^2\right)}\,\invP{\eps} + \cO(\log n)\label{eq:ec-lower-bound}
	\end{align}
	where we used \Cref{lem:phi-trick} in the third equality. For the upper bound in \eqref{eq:ec-blocklength-bounds} we set $\eta=\frac{1}{\sqrt{n}}$ and apply \Cref{thm:D_s-asymptotics} and \Cref{lem:phi-trick} to obtain
	\begin{align}
		-\oH_s^{\eps+\eta}(A^n|B^n)_{\psi^n} &= \uD_s^{\eps+\eta}(\psi_{AB}\|\one_A\otimes\rho_B)\notag\\
		&= nD(\psi_{AB}\|\one_A\otimes\rho_B) + \sqrt{n}\,\sk(\psi_{AB}\|\one_A\otimes\rho_B)\invP{\eps} + \cO(\log n)\\
		&= nS(\rho_A) + \sqrt{n}\,\sk(\psi_{AB}\|\one_A\otimes\rho_B)\invP{\eps} + \cO(\log n)\label{eq:ec-upper-bound-inter}.
	\end{align}
	Moreover,
	\begin{align}
		 \sk(\psi_{AB}\|\one_A\otimes\rho_B)^2 &= \tr[\psi_{AB}(\log\psi_{AB}-\log(\one_A\otimes\rho_B))^2] - D(\psi_{AB}\|\one_A\otimes\rho_B)^2\notag\\
		 &= \tr[\rho_B(\log\rho_B)^2] - S(\rho_A)^2\notag\\
		 &= \tr[\rho_A(\log\rho_A)^2] - S(\rho_A)^2\label{eq:s(A|B)},
	\end{align}
	since the first term in the last line only depends on the eigenvalues of $\rho_A$, which are identical to the eigenvalues of $\rho_B$. Substituting \eqref{eq:s(A|B)} in \eqref{eq:ec-upper-bound-inter}, we obtain
	$$-\oH_s^{\eps+\eta}(A^n|B^n)_{\psi^n} = nS(\rho_A) + \sqrt{n\left( \tr[\rho_A(\log\rho_A)^2]-S(\rho_A)^2\right)}\,\invP{\eps} + \cO(\log n),$$
	 which, when substituted in \eqref{eq:ec-blocklength-bounds} together with the lower bound \eqref{eq:ec-lower-bound}, yields the result.
\end{proof}

\subsection{Pure-state entanglement dilution}\label{sec:ent-dilution}
Pure-state entanglement dilution is the protocol which is essentially opposite to entanglement concentration. Here, the two parties Alice and Bob share a maximally entangled state and wish to convert it into a particular non-maximally entangled state $|\psi_{AB}\rangle$ by LOCC alone. 

Suppose Alice and Bob initially share a maximally entangled state $|\Phi^+_M\rangle$ of Schmidt rank $M$ given by \eqref{eq:mes}. Let $\rho_A = \tr_B \psi_{AB},$ and suppose the Schmidt decomposition of 
$|\psi_{AB}\rangle$ be given by
\begin{align}\label{schmidt1}
|\psi_{AB}\rangle = \sum_{k=1}^N \sqrt{\lambda_k} |e_k\rangle \otimes |e_k\rangle,
\end{align}
where the Schmidt coefficients $\lambda_k$ are arranged in non-increasing order,
i.e., $\lambda_1 \ge \lambda_2 \ge \ldots \ge \lambda_N$. 

\begin{definition}\label{def:entcost}
For any $\eps > 0$, the one-shot $\eps$-{\em{error entanglement cost}} of the
pure state $\psi_{AB}$ is defined as
\begin{align*}
E_C^{(1), \eps}(\psi_{AB}) \coloneqq \inf \{ \log M\mid\exists\text{ an LOCC map }\Lambda \text{ such that }F\left(\psi_{AB}, \Lambda(\Phi^+_M)\right) \ge 1- \eps\}. 
\end{align*}
\end{definition}
The entanglement cost in the case of asymptotically perfect dilution is known \cite{BBPS96} to be given by the entropy of entanglement.

\subsubsection{One-shot bounds for pure-state entanglement dilution}
We obtain the following one-shot bounds for pure-state entanglement dilution:
\begin{theorem}\label{thm_entdil} 
Fix $\eps >0$ and let $\delta, \eta >0$. Then the one-shot {{$\eps$-error entanglement cost}} of a bipartite pure state $\psi_{AB}$ satisfies the following bounds:
\begin{align*} 
\oH_s^{\eps+\eta} (\rho_A)  - \delta + \log \eta \le E_C^{(1), \eps}(\psi_{AB})  \le \oH_s^\eps(\rho_A) 
\end{align*}
\end{theorem}
\begin{proof}[Proof of the upper bound in \Cref{thm_entdil} (Achievability)] 
Suppose that Alice and Bob share a maximally entangled state of Schmidt rank $M$ given by \eqref{eq:mes}. Alice locally prepares the desired state \eqref{schmidt1}. If $M\ge N$, then the part $B$ of the above state can be teleported to Bob perfectly. However, if $M<N$ then the part $B$ of only the following truncated state can be perfectly teleported:
$$ |\psi^\prime_{AB}\rangle \coloneqq \frac{1}{\tr(P_M \psi_{AB})}
(P_M \otimes \one_B) |\psi_{AB}\rangle, \quad \text{where }
(P_M \otimes \one_B) |\psi_{AB}\rangle = \sum_{k=1}^M \sqrt{\lambda_k} |e_k\rangle \otimes |e_k\rangle,
$$
that is, $P_M$ denotes the orthogonal projection onto the eigenspace of $\rho_A$ spanned by its $M$ largest eigenvalues. In this case the final state shared between Alice and Bob after the teleportation is given by $|\psi^\prime_{AB}\rangle$. The fidelity of the entanglement dilution protocol in this case is given by
\begin{align*}
\langle\psi_{AB}|\psi^\prime_{AB} |\psi_{AB}\rangle
&= \sum_{k=1}^M \lambda_k.
\end{align*} 
This simple protocol turns out to be sufficient for proving the upper bound
in \Cref{thm_entdil}.

In particular, suppose that the final state shared between Alice and Bob after teleportation is given by
\begin{align*}
|\phi_{AB}\rangle \coloneqq \frac{1}{\tr [(Q \otimes \one_B)\psi_{AB}]}
Q|\psi_{AB}\rangle,
\end{align*}
where $Q\equiv Q(\gamma)$ is the projection operator given by
\begin{align*} 
Q\coloneqq \{ \rho_A \ge 2^{-\gamma \one_A}\},
\end{align*}
where $\gamma >0$ and $\rho_A = \tr_B \psi_{AB}$. Further, note that in order to constitute the final state, we necessarily have
\begin{align}
\log M = \gamma.\label{cond}
\end{align}
This is because the operator $Q$ projects onto the subspace spanned by eigenvectors of $\rho_A$ corresponding to eigenvalues  
which are less than or equal to $2^{-\gamma}$, and the number of such eigenvalues is at most $2^\gamma$. As the initial maximally entangled state \eqref{eq:mes} is of Schmidt rank $M$, perfect teleportation can be achieved if $M = 2^\gamma$, which in turn implies \eqref{cond}.

Choosing $\gamma = \oH_s^\eps(\rho_A)$, it follows from \Cref{def:children-entropies} of $\oH_s^\eps(\rho_A)$ that 
we have
\begin{align*}
\tr(Q\rho_A) \ge \tr\left[  \{ \rho_A \ge 2^{-\gamma} \one_A\}(\rho_A - 2^{-\gamma} \one_A)\right] = 1 - \eps.
\end{align*}
Hence, denoting the teleportation by $\Lambda$, the fidelity of the entanglement dilution protocol can be bounded as follows:
\begin{align*}
F\left(\psi_{AB}, \Lambda(\Phi^+_M)\right) &= \frac{1}{\tr\left[(Q \otimes \one_B)\psi_{AB}\right]}
|\langle \psi_{AB} |(Q \otimes \one_B) |\psi_{AB} \rangle |^2\nonumber\\
&= \tr Q \rho_A  \nonumber\\
& \ge 1- \eps
\end{align*}
Further, from \eqref{cond} we infer that 
$$\log M =  \oH_s^\eps(\rho_A) + \delta,$$
which along with \Cref{def:entcost} of the one-shot $\eps$-error entanglement cost implies the upper bound of \Cref{thm_entdil}.
\end{proof}

\begin{proof}[Proof of the lower bound in \Cref{thm_entdil} (Converse)]
In order to establish the converse, we need to prove that for any LOCC operation $\Lambda$ for which $F\left(\psi_{AB},\Lambda(\Phi^+_M)\right)\geq 1-\eps$ for some given $\eps\in(0,1)$, we must have 
\begin{align}\label{eq:ed-to-prove} 
\log M \geq \oH_s^{\eps+\eta} (\rho_A)  - \delta + \log \eta.
\end{align}
We prove this by contradiction, assuming that 
\begin{align}\label{eq:ed-assumption} 
\log M < \oH_s^{\eps+\eta} (\rho_A)  - \delta + \log \eta.
\end{align}

As in Section \ref{sec:ent-concentration}, by \Cref{lem:LOCC-to-pure} there is an LOCC operation $\bar{\Lambda}$ such that $\bar{\Lambda}(\Phi^+_M)$ is pure and $F(\psi_{AB},\Lambda(\Phi^+_M))\leq F(\psi_{AB},\bar{\Lambda}(\Phi^+_M))$. 
We then have
\begin{align}
1-\eps &\leq F(\psi_{AB},\Lambda(\Phi^+_M))\nonumber\\
&\leq F\left( \psi_{AB}, \bar{\Lambda}(\Phi^+_M) \right)\nonumber\\
&= \langle  \psi_{AB} | \sum_j (U_j \otimes K_j) |\Phi^+_M\rangle \langle \Phi^+_M| (U_j^\dagger \otimes K_j^\dagger)| \psi_{AB}\rangle,\nonumber\\
&= \sum_j | \langle  \psi_{AB} |(U_j \otimes K_j) |\Phi^+_M\rangle |^2,
\nonumber
\end{align}
where we made use of the Lo-Popescu characterization \cref{Lo-Pop} of LOCC operations which take pure states to pure states. 

Let $W$ be a unitary operator such that $W|e_k\rangle = |\tk\rangle$ for each $k=1,2,\ldots, d$. Recall that $\{|e_k\rangle\}$ denotes the Schmidt basis of $|\psi_{AB}\rangle$, whereas $ \{|\tk\rangle\}$ denotes the given orthonormal basis in $\cH_A ( \simeq \cH_B)$ in terms of which the maximally entangled state $|\Phi_M^+\rangle$ has been defined in \eqref{eq:mes}. Further, $N$ is the number of non-zero Schmidt coefficients of $|\psi_{AB}\rangle$. Then, defining the projector $P_M = \sum_{k=1}^M |\tk_A\rangle \langle \tk_A|$,  we have
\begin{align*}
(U_j \otimes \one_B) |\Phi^+_M\rangle &= \frac{1}{\sqrt{M}} \sum_{k=1}^M U_j |\tk_A\rangle \otimes |\tk_B\rangle \nonumber\\
&=  \frac{1}{\sqrt{M}}  (\one_A \otimes P_M)\sum_{k=1}^N U_j W |k_A\rangle \otimes W |k_B\rangle \nonumber\\
&=   \frac{1}{\sqrt{M}} (\one_A \otimes P_B) \sum_{k=1}^N  |k_A\rangle \otimes (U_j W)^T W |k_B\rangle,\nonumber\\
&=  \frac{1}{\sqrt{M}} (\one_A \otimes P_B) \sum_{k=1}^N  |k_A\rangle \otimes V_j |k_B\rangle,\nonumber\\
\end{align*}
where we have defined $V_j = (U_j W)^T W$ and used the fact that $$\sum_k U|k\rangle \otimes |k\rangle =  \sum_k |k\rangle \otimes U^T|k\rangle.$$
Then, using the the Schmidt decomposition \eqref{schmidt1} of $|\psi_{AB}\rangle$, we obtain
\begin{align*}
\langle  \psi_{AB} |(U_j \otimes K_j) |\Phi^+_M\rangle
& = \sum_k \sqrt{\frac{\lambda_k}{M}} \langle e_k| K_j P_M V_j |e_k\rangle\nonumber\\
&= \tr \left(\frac{1}{\sqrt{M}} \sqrt{\rho_A} K_j P_M V_j \right).
\end{align*}
Using the Cauchy-Schwarz inequality, we further obtain
\begin{align}
F\left( \psi_{AB}, \bar{\Lambda}(\Phi^+_M) \right) &= \sum_k \left| \tr \left(\frac{1}{\sqrt{M}} \sqrt{\rho_A} K_j P_M V_j \right)\right|^2\nonumber\\
&=  \sum_k \left| \tr \left(\frac{1}{\sqrt{M}} \sqrt{\rho_A} K_j P_M \cdot P_M  V_j \right)\right|^2\nonumber\\
& \le \frac{1}{M} \tr (P_M) \max_j \tr(V_j^\dagger P_M V_j \rho_A)\nonumber\\
&= \frac{1}{M} \tr (P_M) \max_j \tr(P_j \rho_A),\nonumber\\
&=  \max_j \tr(P_j \rho_A),\label{eq:ed-F-bound}
\end{align}
where we have defined $P_j = V_j^\dagger P_M V_j$ and used the fact that $\tr P_M = M$.

Using \Cref{lem:tr-projector} we have, for any $\gamma >0$ that
\begin{align}
\tr P_j \rho_A &= \tr \left[ P_j( \rho_A - 2^{-\gamma} \one_A)\right]+ 2^{-\gamma} \tr P_j,\nonumber\\
&\le \tr \left[\{( \rho_A \ge 2^{-\gamma} \one_A\}( \rho_A - 2^{-\gamma} \one_A) \right]+ 2^{-\gamma} \tr P_j,\nonumber\\
&= \tr \left[\{( \rho_A \ge 2^{-\gamma} \one_A\}( \rho_A - 2^{-\gamma} \one_A) \right]+ 2^{-\gamma}M,\label{last0}
\end{align}
since by the cyclicity of the trace and the unitarity of the operators $V_j$ we have $$\tr P_j = \tr V_j^\dagger P_M V_j= \tr P_M = M.$$ 

Choose 
\begin{align}\label{eq:ed-gamma}
\gamma = \oH_s^{\eps + \eta}(\rho_A) - \delta
\end{align}
for some arbitrary $\delta, \eta >0$. Then substituting \eqref{last0} and \eqref{eq:ed-gamma} in \eqref{eq:ed-F-bound}, and using \eqref{eq:ed-assumption}, we obtain
\begin{align*}
	1-\eps&< 1-(\eps+\eta) + 2^{-\gamma + \log M}\\
	&< 1- (\eps+\eta) + 2^{\log\eta}\\
	&=1-\eps,
\end{align*}
which is clearly a contradiction. Hence, we proved that \eqref{eq:ed-to-prove} must hold for any LOCC operation $\Lambda$ such that $F\left(\psi_{AB},\Phi^+_M\right)\geq 1-\eps$ for some given $\eps\in(0,1)$.
\end{proof}

\subsubsection{Second order asymptotics for pure-state entanglement dilution}
We obtain the following second order asymptotic expansion for entanglement dilution, where $E_C^{n,\eps}(\psi_{AB})$ is related to the one-shot entanglement cost via \Cref{def:blocklength-quantity}.
\begin{theorem}\label{thm:ed-second-order}
Let $\psi_{AB}$ be a pure bipartite state. For any $\eps\in(0,1)$, the second order asymptotic expansion of $E_C^{n,\eps}(\psi_{AB})$ is given by
$$E_C^{n,\eps}(\psi_{AB}) = nS(\rho_A) - \sqrt{n(\tr[\rho_A(\log\rho_A)^2]-S(\rho_A)^2)}\,\invP{\eps} + \cO(\log n),$$
where $\rho_A=\tr_B\psi_{AB}$.
\end{theorem}
\begin{proof}
We abbreviate $\rho_A^n\equiv\rho_A^{\otimes n}$. By \Cref{thm_entdil} we have the following one-shot bounds for finite block length size:
\begin{align}\label{eq:ed-finite-blocklength}
\oH_s^{\eps+\eta}(\rho_A^n) - \delta + \log\eta \leq E_C^{n,\eps}(\psi_{AB}) \leq \ohs(\rho_A^n)
\end{align}
We apply \Cref{thm:D_s-asymptotics} to $\ohs(\rho_A^n)$ to obtain
\begin{align}
\oH_s^{\eps}(\rho_A^n) &= -nD(\rho_A\|\one_A) - \sqrt{n}\,\sk(\rho_A\|\one_A)\invP{\eps} + \cO(\log n)\notag\\
&= nS(\rho_A) - \sqrt{n(\tr[\rho_A(\log\rho_A)^2]-S(\rho_A)^2)}\,\invP{\eps} + \cO(\log n)\label{eq:H_s-expansion2}
\end{align}
For the lower bound in \eqref{eq:ed-finite-blocklength} we set $\eta=\frac{1}{\sqrt{n}}$ and apply \eqref{eq:H_s-expansion2} and \Cref{lem:phi-trick}, which finally yields the claim.
\end{proof}
The converse bound in \Cref{thm_entdil} together with the asymptotic expansion for entanglement dilution in \Cref{thm:ed-second-order} implies the following result, which we employ in the next section:
\begin{corollary}\label{cor:ed-n-shot-bound}
Let $\delta>0$ and $n$ be sufficiently large. For any LOCC map $\cD_n$ used in an entanglement dilution protocol, for which $$F\left(\psi_{AB}^{\otimes n},\cD_n\left(\Phi^+_{M_n}\right)\right)\geq 1-\delta,$$ the Schmidt rank $M_n$ of the MES $\Phi^+_{M_n}$ must satisfy
\begin{align*}
\log M_n &\geq nS(\rho_A) - \sqrt{n(\tr[\rho_A(\log\rho_A)^2]-S(\rho_A)^2)}\,\invP{\delta} + \cO(\log n).
\end{align*}
\end{corollary}

\subsection{Irreversibility of entanglement concentration}\label{sec:irreversibility}
In the asymptotic limit, the distillable entanglement for any given bipartite pure 
state $\psi_{AB}$ is equal to its entanglement cost, and is given by the {\em{entropy
of entanglement}} $S(\rho_A)$ (where $\rho_A = \tr_B \psi_{AB}$). This equality between
the optimal rates of asymptotically perfect entanglement concentration and entanglement dilution 
has led to the popular belief that entanglement conversions of pure bipartite states are \emph{asymptotically 
reversible}. The reversibility is in the sense that, in the asymptotic limit, the ebits extracted from multiple
copies of $\psi_{AB}$ via entanglement concentration can subsequently be used to recover the original 
number of copies of the state  $\psi_{AB}$ via entanglement dilution. However, recently Kumagai and 
Hayashi \cite{KH13} proved that the error incurred in the composite process (i.e.~concentration followed by dilution) is necessarily non-zero, even
in the asymptotic limit. This implies that, contrary to popular belief, entanglement concentration is in fact
irreversible.

Second order asymptotic expansions for the distillable entanglement and the entanglement cost (as 
given by \Cref{thm:ec-second-order} and \Cref{thm:ed-second-order} respectively) show that even though to leading order the number of ebits which can be distilled from $\psi_{AB}^{\otimes n}$ is equal to the number of ebits needed to create $\psi_{AB}^{\otimes n}$,
there is a discrepancy between these two quantities in the second order ($\sqrt{n}$). In this section we show how the irreversibility of entanglement concentration can be easily proved using this discrepancy. We formalize the concept of (ir-)reversibility in the following way:

\begin{definition}~\label{def:reversibility}
\begin{enumerate}[(i)]
\item Set $M_n=\lfloor \exp(nS(\rho_A))\rfloor$. A sequence $\lbrace\cC_n\rbrace_{\nin}$ of LOCC maps, where $\cC_n\colon\cD(\cH_{AB}^{\otimes n})\rightarrow \cD(\cH_{AB}^{\otimes n})$, is called an \emph{asymptotically perfect concentration protocol} (APCP), if it converts the sequence of pure states $\lbrace\psi_{AB}^{\otimes n}\rbrace_{\nin}$ to the sequence of MES $\lbrace\Phi_{M_n}^+\rbrace_{\nin}$ with asymptotically vanishing error, that is, for any $\eps>0$ there exists an $N\in\mathbb{N}$ such that for all $n>N$ we have
$$F\left(\cC_n(\psi_{AB}^{\otimes n}),\Phi_{M_n}^+\right)\geq 1-\eps.$$  This is equivalent to $\lim_{n\rightarrow \infty}F\left(\cC_n(\psi_{AB}^{\otimes n}),\Phi_{M_n}^+\right)=1$.
\item We say that entanglement concentration is \emph{reversible} if 
$$\lim_{n\rightarrow \infty}\max_{\cD_n}F\left(\psi_{AB}^{\otimes n},(\cD_n\circ\cC_n)(\psi_{AB}^{\otimes n})\right) = 1$$
where $\lbrace\cC_n\rbrace_{\nin}$ is a sequence of APCPs and the minimization is over all sequences of LOCC maps $\lbrace\cD_n\rbrace_{\nin}$, with $\cD_n\colon\cD(\cH_{AB}^{\otimes n})\rightarrow \cD(\cH_{AB}^{\otimes n})$, used in the entanglement dilution protocol.
\end{enumerate}
\end{definition}

\begin{theorem}
Entanglement concentration of any given bipartite pure state $\psi_{AB}$ is irreversible, i.e.~for any sequence $\lbrace\cC_n\rbrace_{\nin}$ of APCPs we have $$\lim_{n\rightarrow \infty}\max_{\cD_n\text{:LOCC}}F\left(\psi_{AB}^{\otimes n},(\cD_n\circ\cC_n)(\psi_{AB}^{\otimes n})\right)<1.$$
\end{theorem}
\begin{proof}
Let $\lbrace\cC_n\rbrace_{\nin}$ be a sequence of APCPs as in \Cref{def:reversibility} and fix $\eps\in(0,\frac{1}{2})$. By our result on the second order asymptotics for entanglement concentration (\Cref{thm_ec}), for sufficiently large $n$ we have 
\begin{align}
\log M_n &= nS(\rho_A)+\sqrt{n(\tr[\rho_A(\log\rho_A)^2]-S(\rho_A)^2)}\,\Phi^{-1}(\eps) + \cO(\log n)\notag\\
&\equiv nS(\rho_A)+\sqrt{n(\tr[\rho_A(\log\rho_A)^2]-S(\rho_A)^2)}\,\Phi^{-1}(\eps) + f(n)\label{eq:log-Mn}
\end{align}
where in the second line $f$ is a function satisfying $\lim_{n\rightarrow\infty}(f(n)/\sqrt{n})=0$,
which can be chosen since $\cO(\log n)\subset o(\sqrt{n})$. Contrarily, \Cref{cor:ed-n-shot-bound} implies that for every dilution protocol $\cD_n$ with $$F\left(\psi_{AB}^{\otimes n},\cD_n\left(\Phi^+_{M'_n}\right)\right)\geq 1-\delta,$$ where $\delta>0$ and $n$ is sufficiently large, the Schmidt rank $M'_n$ of the initially shared MES $\Phi^+_{M'_n}$ must satisfy
\begin{align}
\log M'_n &\geq nS(\rho_A) - \sqrt{n(\tr[\rho_A(\log\rho_A)^2]-S(\rho_A)^2)}\,\invP{\delta} + \cO(\log n)\notag\\
&\equiv nS(\rho_A) - \sqrt{n(\tr[\rho_A(\log\rho_A)^2]-S(\rho_A)^2)}\,\invP{\delta} + g(n).\label{eq:log-Mn-prime}
\end{align}
Here, $g$ is chosen such that $\lim_{n\rightarrow\infty}(g(n)/\sqrt{n})=0$. We show in the following that \eqref{eq:log-Mn} does not satisfy this bound. For $\delta,\eps>(0,\frac{1}{2})$ we have that $\invP{\eps}$ and $\invP{\delta}$ are strictly negative. Note that the restriction of both $\eps$ and $\delta$ to the interval $(0,\frac{1}{2})$ arises from the requirement that the overall error of the recovery process can be at most $1$. 

Subtracting the expression for $\log M_n$ given by \eqref{eq:log-Mn} from the RHS of \eqref{eq:log-Mn-prime} yields
\begin{align}
\sqrt{n}\left(-\sqrt{\tr[\rho_A(\log\rho_A)^2]-S(\rho_A)^2}\,(\invP{\delta}+\invP{\eps})+ \frac{g(n)-f(n)}{\sqrt{n}} \right).\label{eq:difference-term}
\end{align}
Since $\lim_{n\rightarrow \infty}\frac{g(n)-f(n)}{\sqrt{n}}=0$, we can make this term arbitrarily small for sufficiently large $n$, that is, 
$$\left|\frac{g(n)-f(n)}{\sqrt{n}}\right| < -\sqrt{\tr[\rho_A(\log\rho_A)^2]-S(\rho_A)^2}\,(\invP{\delta}+\invP{\eps}),$$ and hence the difference in \eqref{eq:difference-term} is strictly positive, yielding 
$$\log M_n < nS(\rho_A) - \sqrt{n(\tr[\rho_A(\log\rho_A)^2]-S(\rho_A)^2)}\,\invP{\delta} + g(n)$$ for sufficiently large $n$. Therefore, $F\left(\psi_{AB}^{\otimes n},\cD_n\left(\Phi_{M_n}^+\right)\right) < 1-\delta,$ and, using the inequality $\sqrt{1-x}\leq 1-\frac{x}{2}$, this implies
\begin{align}
\sqrt{F\left(\psi_{AB}^{\otimes n},\cD_n\left(\Phi_{M_n}^+\right)\right)} < 1-\frac{\delta}{2}.\label{eq:fidelity-bound}
\end{align}
Using the trace distance $d(.,.)$ we now compute:
\begin{align}
\frac{\delta}{2} &< 1- \sqrt{F\left(\psi_{AB}^{\otimes n},\cD_n\left(\Phi_{M_n}^+\right)\right)}\notag\\
&\leq d\left(\psi_{AB}^{\otimes n},\cD_n\left(\Phi_{M_n}^+\right)\right)\notag\\
&\leq d\left(\psi_{AB}^{\otimes n}, (\cD_n\circ\cC_n)(\psi_{AB}^{\otimes n})\right) + d\left((\cD_n\circ\cC_n)(\psi_{AB}^{\otimes n}),\cD_n\left(\Phi_{M_n}^+\right)\right)\notag\\
&\leq d\left(\psi_{AB}^{\otimes n}, (\cD_n\circ\cC_n)(\psi_{AB}^{\otimes n})\right) + d\left(\cC_n(\psi_{AB}^{\otimes n}),\Phi_{M_n}^+\right)\label{eq:trace-bound},
\end{align}
where the first inequality follows from \eqref{eq:fidelity-bound}, the second inequality is the Fuchs-van-de-Graaf inequality \cite{FVdG99}, the third inequality is the triangle inequality for the trace distance and the fourth inequality follows from the monotonicity of the trace distance under CPTP maps. Applying the Fuchs-van-de-Graaf-inequality to the term $d\left(\cC_n(\psi_{AB}^{\otimes n}),\Phi_{M_n}^+\right)$, we obtain 
\begin{align*}
1-\sqrt{F\left(\cC_n(\psi_{AB}^{\otimes n}),\Phi_{M_n}^+\right)} \leq d\left(\cC_n(\psi_{AB}^{\otimes n}),\Phi_{M_n}^+\right)\leq \sqrt{1-F\left(\cC_n(\psi_{AB}^{\otimes n}),\Phi_{M_n}^+\right)}
\end{align*}
and consequently, 
\begin{align}
\lim_{n\rightarrow \infty}d\left(\cC_n(\psi_{AB}^{\otimes n}),\Phi_{M_n}^+\right) = 0\label{eq:trace-limit}
\end{align}
since $\lim_{n\rightarrow \infty} F(\cC_n(\psi_{AB}^{\otimes n}),\Phi_{M_n}^+)=1$ by \Cref{def:reversibility} of the APCP $\lbrace\cC_n\rbrace_{\nin}$. Therefore, for any sequence of LOCC maps $\lbrace\cD_n\rbrace_{\nin}$, we have by \eqref{eq:trace-bound} and \eqref{eq:trace-limit} that
$$\lim_{n\rightarrow \infty}d\left(\psi_{AB}^{\otimes n}, (\cD_n\circ\cC_n)(\psi_{AB}^{\otimes n})\right) > \frac{\delta}{2}.$$ Choose the sequence $\lbrace\tilde{\cD}_n\rbrace_{\nin}$ such that 
$$\max_{\cD_n}F\left(\psi_{AB}^{\otimes n}, (\cD_n\circ\cC_n)(\psi_{AB}^{\otimes n})\right) = F\left(\psi_{AB}^{\otimes n}, (\tilde{\cD}_n\circ\cC_n)(\psi_{AB}^{\otimes n})\right).$$
Then, by another application of the Fuchs-van-de-Graaf inequality, we obtain 
\begin{align*}
\frac{\delta}{2} &< \lim_{n\rightarrow \infty}d\left(\psi_{AB}^{\otimes n}, (\tilde{\cD}_n\circ\cC_n)(\psi_{AB}^{\otimes n})\right)\\
&\leq \lim_{n\rightarrow\infty}\sqrt{1-F\left(\psi_{AB}^{\otimes n}, (\tilde{\cD}_n\circ\cC_n)(\psi_{AB}^{\otimes n})\right)}\\
&= \sqrt{1-\lim_{n\rightarrow\infty}F\left(\psi_{AB}^{\otimes n}, (\tilde{\cD}_n\circ\cC_n)(\psi_{AB}^{\otimes n})\right)}\\
&= \sqrt{1-\lim_{n\rightarrow\infty}\max_{\cD_n}F\left(\psi_{AB}^{\otimes n}, (\cD_n\circ\cC_n)(\psi_{AB}^{\otimes n})\right)},
\end{align*}
which yields
\begin{align*}
\lim_{n\rightarrow\infty}\max_{\cD_n}F\left(\psi_{AB}^{\otimes n}, (\cD_n\circ\cC_n)(\psi_{AB}^{\otimes n})\right) < 1-\frac{\delta^2}{4} < 1,
\end{align*}
since $\delta$ is strictly positive. This proves the claim.
\end{proof}

\subsection{Classical-quantum channels}
The transmission of information through classical-quantum (c-q) channels has been studied by various authors. Wang and Renner~\cite{WR12} obtained bounds on the one-shot capacity in terms of the hypothesis testing relative entropy, whereas in \cite{DMHB13} the bounds were expressed in terms of the smooth max-relative entropy. Both these sets of bounds converged to the Holevo capacity~\cite{Hol98,SW97} in the asymptotic i.i.d.~limit. However, \cite{DMHB13} had the advantage of also yielding the strong converse property of the Holevo capacity. The latter ensures that for transmission rates above the Holevo capacity, information transmission fails with certainty. This was originally proved by Ogawa and Nagaoka \cite{ON99}, and independently by Winter \cite{Win99}.

More recently, in~\cite{TT13}, Tomamichel and Tan derived the second order asymptotics of the capacity of a c-q channel, starting from one-shot bounds expressed in terms of the hypothesis testing relative entropy. Here we obtain
one-shot bounds in terms of an information spectrum relative entropy, and then use its second order asymptotic expansion to recover the result of Tomamichel and Tan~\cite{TT13}. Like \cite{DMHB13} and \cite{TH13}, our result too has the advantage of yielding the strong converse property of the Holevo capacity, in addition to the direct and weak converse parts of the HSW theorem~\cite{Hol98,SW97}. 

Consider a classical-quantum channel $W\colon\cX\to\cD(\cH_B)$, where $\cH_B$ is a finite-dimensional Hilbert space, and $\cX$ denotes the (finite) input alphabet. Suppose that Alice (the sender) wants to communicate with Bob (the receiver) using the channel $W$. To do this, they agree on a finite set of possible messages, labelled by natural numbers from $1$ to $M$. To send the message labelled by $m\in\{1,\ldots,M\}$, Alice has to encode her message into an input signal of the channel, $\varphi(m)\in\cX$, and send it through the channel $W$, resulting in the quantum state $W(\varphi(m))$ at Bob's side. Bob then performs a POVM (positive operator-valued measure) $\Pi\coloneqq \{\Pi_i\}_{i=1}^{M}$, and if the outcome corresponding to 
$\Pi_k$ happens, he concludes that the message with label $k$ was sent. The probability of this event is $\tr\left[W(\varphi(m))\Pi_k\right]$. 
A triple $\cC= (M,\varphi,\Pi)$ defines a code, where:
\begin{itemize}
\item $M\in\mathbb{N}$ is the number of possible messages.
\item $\varphi\colon\,\{1,2,\cdots,M\}\to\cX$ is Alice's encoding of possible messages into input signals of the channel.
\item $\Pi\coloneqq  \{\Pi_m\}_{m=1}^M$ (with $\Pi_m \ge 0$ $\forall\, m =1,2, \ldots M$, and $\sum_{m=1}^M \Pi_m=I$) is a POVM on $\cH_B$, performed by Bob to identify the message (decoding).
\end{itemize}

The \emph{average error probability} $p_e(\cC,W)$ of a code $\cC=(M,\varphi,\Pi)$ is defined as
\begin{align}\label{perr}
p_e(\cC,W)\coloneqq \frac{1}{M}\sum_{i=1}^M \left(1- \tr \left[W(\varphi(i))\Pi_i\right]\right).
\end{align}

\begin{definition}\label{def:cq-one-shot-capacity}
For a given $\eps > 0$, the \emph{one-shot $\eps$-error capacity} $C_\eps^{(1)}(W)$ of a channel
$W$ is defined as
\begin{align*}
C_\eps^{(1)}(W)\coloneqq \sup \{\log M\mid\exists\text{ a code }\cC= (M,\varphi,\Pi)\text{ such that }p_e(\cC, W)\leq \eps\}.
\end{align*}
\end{definition}
Note that it quantifies the maximum number of bits that can be transmitted through a single use of the channel with an average error probability of at most $\eps$.

\subsubsection{One-shot bounds for c-q channels}
Our aim is to give bounds on $C_\eps^{(1)}(W)$ in terms of the information spectrum mutual information. For the input alphabet $\cX$, let $\cP(\cX)$ denote the set of probability distributions on $\cX$. Further, let $\cH_{\cX}$ be a Hilbert space with 
$\dim\cH_{\cX}=|\cX|$ and an orthonormal basis $\{|x\rangle\}_{x\in\cX}$. For any probability distribution $p = \{p(x)\}_{x \in {\cal X}} \in \cP(\cX)$, consider the classical-quantum (c-q) state
\begin{align}\label{rhoxb}
\rho_{XB}(p)\coloneqq \sum_{x\in\cX}p(x)|x\rangle \langle x|\otimes W(x)\quad\text{with }p\in\cP(\cX).
\end{align}
Its reduced states are given by $\rho_B=\sum_{x\in\cX}p(x)W(x)$ and $$\rho_{X}(p)\coloneqq \tr_B\rho_{\cX B}(p)=\sum_{x\in\cX}p(x)|x\rangle \langle x|.$$

\begin{theorem}\label{thm:cq-one-shot}
	Fix $\eps >0$, $\delta, c>0$ and let $\eta > \frac{c}{c+1} \eps$. Then the one-shot $\eps$-error capacity of a c-q channel $W\colon\cX\to\cD(\cH_B)$ satisfies the following bounds:
\begin{align}
 C_\eps^{(1)}(W) &\geq \max_{p \in \cP(\cX)} \uD_s^{\eps - \eta}(\rho_{XB}\|\rho_X\otimes\rho_B) + \log\frac{c}{1+c} + \log \left( \eta - \frac{c}{1+c} \eps\right)\label{eq:cq-one-shot-lower}\\
 C_\eps^{(1)}(W)& \leq \max_{p \in \cP(\cX)}  \min_{\sigma_B}\uD_s^{\eps + \eta} (\rho_{XB}\|\rho_X\otimes\sigma_B)  + \delta - \log \eta,\label{eq:cq-one-shot-upper}
\end{align}
where $\rho_{XB} = \rho_{XB}(p)$ with $p \in \cP(\cX)$ as defined in \eqref{rhoxb}.
\end{theorem}

\begin{proof}[Proof of the lower bound in \Cref{thm:cq-one-shot} (Achievability)]
To establish the lower bound in \eqref{eq:cq-one-shot-lower} of \Cref{thm:cq-one-shot}, we again make use of \Cref{lemHN} by Hayashi and Nagaoka. First, let $p=\{p(x)\}_{x \in \cX} \in \cP(\cX)$ denote the probability distribution for which the maximum in the lower bound of \eqref{eq:cq-one-shot-lower} holds, and let $\rho_{XB} \equiv \rho_{XB}(p)$. For any $\gamma >0$, define the difference operator
\begin{align*} 
\Pi(\gamma) &= \rho_{XB} - 2^\gamma \rho_X \otimes \rho_B \nonumber\\
&= \sum_{x \in \cX} p(x) |x\rangle \langle x| \otimes \left(W(x) - 2^\gamma 
\overline{W}\right),
\end{align*}
where $\overline{W} = \sum_{x \in \cX} p(x) W(x)$.
Then
\begin{align*}
\alpha &\coloneqq  \tr \left[\{ \Pi(\gamma) \ge 0\} \Pi(\gamma) \right]\nonumber\\
&=  \sum_{x \in \cX}p(x) \tr \left[\{ W(x) \ge 2^\gamma \overline{W}\} \left( W(x) - 2^\gamma \overline{W} \right) \right].
\end{align*}
Fix $c>0$ and $\eta>\frac{c}{1+c}\eps$. Choose 
\begin{align*}
	\gamma &= \uD_s^{\eps-\eta}(\rho_{XB}\|\rho_X\otimes\rho_B),
\end{align*}
with $\rho_{XB} \equiv \rho_{XB}(p)$ for the optimizing distribution $p \in \cP(\cX)$ chosen above. For this choice of $\gamma$ we have $\alpha \ge 1 - (\eps - \eta)$. Then \Cref{lemHN} implies the existence of a code $\cC$ of size $M$ such that for any $c>0$,
\begin{align*}
p_e(\cC) & \le (1+c)(1-\alpha) + (2+c+c^{-1}) 2^{-\gamma} M \nonumber\\
& \le  (1+c)(\eps - \eta) + (2+c+c^{-1}) 2^{-\gamma} M.
\end{align*} 
Thus $p_e({\cal C}) \le \eps$ for the choice of $M$ such that
$$(1 +c) (\eps- \eta)  + \frac{(1+c)^2}{c} 2^{-\gamma} M =  \eps,$$
that is, for 
\begin{align}\log M &= 
\gamma + \log\frac{c}{1+c} + \log \left( \eta - \frac{c}{1+c} \eps\right)\nonumber\\
&= \uD_s^{\eps-\eta}(\rho_{XB}\|\rho_X\otimes\rho_B) + \log\frac{c}{1+c} + \log \left( \eta - \frac{c}{1+c} \eps\right).\label{eq:achieve}
\end{align}
Note that the argument of the right-most logarithm yields the condition $\eta>\frac{c}{c+1}\eps$. \Cref{def:cq-one-shot-capacity} of the one-shot $\eps$-error capacity along with \eqref{eq:achieve} implies the lower (achievability) bound \eqref{eq:cq-one-shot-lower}.
\end{proof}

\begin{proof}[Proof of the upper bound in \Cref{thm:cq-one-shot} (Converse)]
To establish the upper bound in \eqref{eq:cq-one-shot-upper} in \Cref{thm:cq-one-shot}, it suffices to prove that for any code ${\cal C}$ of $M$ codewords with 
\begin{align} 
\log M &> \max_{p\in\cP(\cX)}\min_{\sigma_B} \uD_s^{\eps+\eta}(\rho_{XB}\|\rho_X\otimes\sigma_B) + \delta - \log\eta\label{eq:cq-assumption}
\end{align}
we must have $p_e(\cC)>\eps$.

We make the following definitions: Let $\cC=(M, \varphi, \Pi)$ be a code such that \eqref{eq:cq-assumption} holds and $x_m=\varphi(m),\,m=1,\ldots,M$ be the codewords. Let $\rho_m=W(x_m)$ be the output states of the channel and let $p\in\cP(\cX)$ be the uniform distribution on the codewords, i.e., $p(x)=1/M$ if $x=x_m$ for some $m=1,\ldots,M$, and $p(x)=0$ otherwise. Further, let
$\overline{\rho} \coloneqq  \frac{1}{M}\sum_{m=1}^M \rho_m$.
For the chosen $p$, we define 
\begin{align}\label{uni}
{\trho}_{XB}\coloneqq \rho_{XB}(p)=\frac{1}{M}\sum_{m=1}^M |x_m\rangle \langle x_m| \otimes \rho_m.
\end{align}
Let $\sigma_B$ be the state for which 
$$\uD_s^{\eps+\eta}(\trho_{XB}\|\trho_X\otimes\sigma_B) = \min_{\sigma'_B}\uD_s^{\eps+\eta}(\trho_{XB}\|\trho_X\otimes\sigma'_B);$$
then we obtain
\begin{align}
p_e(\cC) &\coloneqq  1 - \frac{1}{M} \sum_{m=1}^M \tr \Pi_m \rho_m\nonumber\\
&=  1 - \frac{1}{M} \sum_{m=1}^M \tr \left[\Pi_m \left(  \rho_m - 2^\gamma 
\sigma_B\right)\right]- \frac{2^\gamma}{M} \sum_{m=1}^M \tr \Pi_m \sigma_B\nonumber\\
&\ge  1 -  \frac{1}{M} \sum_{m=1}^M \tr \left[\{\rho_m \ge 2^\gamma \sigma_B\} \left(\rho_m - 2^\gamma \sigma_B \right)\right]- \frac{2^\gamma}{M},\label{seq}
\end{align}
where we used \Cref{lem:tr-projector} in the last inequality. Choose 
\begin{align}
\gamma &= \max_{p \in \cP(\cX)} \min_{\sigma_B}\uD_s^{\eps + \eta} (\rho_{XB}\|\rho_X\otimes\sigma_B) + \delta\notag\\
&= \max_{p \in \cP(\cX)} \uI_s^{\eps + \eta} (X:B)_\rho + \delta\label{eq:cq-gamma}\\
& \ge \uI_s^{\eps + \eta} (X:B)_{\trho} + \delta\label{eq:cq-gamma-bound}
\end{align}
for some arbitrary $\delta >0$. Here, $\trho$ is the c-q state corresponding to the uniform distribution, defined through \eqref{uni}. 
Hence, 
\begin{align} 
\uI_s^{\eps + \eta} (X:B)_{\trho} = \sup \left\{ \gamma \colon\frac{1}{M} \sum\nolimits_{m=1}^M \tr \left(\rho_m - 2^\gamma \sigma_B \right)_+ \ge 1 - (\eps + \eta)\right\}.
\label{last3}
\end{align} 
Using \eqref{eq:cq-gamma-bound} in \eqref{seq} yields
\begin{align*}
	p_e(\cC) &> \eps+\eta - 2^{\gamma-\log M}\\
	&> \eps + \eta - 2^{\log\eta}\\
	&= \eps,
\end{align*}
where we used \eqref{eq:cq-assumption} and \eqref{eq:cq-gamma} in the second inequality.
\end{proof}

\subsubsection{Second order asymptotics of c-q channels}
In this section we show that our one-shot bounds for the $\eps$-error capacity of a cq-channel from \Cref{thm:cq-one-shot} reproduce its second order asymptotics proved in \cite[Theorem 10]{TT13}. Before we state this result precisely, we make the following definitions:

The capacity $C(W)$ of a c-q channel is defined as $$C(W) = \max_{p\in\cP(\cX)}I(X:B)_\rho$$ with $\rho_{XB}$ as in \eqref{rhoxb}. Furthermore, let $\Pi$ be the set of probability distributions achieving the maximum in $C(W)$, and define 
\begin{align*}
V_\text{\normalfont min} \coloneqq  \min_{p\in\Pi}V(\rho_{XB}\|\rho_X\otimes\rho_B) \qquad\text{and}\qquad V_\text{\normalfont max} \coloneqq  \max_{p\in\Pi}V(\rho_{XB}\|\rho_X\otimes\rho_B).
\end{align*}
Then $V_\eps$ is defined by $$V_\eps = \begin{cases}V_\text{\normalfont min} & \text{if }\eps\in (0,\frac{1}{2})\\ V_\text{\normalfont max} & \text{if }\eps\in [\frac{1}{2},1).\end{cases}$$
\begin{proposition}\label{cq-second-order-asymptotics}
	Let $\eps\in(0,1)$ and assume $V_\text{\normalfont min}>0$. Then the second order expansion of the $\eps$-error capacity $C(n,\eps)$ of a cq-channel $W$ is given by
	$$C(n,\eps) = nC(W) + \sqrt{nV_\eps}\Phi^{-1}(\eps) + o(\sqrt{n}).$$
\end{proposition}
\begin{proof}
Note that the condition $V_\text{\normalfont min}>0$ implies $V_\eps > 0$. We first prove the bound
\begin{align}\label{eq:cq-asymptotics-lower-bound}
	C(n, \eps) \geq nC(W) + \sqrt{nV_\eps}\Phi^{-1}(\eps) + o(\sqrt{n}).
\end{align}
To this end, we write $\rho_{X^nB^n}=\sum_{\mathbf{x}\in\cX^n}p(\mathbf{x})|\mathbf{x}\X \mathbf{x}	|\otimes W^{(n)}(\mathbf{x})$ where $W^{(n)}\colon\cX^n\rightarrow\cH_{B^n}$ and $\mathbf{x}\coloneqq x_1\dots x_n$ with $x_i\in\cX$, and employ the one-shot lower bound from \Cref{thm:cq-one-shot} after $n$ rounds, which reads:
\begin{align*}
	C(n, \eps) &\geq \max_{p\in P(\cX^n)}\uD_s^{\eps-\eta}(\rho_{X^nB^n}\|\rho_{X^n}\otimes\rho_{B^n}) + \log\frac{c}{1+c}+\log\left(\eta-\frac{c}{1+c}\eps\right)\\
	&\geq \max_{p\in P(\cX)}\uD_s^{\eps-\eta}(\rho_{XB}^{\otimes n}\|\rho_X^{\otimes n}\otimes \rho_B^{\otimes n}) + \log\frac{c}{1+c}+\log\left(\eta-\frac{c}{1+c}\eps\right)
\end{align*}
Setting $c=\frac{1}{n}$ and choosing $\eta=\frac{1}{\sqrt{n}}$, we apply \Cref{thm:D_s-asymptotics} to obtain
\begin{align*}
	C(n, \eps) &\geq n\max_{p\in P(\cX)}D(\rho_{XB}\|\rho_X\otimes\rho_B) + \sqrt{n}\,\max_{p\in P(\cX)}s(\rho_{XB}\|\rho_X\otimes\rho_B)\Phi^{-1}(\eps) + \cO(\log n)\\
	&= n C(W) + \sqrt{n\max_{p\in P(\cX)}V(\rho_{XB}\|\rho_X\otimes\rho_B)}\Phi^{-1}(\eps) + \cO(\log n).
\end{align*}
According to the proof of Proposition 11 in \cite{TT13}, we can choose a maximizing distribution $\bar{p}\in P(\cX)$ such that $V(\rho_{XB}\|\rho_X\otimes\rho_B)=V_\eps$. Since $\cO(\log n)\subseteq o(\sqrt{n})$, this yields the lower bound \eqref{eq:cq-asymptotics-lower-bound}.

In order to prove 
\begin{align}\label{eq:cq-asymptotics-upper-bound}
	C(n, \eps) \leq nC(W) + \sqrt{nV_\eps}\Phi^{-1}(\eps) + o(\sqrt{n}),
\end{align}
we first rewrite the one-shot upper bound
	\begin{align}\label{eq:cq-one-shot-original}
		C_\eps^{(1)}(W) \leq \max_{p \in \cP(\cX)}  \uI_s^{\eps + \eta} (X:B)_\rho  + \delta - \log \eta
	\end{align}
	from \Cref{thm:cq-one-shot}. To this end, let $\brho^{\bar{x}}_B$ maximize $\max_{x\in\cX}\uD_s^{\eps+\eta}(\rho_B^x\|\sigma_B)$. It follows that
	\begin{align*}
		1-(\eps+\eta) &\leq \sum_{x\in\cX}p(x)\tr(\rho_B^x-2^\gamma\sigma_B)_+\\
		&\leq \tr(\brho_B^{\bar{x}}-2^\gamma\sigma_B)_+\sum_{x\in\cX}p(x)\\
		&= \tr(\brho_B^{\bar{x}}-2^\gamma\sigma_B)_+,
	\end{align*}
and hence, 
\begin{align}\label{eq:max-p-max-x}
	\max_{x\in\cX}\uD_s^{\eps+\eta}(\rho_B^x\|\sigma_B) \geq \uD_s^{\eps+\eta}(\rho_{XB}\|\rho_X\otimes\sigma_B).
\end{align}
We therefore have
\begin{align}
	\max_{p \in \cP(\cX)}  \uI_s^{\eps + \eta} (X:B)_\rho &= \max_{p \in \cP(\cX)} \min_{\sigma_B} \uD_s^{\eps+\eta}(\rho_{XB}\|\rho_X\otimes\sigma_B)\notag\\
	&\leq \min_{\sigma_B} \max_{p\in \cP(\cX)} \uD_s^{\eps+\eta}(\rho_{XB}\|\rho_X\otimes\sigma_B)\notag\\
	&\leq \min_{\sigma_B} \max_{x\in\cX}\uD_s^{\eps+\eta}(\rho_B^x\|\sigma_B),\label{eq:cq-one-shot-rewritten}
	\end{align}
	where we used \eqref{eq:max-p-max-x} in the last inequality. Inserting \eqref{eq:cq-one-shot-rewritten} into \eqref{eq:cq-one-shot-original} yields
	\begin{align*}
		C(n, \eps) &\leq \min_{\sigma_{B^n}} \max_{\mathbf{x}\in\cX^n}\uD_s^{\eps+\eta}(\rho_B^{\mathbf{x}}\|\sigma_{B^n}) + \delta - \log\eta\\
		 &\leq \min_{\sigma_{B^n}} \max_{\mathbf{x}\in\cX^n}D_H^{\eps+\eta}(\rho_B^{\mathbf{x}}\|\sigma_{B^n}) + \delta - \log\eta.
	\end{align*}
	with the second inequality following from the bound $\uds(\rho\|\sigma)\leq D_H^\eps(\rho\|\sigma)$ in \Cref{thm:D_s-bounds}(i). We finally apply the reasoning in the proof of Proposition 13 in \cite{TT13} to the term $$\min_{\sigma_{B^n}} \max_{\mathbf{x}\in\cX^n}D_H^{\eps+\eta}(\rho_B^{\mathbf{x}}\|\sigma_{B^n})$$ and set $\eta=\frac{1}{\sqrt{n}}$ to arrive at  \eqref{eq:cq-asymptotics-upper-bound}.
\end{proof}

\section{Optimal rates for the case of arbitrary resources}
The information spectrum approach (ISA) allows us to obtain expressions for
optimal rates of information-processing tasks involving arbitrary sources,
channels and entanglement resources.. That is, the assumption for the resources employed in the protocols being memoryless (or i.i.d.) is not imposed. As mentioned earlier, the central quantities in the ISA are the spectral divergence rates $\Du(\hrho\|\homega)$ and $\oD(\hrho\|\homega)$ defined in eqs.\reff{infdv} and \reff{supdv} respectively. Here, $\hrho=\lbrace\rho_n\rbrace_{n\in\mathbb{N}}$ denotes an arbitrary sequence of states and $\homega=\lbrace\omega_n\rbrace_{n\in\mathbb{N}}$ denotes an arbitrary sequence of positive semi-definite
operators, with $\rho_n,\omega_n\in\cB(\cH^{\otimes n})$. Further, the following entropic quantities are derived from them:
\begin{enumerate}[(i)]
\item{The quantum inf- and sup-spectral entropy rates
\begin{align*}
\uS(\hrho) \coloneqq - \oD(\hrho||\hI)\quad\text{and}\quad \oS(\hrho) \coloneqq - \Du(\hrho||\hI),
\end{align*}
where $\hI= \{I_n\}_{n\in\mathbb{N}}$ denotes a sequence of identity operators, with $I_n \in \cB(\cH^{\otimes n})$.}
\item{The inf- and sup-spectral conditional entropy rates for an arbitrary 
sequence of bipartite states $\hrho_{AB} \coloneqq \{\rho^n_{AB}\}_{n\in\mathbb{N}}$:
\begin{align*}
\uS(A|B)_{\hrho} \coloneqq - \oD(\hrho_{AB}||\hat{\one}_A \otimes \hrho_B)\quad\text{and}\quad \oS(A|B)_{\hrho} \coloneqq - \Du(\hrho_{AB}||\hat{\one}_A \otimes \hrho_B),
\end{align*}
where $ \hrho_B= \{\rho_B^n\}_{n\in\mathbb{N}}$, with $\rho_B^n$ denoting the reduced state of the state $\rho_{AB}^n$.}
\item{The inf- and sup-spectral mutual information rates for sequences 
of bipartite states $\hrho_{AB} \coloneqq \{\rho^n_{AB}\}_{n\in\mathbb{N}}$:
\begin{align*}
\uI(A:B)_{\hrho} \coloneqq \Du(\hrho_{AB}||\hrho_A \otimes \hrho_B)\quad\text{and}\quad \oI(A:B)_{\hrho} \coloneqq \oD(\hrho_{AB}||\hrho_A \otimes \hrho_B).
\end{align*}}
\end{enumerate}

Starting from our one-shot bounds for the various tasks studied in this paper, we can directly recover the known expressions for corresponding optimal rates for the case of arbitrary resources, as obtained in the ISA. This is done by employing the relations between the information spectrum relative entropies and the spectral divergence rates proved in \Cref{prop:limit-inf-spectrum}. For the tasks considered in this paper, the rates in the ISA are defined as follows:
\begin{definition}~\label{def:ISA-rates}
\begin{enumerate}[{\normalfont (i)}]
\item For a general quantum information source characterized by an arbitrary sequence of states $\hrho=\lbrace\rho_n\rbrace_{n\in\mathbb{N}}$ acting on a corresponding sequence of Hilbert
spaces $\hcH = \lbrace\cH_n\rbrace_{n\in\mathbb{N}}$, the optimal rate of fixed-length source-coding is given by:
$$R_{sc}(\hrho)\coloneqq \lim_{\eps\rightarrow 0}\liminf_{n\rightarrow\infty}\frac{1}{n}m^{(1),\eps}(\rho_n).$$
\item The noisy dense coding capacity for an arbitrary sequence of
bipartite states $\hrho_{AB} \coloneqq \{\rho^n_{AB}\}_{n\in\mathbb{N}}$ is given by:
$$C_{dc}(\hrho_{AB})\coloneqq \lim_{\eps\rightarrow 0}\limsup_{n\rightarrow\infty}\frac{1}{n}C_{dc}^{(1),\eps}(\rho_{AB}^n).$$
\item The distillable entanglement for an arbitrary sequence of bipartite pure states $\hpsi_{AB} \coloneqq \{\psi^n_{AB}\}_{n\in\mathbb{N}}$ is given by:
$$E_D(\hpsi_{AB})\coloneqq\lim_{\eps\rightarrow 0}\limsup_{n\rightarrow\infty}\frac{1}{n}E_D^{(1),\eps}(\psi_{AB}^n).$$
\item The entanglement cost of an arbitrary sequence of bipartite pure states $\hpsi_{AB} \coloneqq \{\psi^n_{AB}\}_{n\in\mathbb{N}}$ is given by:
$$E_C(\hpsi_{AB})\coloneqq\lim_{\eps\rightarrow 0}\liminf_{n\rightarrow\infty}\frac{1}{n}E_C^{(1),\eps}(\psi_{AB}^n).$$
\end{enumerate}
\end{definition}
Theorems~\ref{thm:source-coding-one-shot}, \ref{thm:dense-coding-one-shot},
\ref{thm_ec}, \ref{thm_entdil} and \ref{thm:cq-one-shot}, lead to the following results, which were originally derived in \cite{NH07,HN03,BD06a,BD08,BD11}):
\begin{proposition}~\label{prop:ISA-recovery}
For the sequences $\hrho$, $\hrho_{AB}$ and $\hpsi_{AB}$ given as in \Cref{def:ISA-rates}, the following results hold:
\begin{enumerate}[{\normalfont (i)}]
\item $R_{sc}(\hrho)=\oS(\hrho)$
\item Let $\hLambda= \{\Lambda_n\}_{n\in\mathbb{N}}$ denote a sequence of CPTP maps $\Lambda_n: \cD(\cH_A^{\otimes n})\to  \cD(\cH_{A}^{\otimes n})$, $ \hsigma = \{ \sigma^n_{AB}\}_{n\in\mathbb{N}}$ a sequence of states with $\sigma^n_{AB}\coloneqq \left(\Lambda_n \otimes \id_{B^n} \right)\rho^n_{AB}$, and $d=\dim\cH_A$. Then,
$$C_{dc}(\hrho_{AB}) = \log d - \min_{\hLambda}  \oS(A|B)_{\hsigma}.$$
\item Let $\hrho_A= \{\rho^n_{A}\}_{n\in\mathbb{N}}$ be a sequence of states with $\rho_A^n = \tr_{B^n} \psi_{AB}^n$, the partial trace being over the Hilbert spaces $\cH_B^{\otimes n}$. Then,
$$E_D(\hpsi_{AB}) = \uS(\hrho_A)\quad\text{and}\quad E_C(\hpsi_{AB}) = \oS(\hrho_A).$$
\end{enumerate}
\end{proposition}

\section*{Conclusion}
\addcontentsline{toc}{section}{Conclusion}
We obtain second order asymptotic expansions for optimal rates of information-processing tasks in the i.i.d.~setting, including fixed-length quantum source coding, noisy dense coding, and pure-state entanglement conversions. To do this, we first obtain one-shot bounds for these protocols in terms of quantities derived from the {\em{information spectrum relative entropies}}. These are two variants of the quantity (of the same name) defined in~\cite{TH13}, which have the particular advantage of satisfying the data-processing inequality. We obtain second order asymptotic expansions for these quantities via the bounds relating them to the 
hypothesis testing relative entropy, and the second order asymptotic expansion of the latter derived in~\cite{TH13}. 

We recover the known second order asymptotics of {\em{classical}} fixed-length source coding (obtained by Hayashi \cite{Hay08}) from our corresponding quantum results. Our results for the entanglement conversions provide a refinement of the inefficiency of these protocols. We also prove that the difference in the second order asymptotics of distillable entanglement and entanglement cost results in irreversibility of entanglement concentration proved in \cite{KH13}.

Furthermore, we recover the known results for the optimal rates for these protocols in the more general setting of the Information Spectrum Approach (ISA) from our one-shot results. This is facilitated by the fact that the spectral divergence rates (which are the central quantities of the ISA) can be readily obtained from our information spectrum relative entropies in the asymptotic limit, when the parameter $0<\eps<1$ (on which they depend) is taken to zero. 

The information spectrum relative entropies can be used for obtaining 
one-shot bounds and second order asymptotics for various other tasks, but we leave the analysis of these for future research. 

\section*{Acknowledgements}
\addcontentsline{toc}{section}{Acknowledgements}
We are grateful to Matthias Christandl, Masahito Hayashi and Richard Jozsa for useful
suggestions. We would also like to thank Ben Derrett, Christoph Harrach, Gowtham Raghunath Kurri and Mark Wilde for helpful discussions and comments.

\emergencystretch=3em

\printbibliography[heading=bibintoc]
\end{document}